\documentclass[11pt,twoside,leqno]{amsart}

\usepackage{amssymb,latexsym,amsmath,amsfonts}

\usepackage{overpic}

\newcommand{\Proof}{\begin{proof}}
\newcommand{\Endproof}{\end{proof}}
\newcommand{\N}{\mathbb N}

\newcommand{\R}{\mathbb{R}}
\newcommand{\C}{\mathbb{C}}
\newcommand{\OO}{\mathcal{O}}
\newcommand{\diag}{\mathop{\mathrm{diag}}}
\newcommand{\dist}{\mathop{\mathrm{dist}}}
\newcommand{\supp}{\mathop{\mathrm{supp}}}

\renewcommand{\Re}{\mathop{\mathrm{Re}}}
\renewcommand{\Im}{\mathop{\mathrm{Im}}}
\renewcommand{\mod}{\mathop{\mathrm{mod}}}
\newcommand{\eps}{\varepsilon}
\DeclareMathOperator{\Tr}{Tr}
\DeclareMathOperator{\Bal}{Bal}
\DeclareMathOperator{\sign}{sgn}

\DeclareMathOperator{\Ai}{Ai}

\newtheorem{theorem}{Theorem}[section]
\newtheorem{lemma}[theorem]{Lemma}
\newtheorem{corollary}[theorem]{Corollary}
\newtheorem{proposition}[theorem]{Proposition}
\newtheorem{Remark}[theorem]{Remark}

\newtheorem{Example}[theorem]{Example}
\newtheorem{Assumption}[theorem]{Assumption}
\newtheorem{definition}[theorem]{Definition}

\newenvironment{example}{\begin{Example}\rm}{\end{Example}}

\numberwithin{equation}{section}
\numberwithin{figure}{section}

\title{Universality in the two matrix model: \\
    a Riemann-Hilbert steepest descent analysis}

\author{Maurice Duits \and Arno B.J. Kuijlaars
}

\begin{document}

\begin{abstract}
The  eigenvalue statistics of a pair $(M_1,M_2)$ of $n\times n$ Hermitian matrices  taken
random with respect to the measure
    $$\frac{1}{Z_n}\exp\big(-n\Tr (V(M_1)+W(M_2)-\tau M_1M_2)\big) \ {\rm d}M_1 {\rm d} M_2 $$
can be described in terms of two families of biorthogonal polynomials.
In this paper we give a  steepest descent analysis of a $4 \times 4$ matrix-valued
Riemann-Hilbert problem characterizing one of the families of biorthogonal polynomials
in the special case $W(y)=y^4/4$ and $V$ an even polynomial. As a result we obtain the
limiting behavior of the correlation kernel associated to the eigenvalues of $M_1$
(when averaged over $M_2$) in the global and local regime as $n\to \infty$ in the one-cut
regular case.
 A special feature in the analysis is the introduction of a vector equilibrium problem
involving  both an external field and an upper constraint.
\end{abstract}

\maketitle

\tableofcontents

\pagestyle{myheadings} \thispagestyle{plain}
\markboth{MAURICE DUITS AND ARNO B.J. KUIJLAARS}{UNIVERSALITY IN THE TWO MATRIX MODEL}

\allowdisplaybreaks

\section{Introduction}

\subsection{Two matrix model}
The two matrix model in random matrix theory is a probability
measure
    \begin{equation}\label{eq:twomatrixmodel}
        \frac{1}{Z_n}{\rm e}^{-n \Tr \big(V(M_1)+W(M_2)-\tau M_1 M_2\big)}
        {\rm d} M_1 {\rm d}M_2
    \end{equation}
defined on pairs $(M_1, M_2)$ of $n \times n$ Hermitian matrices.
Here $Z_n$ is a normalization constant, $V$ and $W$ are two
polynomials of even degree and positive leading coefficients,
$\tau\neq0 $ is a coupling constant and ${\rm d}M_i$ denotes the
product of the Lebesgue measures on the independent entries of
$M_i$ for $i=1,2$.
The model was introduced in \cite{IZ,Meh} as a generalization
of the unitary one matrix model which allows for a larger class
of critical phenomena, see also \cite{DKK,Kaz}
and \cite{DiF,DiFGZJ} for a survey.
For more recent advances in the physics literature, see e.g.\
\cite{BergereEynard,Eynard2,EynardOrantin,EynardOrantin2}
and the references cited therein.

An open problem in random matrix theory is to give rigorous
asymptotic results on the eigenvalue statistics of $M_1$ and $M_2$
as $n\to \infty$. A natural approach is to use the connection with
biorthogonal polynomials. In \cite{MehtaShukla} it is shown that
the eigenvalue statistics can be described in terms of two
families of polynomials $(p_{k,n})_k$ and $(q_{l,n})_l$, where
$p_{k,n}$ and $q_{l,n}$ are monic polynomials of degrees $k$ and
$l$, respectively, satisfying
\begin{align} \label{eq:orthogonality}
 \iint_{\R^2} p_{k,n}(x) q_{l,n}(y) & {\rm e}^{-n \big(V(x)+W(y)-\tau
        x y\big)}\ {\rm d} x\ {\rm d}y=0, \qquad k\neq l.
\end{align}
These polynomials are well-defined and have real and simple zeros
\cite{ErcoMcL}. From a complete asymptotic description of the
polynomials, it is possible to compute the limiting behavior of
the eigenvalue statistics. This has been carried out for the case
\begin{align} \label{eq:fullquadratic}
    V(x)=x^2/2 \qquad \text{ and } \quad  W(y)=ay^2/2
\end{align}
in \cite{ErcoMcL}. Although heuristic calculations can be found in
the physics literature \cite{Eynard}, fully rigorous asymptotic
results for the biorthogonal polynomials for more general $V$ and
$W$ are not known.

The orthogonal polynomial approach to random matrices proved to be
successful in the one matrix models. The orthogonal polynomials
appearing in these models are characterized by a $2 \times 2$
matrix valued Riemann-Hilbert problem \cite{FIK2}. In
\cite{DKMVZuniform,DKMVZstrong} the authors applied the Deift/Zhou
steepest descent method to this Riemann-Hilbert problem and
obtained a complete asymptotic description of the polynomials. As
a result it was possible to give a rigorous  proof of the
universality conjecture for the local eigenvalue correlations.

Inspired by the success of \cite{DKMVZuniform,DKMVZstrong},
several attempts were made to study the asymptotic behavior of the
biorthogonal polynomials by Riemann-Hilbert methods. A number of
Riemann-Hilbert problems for the biorthogonal polynomials were
formulated \cite{BEH2,ErcoMcL,KapaevRH,KuijlaarsMcLaughlin}, but a
successful steepest descent analysis has not been carried out so
far. It is included in Deift's list of major open problems in
random matrix theory and the theory of integrable systems
\cite{DeiftOpen}.

In this paper we present the first complete steepest descent
analysis for a case beyond the fully quadratic case
\eqref{eq:fullquadratic}. We analyze the Riemann-Hilbert problem
for the biorthogonal polynomials $p_{n,n}$ given in
\cite{KuijlaarsMcLaughlin} for the special case
    \begin{align} \label{eq:assumptions}
        W(y)=y^4/4 \quad \textrm{ and } \quad V \textrm{ is an  even polynomial}.
    \end{align}
As a result we are able to compute asymptotics of  the eigenvalue
correlations of the matrix $M_1$, when averaged over $M_2$.

\subsection{Unitary ensembles}

Let us first recall some aspects of the unitary ensembles and
orthogonal polynomials. In the  unitary ensemble one considers
$n\times n$ Hermitian matrices taken randomly with respect to  the
probability measure defined by
    \begin{align} \label{eq:onematrixmodel}
        \frac{1}{Z_n} {\rm e}^{-n\Tr V(M)} {\rm d}M,
    \end{align}
where $V$ is such that
\[ \lim_{x \to \pm \infty} \frac{V(x)}{\log(x^2+1)} = +\infty. \]
Let $p_{k,n}$ be the unique monic polynomial of degree $k$
satisfying
    \begin{align}
        \int_{-\infty}^{\infty} p_{k,n}(x)x^j {\rm e}^{-nV(x)}\ {\rm d}x=0, \qquad j=0,\ldots, k-1.
    \end{align}
Then the eigenvalues of $M$ describe a determinantal point process
on $\R$ with kernel $K_n$ defined by
    \begin{align}
        K_n(x,y)= \gamma_{n-1}^2 {\rm e}^{-n\frac{V(x)+V(y)}{2} }\frac{p_{n,n}(x)p_{n-1,n}(y)-p_{n,n}(y)p_{n-1,n}(x)}{x-y},
    \end{align}
where the constant $\gamma_{n-1}$ is the leading coefficient of
the orthonormal polynomial of degree $n-1$. Thus the joint
probability density $\mathcal P$ on the eigenvalues
$x_1,\ldots,x_n$ is (up to a constant) equal to the determinant
    \begin{align}
    \mathcal P(x_1,\ldots, x_n) = \frac{1}{n!} \det
        \left(K_n(x_i,x_j)\right)_{i,j=1}^n,
        \end{align}
and similarly for the marginal densities for $k=1, \ldots, n-1$,
    \begin{multline}
        \underbrace{\int \cdots \int}_{n-k \textrm{ times}}
        \mathcal P(x_1,\ldots,x_n) \ {\rm d}x_{k+1}\cdots {\rm d}x_n
        = \frac{(n-k)!}{n!} \det
        \left(K_n(x_i,x_j)\right)_{i,j=1}^k.
    \end{multline}
In order to compute the asymptotic behavior of the eigenvalue
statistics, it is sufficient  to obtain asymptotics for the
orthogonal polynomials and the kernel $K_n$.

Fokas, Its, and Kitaev \cite{FIK2} characterized the orthogonal
polynomials $p_{n,n}$ in terms of a Riemann-Hilbert problem (RH
problem). It consists of seeking a $2\times 2$ matrix valued
function $Y$ satisfying
    \begin{align} \label{eq:RHP-for-OPs}
        \left\{
            \begin{array}{ll}
                \multicolumn{2}{l}{Y \textrm{ is analytic in } \C\setminus \R,}\\
                Y_+(x)=Y_-(x)
                    \begin{pmatrix}
                        1 & {\rm e}^{-nV(x)}\\
                        0 & 1
                    \end{pmatrix}, & x\in \R,\\
                Y(z)=\left(I+\OO(z^{-1})\right)
                    \begin{pmatrix}
                        z^n & 0 \\
                        0 & z^{-n}
                    \end{pmatrix},& z\to \infty.
            \end{array}
                \right.
    \end{align}
Here $Y_+$ and  $Y_-$ denote the limiting values of $Y$ on $\R$
when $\R$ is approached from above and below, respectively. The
unique solution of the RH problem is given by
    \begin{align}
    Y=
    \begin{pmatrix}
        p_{n,n} & C \left(p_{n,n} {\rm e}^{-nV}\right)\\
        -2\pi {\rm i}\gamma_{n-1}^2 p_{n-1,n} & - 2\pi {\rm i} \gamma_{n-1}^2
            C\left(p_{n-1,n} {\rm e}^{-nV}\right)
    \end{pmatrix},
    \end{align}
where $C f$ denotes the Cauchy transform
\begin{align} \label{eq:CauchyCf}
            (Cf) (z) = \frac{1}{2\pi {\rm i}} \int_{-\infty}^{\infty}
            \frac{f(x)}{x-z} {\rm d}x, \qquad z\in \C\setminus \R.
    \end{align}
In  \cite{DKMVZstrong} and \cite{DKMVZuniform}, the Deift/Zhou
steepest descent method for RH problems \cite{DZ} was applied to
obtain the asymptotic behavior of $Y$ as $n\to \infty$, and hence
of the orthogonal polynomials.

A key ingredient in the steepest descent analysis in
\cite{DKMVZuniform,DKMVZstrong} is an equilibrium measure, see
also \cite{DeiftBook,ST}. This measure is the unique minimizer of the
energy functional
    \begin{align}
        \iint \log\frac{1}{|x-y|}\ {\rm d}\mu(x){\rm d}\mu(y)+\int V(x) {\rm d}\mu(x)
    \end{align}
among all Borel probability measures $\mu$ on $\R$. If $V$ is real
analytic then the equilibrium measure is supported on a finite
number of intervals, has an analytic density in the interior of
each interval and vanishes at the endpoints \cite{DKM}. It is the
weak limit of the normalized counting measure on the zeros of
$p_{n,n}$ as $n\to \infty$. Moreover, it describes the limiting
mean eigenvalue distribution of a matrix from
\eqref{eq:onematrixmodel}.

The papers \cite{DKMVZuniform,DKMVZstrong} had a major impact on
the theory of random matrices and orthogonal polynomials. Inspired
by these papers, several authors extended the methods of
\cite{DKMVZuniform,DKMVZstrong} to obtain asymptotics for
different types of orthogonal polynomials. For example, for
orthogonal polynomials on the half line \cite{Vanlessen}, on the
interval \cite{KMVV}, and on the unit circle \cite{MMS1,MMS2}.
Another important development is the asymptotic analysis for
discrete orthogonal polynomials \cite{BKMM}. In all these cases
the orthogonal polynomials can be characterized in terms of a
$2\times 2$ matrix-valued RH problem and an associated equilibrium
measure plays an important role.

\subsection{Two matrix models and biorthogonal polynomials}

Let us now return to the two matrix model
\eqref{eq:twomatrixmodel} and the biorthogonal polynomials
$p_{k,n}$ and $q_{l,n}$ in \eqref{eq:orthogonality}. In the
one-matrix model  the eigenvalues of the random matrix follow a
determinantal point process on $\R$ whose kernel is the
reproducing kernel corresponding to the  orthogonal polynomials.
In the two matrix model we have a similar result but the situation
is more complicated.

Define the transformed functions
    \begin{align} \label{eq:transformQj}
        Q_{k,n}(x)&={\rm e}^{-n V(x)}\int_{-\infty}^{\infty} q_{k,n}(y) {\rm e}^{-n\big( W(y)-\tau x y\big)} {\rm d}y, \\
        P_{k,n}(y)&={\rm e}^{-n W(y)} \int_{-\infty}^{\infty} p_{k,n}(x) {\rm e}^{-n \big(V(x)-\tau x y\big)} {\rm
        d}x,
    \end{align}
and let $h_{k,n}^2$ be defined as
    \begin{align}
        h_{k,n}^2=\int_{-\infty}^{\infty} \int_{-\infty}^{\infty}
        p_{k,n}(x) q_{k,n}(y) {\rm e}^{-n\big(V(x)+W(y)-\tau xy \big)}\ {\rm d } x {\rm d}y.
    \end{align}
Consider the following kernels
    \begin{align} \label{eq:KernelK11}
            K^{(n)}_{11}(x_1,x_2)&=\sum_{k=0}^{n-1} \frac{1}{h_{k,n}^2 } p_{k,n}(x_1)Q_{k,n}(x_2),\\
            K^{(n)}_{22}(y_1,y_2)&=\sum_{k=0}^{n-1} \frac{1}{h_{k,n}^2 } P_{k,n}(y_1)q_{k,n}(y_2),\\
            K^{(n)}_{12}(x,y)&=\sum_{k=0}^{n-1} \frac{1}{h_{k,n}^2 } p_{k,n}(x)q_{k,n}(y),\\
            K^{(n)}_{21}(y,x)&=\sum_{k=0}^{n-1} \frac{1}{h_{k,n}^2 } P_{k,n}(y) Q_{k,n}(x)-{\rm e}^{-{n\big(V(x)+W(y)-\tau x
            y\big)}}.
    \end{align}
Then the joint probability density function $\mathcal P$ for the
eigenvalues $x_1,\ldots,x_n$ of $M_1$ and the eigenvalues
$y_1,\ldots,y_n$ of $M_2$ is given by \cite{Eynard,MehtaShukla}
(see also \cite[Chapter 23]{Mehta-book}),
\begin{multline}
    \mathcal    P(x_1,\ldots,x_n,y_1,\ldots,y_n) \\
        =\frac{1}{(n!)^2}\det
                \begin{pmatrix}
                \big(K^{(n)}_{11}(x_i,x_j)\big)_{i,j=1}^n & \big(K^{(n)}_{12}(x_i,y_j)\big)_{i,j=1}^n\\
                \big(K^{(n)}_{21}(y_i,x_j)\big)_{i,j=1}^n & \big(K^{(n)}_{22}(y_i,y_j)\big)_{i,j=1}^n
        \end{pmatrix}.
\end{multline}
Moreover, the marginal densities also have the determinantal form
    \begin{multline}
    \underbrace{\int\cdots \int}_{n-k +n-l \textrm{ times}}
    \mathcal P(x_1,\ldots,x_n,y_1,\ldots,y_n) {\rm d}x_{k+1}\cdots{\rm d} x_n {\rm d} y_{l+1}\cdots {\rm d} y_{n}\\
    =\frac{(n-l)!(n-k)!}{(n!)^2}\det \begin{pmatrix}
    \big(K^{(n)}_{11}(x_i,x_j)\big)_{i,j=1}^{k} & \big(K^{(n)}_{12}(x_i,y_j)\big)_{i,j=1}^{k,l}\\
    \big(K^{(n)}_{21}(y_i,x_j)\big)_{i,j=1}^{l,k} & \big(K^{(n)}_{22}(y_i,y_j)\big)_{i,j=1}^{l}
    \end{pmatrix}.
    \end{multline}
After averaging over the eigenvalues of $M_2$, we see that the
eigenvalues of $M_1$ follow a determinantal point process with
kernel $K^{(n)}_{11}$. Similarly, the eigenvalues of $M_2$ follow
a determinantal point process with kernel $K^{(n)}_{22}$. Both
these determinantal point processes are examples of biorthogonal
ensembles in the sense of \cite{Borodin}.

To determine the asymptotic behavior of the eigenvalues of $M_1$
and $M_2$ it is sufficient to determine asymptotic behavior of the
biorthogonal polynomials and the kernels $K^{(n)}_{ij}$.

The biorthogonal polynomials have been studied for many years and many
interesting properties have been discovered
\cite{AvM,Bertola,BertolaEynard,BEH1,BEH2,BHI,ErcoMcL,Eynard,EM1,KapaevRH,KuijlaarsMcLaughlin}.
Although heuristic results on the asymptotic behavior of the biorthogonal polynomials
can be found in \cite{Eynard},
rigorous asymptotic results have not yet been obtained.

A first step to an asymptotic analysis by means of RH methods is
the formulation of a RH problem for the biorthogonal polynomials
\cite{BEH2,ErcoMcL,KapaevRH,KuijlaarsMcLaughlin}. Here we  follow
\cite{KuijlaarsMcLaughlin}. The relationship with the RH problem
of \cite{BEH2} has been clarified in \cite{BHI}. In fact, this
relationship will be exploited as part of our analysis.

Assume $W$ is a polynomial of degree $d_W$ and define $d_W-1$
weights
\begin{align}
    w_{j,n}(x)={\rm e}^{-n V(x)} \int_{-\infty}^{\infty} y^j
    {\rm e}^{-n(W(y)-\tau xy)} \ {\rm d} y, \qquad j=0,1,\ldots, d_W-2.
\end{align}
Then the RH problem associated with the biorthogonal polynomial
$p_{k,n}$ is the following. We look for a $d_W\times d_W$-matrix
valued function $Y$ satisfying the following properties
\begin{equation}\label{RH-problem-pn-general}
    \left\{\begin{array}{l}\multicolumn{1}{l}{Y  \textrm{ is analytic in } \C\setminus \R,}
    \\
    Y_+(x)=Y_-(x) \begin{pmatrix}   1& w_{0,n}(x)& w_{1,n}(x) & \cdots & w_{d_W-2,n}(x)\\
                                0& 1 & 0 & \cdots &0\\
                                \vdots &  & \ddots    & & \vdots \\
                                \vdots &  & &\ddots&0\\
                                0& \cdots & \cdots&  0 & 1
\end{pmatrix},\ \  x\in \R,\\
Y(z)=(I+\OO(1/z)) \diag\left(z^k, z^{-k_1}, \ldots,
z^{-k_{d_W-1}}\right), \ \  z\to \infty.\end{array}\right.
\end{equation}
where $k_j$ is the integer part of $(k+d_W-1-j)/(d_W-1)$.  In
\cite{KuijlaarsMcLaughlin} it is proved that this RH problem has a
unique solution given by

\begin{equation}\label{eq:solutionRHproblemY}
  Y=\begin{pmatrix}
    p_{k,n} & C(p_{k,n}w_{0,n}) & \cdots  & C(p_{k,n}w_{d_W-2,n}) \\
    p_{k-1,n}^{(0)} & C(p_{k-1,n}^{(0)}w_{0,n}) & \cdots  &C( p_{k,n}^{(0)}w_{d_W-2,n}) \\
    \vdots & \vdots&  &\vdots\\
    p_{k-1,n}^{(d_W-2)} & C(p_{k-1,n}^{(d_W-2)}w_{0,n}) & \cdots&
    C(p_{k-1,n}^{(d_W-2)}w_{d_W-2,n})
  \end{pmatrix}
\end{equation}
where $p_{k,n}$ is the monic biorthogonal polynomial of degree $k$
and $p_{k-1,n}^{(j)}$, $j=0, \ldots, d_W-2$, are certain
polynomials of degree $\leq k-1$. Here $Cf$ denotes the Cauchy
transform as given in \eqref{eq:CauchyCf}.

In the case $d_W=2$ the RH problem  \eqref{RH-problem-pn-general}
reduces to the $2\times 2$ matrix valued RH problem for orthogonal
polynomials. Then the biorthogonal polynomials $p_{k,n}$ are
simply orthogonal polynomials on the real line with respect to a
varying exponential weight, see also \cite{ErcoMcL}.

For $d_W > 2$ the biorthogonal polynomials $p_{k,n}$ do not reduce
to orthogonal polynomials, but instead they are examples of what
is known as multiple orthogonal polynomials. Multiple orthogonal
polynomials for $r$ weights are characterized by $(r+1)\times
(r+1)$ matrix valued RH problems \cite{VAGK}, and
\eqref{RH-problem-pn-general} is an example of such a RH problem.
The steepest descent analysis has not been applied to the RH
problem \eqref{RH-problem-pn-general}.

Certain other systems of multiple orthogonal polynomials and their
associated RH problems were successfully analyzed recently. These
polynomials were related to matrix models with external source
\cite{ABK,BK1,BK2,McL} and non-intersecting paths \cite{DaKu2,DKV,KMW}.
See also \cite{BGS1,BGS2} for a $3 \times 3$ matrix valued RH
problem describing a Cauchy two matrix model. The steepest descent
analysis of the larger size RH problems revealed several new
features that are not present in the analysis of the $2 \times 2$
RH problem \eqref{eq:RHP-for-OPs}.

In our opinion, an important obstacle in the asymptotic
analysis of \eqref{RH-problem-pn-general} is the lack of a
tractable equilibrium problem.  One of the main contribution of the present paper
is  a suitable equilibrium problem that we use in the steepest descent
analysis of \eqref{RH-problem-pn-general} for the special case
\eqref{eq:assumptions}.

\subsection{RH problem for case $W(y) = y^4/4$}

In this paper, we will not treat the general two matrix model with
polynomial potentials, but we restrict ourselves to the special
case \eqref{eq:assumptions}. Due to the assumption $W(y)=y^4/4$ the RH problem
\eqref{RH-problem-pn-general} is of size $4\times 4$.
The assumption that $V$ is an even polynomial introduces a
symmetry with respect to the imaginary axis into the problem. This
will be important for the steepest descent analysis.
We also make
the additional assumption that $n$ is a multiple of three. This is not
essential and is made only for reasons of exposition, since it
simplifies many of the formulas.

Due to these assumptions the RH problem
\eqref{RH-problem-pn-general} characterizing the biorthogonal
polynomial $p_{n,n}$ takes the form
\begin{equation}\label{RH-problem-pn}
\left\{\begin{array}{ll}\multicolumn{2}{l}{Y  \textrm{ is analytic
in } \C\setminus \R,} \\
Y_+(x)=Y_-(x) \begin{pmatrix}   1& w_{0,n}(x)& w_{1,n}(x)& w_{2,n}(x)\\
                                0&1&0&0\\
                                0&0&1&0\\
                                0&0&0&1
\end{pmatrix},& x\in \R,\\
Y(z)=(I+\OO(1/z))\begin{pmatrix} z^n & 0&0&0\\
                                0&z^{-n/3}&0&0\\
                                0&0&z^{-n/3}&0\\
                                0&0&0&z^{-n/3}
\end{pmatrix}, & z\to \infty.\end{array}\right.
\end{equation}
with
\begin{equation} \label{eq:functionswj}
    w_{j,n}(x)={\rm e}^{-n V(x)} \int_{-\infty}^{\infty} y^j {\rm e}^{-n(y^4/4-\tau
    xy)} \ {\rm d}y, \qquad j=0,1,2,
\end{equation}

This RH problem has a unique solution which is given by
\eqref{eq:solutionRHproblemY}. A central observation in
\cite{KuijlaarsMcLaughlin} that leads to this result,  is that the
polynomial $p_{n,n}$ is a multiple orthogonal polynomial with
respect to the three weights $w_{j,n}$ on the real line. This
means that $p_{n,n}$ is the unique monic polynomial of degree $n$
such that
\begin{equation}
    \int_{-\infty}^{\infty} p_{n,n}(x) x^l w_{j,n}(x) \ {\rm d}x=0,
    \qquad l=0,\ldots, \frac{n}{3} -1, \quad j=0,1,2.
\end{equation}

The polynomials $p_{k,n}$ are multiple orthogonal polynomials of
type II. There are also multiple orthogonal polynomials of type I.
These appear in the function $Q_{k,n}$ of \eqref{eq:transformQj},
which can be written as
\begin{align}
    Q_{k,n}(x)=\sum_{j=0}^3 A^{(j)}_{(k,n)}(x) w_{j,n}(x),
\end{align}
for certain polynomials $A^{(j)}_{(k,n)}$, called multiple
orthogonal polynomials of type I. In \cite{DaKu}  a
Christoffel-Darboux formula is proved for the reproducing kernel
associated to multiple orthogonal polynomials in a general
setting. In our case it applies to the kernel $K_{11}^{(n)}$ of
\eqref{eq:KernelK11}. As a result of the Christoffel-Darboux
formula the kernel $K_{11}^{(n)}$ can be expressed in terms of the
solution $Y$ to the RH problem \eqref{RH-problem-pn} as follows.
\begin{proposition}
We have that
\begin{align}\label{eq:kernelintermsofY}
K_{11}^{(n)}(x,y)=\frac{1}{2\pi{\rm i}(x-y)}\begin{pmatrix} 0 &
w_{0,n} (y)& w_{1,n}(y) & w_{2,n}(y)
\end{pmatrix} Y_+^{-1}(y)Y_+(x) \begin{pmatrix}
1\\ 0 \\ 0\\ 0
\end{pmatrix},
\end{align}
for  $x,y\in \R$.
\end{proposition}
\Proof See \cite{DaKu}. \Endproof

\section{Statement of results}

Our results deal with the limiting behavior as $n \to \infty$ of
the eigenvalues of the matrix $M_1$ in the two matrix model
\[  \frac{1}{Z_n}{\rm e}^{-n \Tr \big(V(M_1)+W(M_2)-\tau M_1 M_2\big)} {\rm d} M_1 {\rm d}M_2
    \]
with $\tau > 0$ under the following assumptions
\begin{enumerate}
    \item[(a)] $V$ is an even polynomial,
    \item[(b)] $W(y) = y^4/4$, and
    \item[(c)] $n$ is a multiple of three.
\end{enumerate}

Guionnet \cite{Guionnet} showed  that the limiting mean eigenvalue
density of the matrix $M_1$ exists. Her result is valid in much
greater generality than stated in the next proposition.

\begin{proposition}
The limiting mean density of the eigenvalues of the matrix $M_1$
exists. That is, there exists a probability measure $\mu_1$ on
$\R$  with respect to Lebesgue measure such that in the sense of
weak convergence of measures
    \begin{align} \label{eq:limitingmeandensity0}
    \lim_{n\to \infty}\frac{1}{n} K^{(n)}_{11}(x,x) \ {\rm d}x ={\rm d}
    \mu_1.
    \end{align}
\end{proposition}
\begin{proof}
See \cite{Guionnet}.
\end{proof}

Guionnet characterized $\mu_1$ by a variational problem. Under the
above assumptions (a)-(b), we are going to characterize $\mu_1$ in
terms of an equilibrium problem from potential theory. It is not
clear to us if the two minimization problems are related to each
other.

\subsection{The equilibrium problem}
\label{subsec:EquilibriumProblem}

As already mentioned, we characterize $\mu_1$ by means of an
equilibrium problem from potential theory. Main references for
potential theory in the complex plane are \cite{Ransford,ST}.

We will work with (non-negative) measures on $\C$, but only on
$\R$ and ${\rm i}\R$.  A signed measure is real and can be
negative. The support of a measure $\nu$ is denoted by $S(\nu)$.
The measure may have unbounded support. If a measure $\nu$ has
unbounded support then we assume that
\begin{equation}\label{cond:unbounded-support}
    \int \log(1+|x|) \ {\rm d} \nu(x)<\infty.
\end{equation}
The logarithmic potential of $\nu$ is the function
\begin{equation}
    U^{\nu}(z) = \int \log \frac{1}{|z-x|} {\rm d}\nu(x).
\end{equation}
 The logarithmic energy of $\nu$ is defined as
\begin{equation}
    I(\nu) = \iint \log \frac{1}{|x-y|}  {\rm d} \nu(x) {\rm d} \nu (y)
\end{equation}
and $I(\nu) \in (-\infty, +\infty]$. If $\nu_1$ and $\nu_2$ are
two measures with finite logarithmic energy, then their mutual
logarithmic energy is defined by
\begin{align}
    I(\nu_1,\nu_2)=\iint \log \frac{1}{|x-y|} \ {\rm d} \nu_1(x) {\rm d} \nu_2 (y).
\end{align}
For a given $V$ and $\tau > 0$, we study the energy functional
$E_V$ defined by
\begin{multline}
        E_V(\nu_1,\nu_2,\nu_3) \\
         = \sum_{j=1}^3 I(\nu_j) - \sum_{j=1}^2 I(\nu_j,\nu_{j+1}) \label{energyfunctional}
          +\int\left( V(x)-\frac{3}{4}\tau^{4/3} |x|^{4/3}\right)  {\rm d}\nu_1(x)
\end{multline}
where $\nu_1$, $\nu_2$, $\nu_3$ are three measures with finite
logarithmic energy. The equilibrium problem is the following.

\begin{definition}\label{def:equilibriummeasure}
The equilibrium problem is to minimize $E_V(\nu_1, \nu_2, \nu_3)$
among all measures $\nu_1$, $\nu_2$, $\nu_3$ such that
\begin{itemize}
\item[(a)] the measures have finite logarithmic energy;

\item[(b)] $\nu_1$ is a measure on $\R$ with $\nu_1(\R) = 1$;

\item[(c)] $\nu_2$ is a measure on ${\rm i}\R$ with $\nu_2({\rm i} \R) = 2/3$;

\item[(d)] $\nu_3$ is a measure on $\R$ with $\nu_3(\R) = 1/3$;

\item[(e)] $\nu_2$ satisfies the upper constraint $\nu_2 \leq
\sigma$ where $\sigma$ is the (unbounded) measure on ${\rm i}\R$
defined by
\begin{align}\label{eq:densitysigma}
{\rm d}\sigma(z)=\frac{\sqrt{3}}{2\pi} \tau^{4/3} |z|^{1/3} | {\rm d}z|,
\qquad z \in {\rm i}\R,
\end{align}
where $|{\rm d}z|$ denotes the arclength on ${\rm i}\R$.
\end{itemize}
\end{definition}

An electrostatic interpretation of the equilibrium problem for the
energy functional \eqref{energyfunctional} is the following.
Consider three types of charged particles. The particles of the
first type are put on $\R$ and have total charge $1$. The
particles of the second type are put on ${\rm i}\R$ and have total
charge $2/3$. The particles of the third type are put on $\R$ and
have total charge $1/3$. Particles of the same type repel each
other. The particles of the first and the second types attract
each other with a strength that is half the strength of the
repulsion of particles of the same type. So do the particles of
the second and the third types. Particles of the first and the
third types do not interact directly. Particles of the first type
are influenced by an external field depending on $V$ and $\tau$.

In the equilibrium problem the particles distribute themselves in
order to minimize their energy under the extra condition that the
particle density of the second type particles does not exceed the
density of $\sigma$, where $\sigma$ is the given measure
\eqref{eq:densitysigma}. Thus $\sigma$ acts as an upper constraint
on the second measure.

Equilibrium problems for a vector of measures with mutual
interaction as in \eqref{energyfunctional} arise for
Nikishin systems \cite{NS} in the theory of rational approximation,
see also the survey \cite{Apt},
and the more recent papers \cite{AKV,BB,DuKu,KMW,LL,VAGK}.

Equilibrium problems with constraint appeared before in asymptotic
results for discrete orthogonal polynomials
\cite{BKMM,Beckermann,Dragnev-Saff,KR,Rakhmanov}, singular limits
of integrable systems \cite{DM}, and convergence results for
Krylov methods in numerical linear algebra \cite{KuijlaarsSIAM}.

We prove the following.
\begin{theorem} \label{prop:structureEqmeasure}
Let $V$ be an even polynomial and $\tau > 0$. Then there is a
unique minimizer $(\mu_1,\mu_2,\mu_3)$  of $E_V$ subject to the
conditions {\rm (a)}--{\rm (e)} in the equilibrium problem.
Moreover,
\begin{enumerate}
\item[(a)] The measure $\mu_1$ is supported on a finite number of
    disjoint intervals $\bigcup_{j=1}^N [a_j, b_j]$ with a density of
    the form
    \begin{equation}\label{eq:measuremu1psi}
        \frac{{\rm d}\mu_1(x)}{{\rm d}x}= h_j(x)\sqrt{(b_j-x)(x-a_j)},
        \qquad x \in [a_j,b_j],
    \end{equation}
    where $h_j$ is real analytic and non-negative on $[a_j,b_j]$, for
    $j=1, \ldots, N$.

\item[(b)] $S(\mu_2)={\rm i} \R$ and there exists a constant $c>0$
such that
\[ S(\sigma-\mu_2) = {\rm i} \R \setminus (-{\rm i}c,{\rm i}c). \]
The measure $\sigma-\mu_2$ has an analytic density on ${\rm i
\R}\setminus (-{\rm i}c,{\rm i}c)$ which vanishes as a square root
at $\pm {\rm i}c$. Moreover, the logarithmic potential $U^{\mu_2}$
is such that
    \begin{align} \label{eq:analyticVeff}
        \frac{3}{4} \tau^{4/3} |x|^{4/3} + U^{\mu_2}(x)
    \qquad \textrm{ is real analytic on } \R.
    \end{align}

\item[(c)] $S(\mu_3)=\R$ and $\mu_3$ has a density which is
analytic in $\R \setminus\{0\}$.

\item[(d)] all three measures $\mu_1$, $\mu_2$ and $\mu_3$ are
symmetric in the sense that $\mu_j(-A) = \mu_j(A)$ for $j=1,2,3$,
and for every Borel set $A$.
\end{enumerate}
\end{theorem}
The symmetry property in part (d) of Theorem
\ref{prop:structureEqmeasure} is a direct consequence of the
uniqueness of the minimizer and the fact that $V$ is even. We
state it explicitly here, since it will be used many times in what
follows, often without explicit mentioning it.

For given $\nu_2$, $\nu_3$, we can minimize $E_V(\nu_1, \nu_2,
\nu_3)$ with respect to $\nu_1$ only. Then we look for the
minimizer of the energy functional
\begin{align} \label{eq:energyfunctionalmu1}
    I(\nu_1)+\int\left( V(x)-\frac{3}{4}\tau^{4/3} |x|^{4/3}-U^{\nu_2}(x) \right) {\rm d}\nu_1(x),
\end{align}
among all probability measures $\nu_1$ on $\R$.  This is a usual
equilibrium problem with external field \cite{DeiftBook,ST}. The
external field is analytic in $\R\setminus \{0\}$, but possibly
not at $0$. However, due to \eqref{eq:analyticVeff} the external
field is analytic at $0$ if $\nu_2 = \mu_2$ is part of the
minimizer $(\mu_1, \mu_2, \mu_3)$ of the full equilibrium problem.
Thus $\mu_1$ is the minimizer of an energy functional with
analytic external field on $\R$. The statements in Theorem
\ref{prop:structureEqmeasure} about the structure of the
measure $\mu_1$ then follow from results of \cite{DKM}.

If we minimize only with respect to $\nu_3$, with $\nu_1$ and
$\nu_2$ fixed, then we minimize the energy functional
\begin{align}
I(\nu_3)-\int U^{\nu_2}(x) \ {\rm d}\nu_3(x),
\end{align}
among all $\nu_3$ on $\R$ with total mass $1/3$.  Again this is an
equilibrium problem with external field, but the external field
$-U^{\nu_2}(x)$ is only slowly growing at infinity. In
\cite{Simeonov} such an external field is called weakly
admissible. It is a consequence of the slow growth that the
support $S(\nu_3)$ of the minimizer will be unbounded, and in fact
it will be the full real line.

Finally, if we fix $\nu_1$ and $\nu_3$, and minimize with respect
to $\nu_2$, then we look for the minimizer of the energy
functional
\begin{align}
I(\nu_2)-\int\left( U^{\nu_1}(x)+U^{\nu_3}(x) \right) {\rm
d}\nu_2(x),
\end{align}
among all $\nu_2$ on ${\rm i}\R$ with total mass $2/3$ satisfying
the constraint $\nu_2\leq \sigma$. Here we again have a weakly
admissible external field, but in addition there is the upper
constraint. Again, we will have that the support of the minimizer
is unbounded, and indeed it is the full imaginary axis. The
constraint $\sigma$ is active on a symmetric interval $[-{\rm i}c,
{\rm i}c]$ around the origin.

\subsection{Variational conditions and regular/singular behavior}
From the discussion above we see that each of $\mu_1$, $\mu_2$,
$\mu_3$ is the minimizer for an equilibrium problem with external
field and/or upper constraint, and as such they are characterized
by the following set of variational conditions.
\begin{proposition} \label{prop:variationalconditions}
 The measures $\mu_1$, $\mu_2$ and $\mu_3$ satisfy
\begin{align}
      2U^{\mu_1}(x)&=U^{\mu_{2}}(x)-V(x)+\frac{3}{4}\tau^{4/3} |x|^{4/3}+\ell,
        \quad x\in S(\mu_1), \label{eq:variationalcondmu1a}\\
      2U^{\mu_1}(x)&\geq U^{\mu_{2}}(x)-V(x)+\frac{3}{4}\tau^{4/3} |x|^{4/3}+\ell,
        \quad x\in \R\setminus S(\mu_1), \label{eq:variationalcondmu1b}
\end{align}
for some constant $\ell$,
\begin{align}
      2U^{\mu_2}(x)&=U^{\mu_1}(x)+U^{\mu_3}(x),
        \quad x\in S(\sigma-\mu_2), \label{eq:variationalcondmu2a} \\
      2U^{\mu_2}(x)&<U^{\mu_1}(x)+U^{\mu_3}(x),
        \quad x\in {\rm i}\R\setminus S(\sigma-\mu_2),
        \label{eq:variationalcondmu2b}
\end{align}
and
\begin{align}
      2U^{\mu_3}(x)&=U^{\mu_{2}}(x),
        \qquad x\in S(\mu_3) = \R. \label{eq:variationalcondmu3}
\end{align}
\end{proposition}

Note that the inequality in \eqref{eq:variationalcondmu2b} is
strict. By contrast the inequality in
\eqref{eq:variationalcondmu1b} need not be strict. Then as in the
case of the one-matrix model and orthogonal polynomials
\cite{DKMVZuniform} we can distinguish between regular and
singular cases as follows.

\begin{definition} \label{def:regular}
The measure $\mu_1$ is called regular (otherwise singular) if the
functions $h_j$ in \eqref{eq:measuremu1psi} satisfy $h_j>0$ on
$[a_j,b_j]$ for every $j=1, \ldots, N$, and if the inequality in
\eqref{eq:variationalcondmu1b} is strict.

A point $x^* \in \R$ at which one of these conditions fails is a
singular point. There are three types of singular points, namely
\begin{itemize}
    \item $x^*$ is a singular end point if $x^* \in \{a_j,b_j\}$
    for some $j=1, \ldots, N$ and $h_j(x^*) = 0$,
    \item $x^*$ is a singular interior point if $x^* \in
    (a_j,b_j)$ for some $j=1, \ldots, N$ and $h_j(x^*) = 0$,
    \item $x^*$ is a singular exterior point if $x^*\in \R\setminus S({\mu_1})$ and equality holds
    for $x = x^*$ in \eqref{eq:variationalcondmu1b}.
\end{itemize}
\end{definition}

We see that the structure for the measure $\mu_1$ in the global
regime is the same as what is known for the one-matrix model in a
real analytic external field \cite{DKMVZuniform}. It is known that
the regular case holds generically in one-matrix models \cite{KM1}
and all three types of singular behavior can occur. The same is
true in the present two matrix model, but the maybe surprising
fact is that no other type of singular behavior can occur. We
emphasize that this is related to our assumption that $W(y) =
y^4/4$.

We conclude with an easy to check convexity result, which is
analogous to the well-known fact that for a convex external field
the equilibrium measure is supported on one interval. The
following is a similar result for the equilibrium problem of
Definition \ref{def:equilibriummeasure}.

\begin{proposition} \label{Th:convexity}
Suppose that $V$ is even and that  $x\mapsto V(\sqrt x )$ is
convex for $x > 0$. Then $S(\mu_1)$ consists of one or two
intervals.
\end{proposition}

The proposition applies to all even quartic potentials
$V(x) = a x^4 + b x^2$ with $a > 0$.

\subsection{One-cut regular case}

The core of the present paper is  a Deift/Zhou steepest descent analysis for
the RH problem for the biorthogonal polynomials. As a particular result we are
able to obtain the limiting behavior of the kernel $K_{11}^{(n)}$ as $n\to \infty$.
The steepest descent analysis simplifies  in the one-cut regular case, that is, the
measure $\mu_1$ is regular and supported on one interval. We will only deal with this
case in this paper, but our methods can be  extended to the non-regular and multi-cut
situations.

The non-regular cases for unitary ensembles have
been treated recently, see  \cite{Bleher-Its2,Claeys-Kuijlaars,Claeys-Kuijlaars-Vanlessen}
for singular interior points, \cite{Claeys-Vanlessen} for singular
endpoints and \cite{Bertola-Lee,Claeys,Mo} for singular exterior
points.  All three singular cases can appear in our situation, and
can be implemented within our methods to obtain the limiting
behavior of the polynomials $p_{n,n}$ and the kernel $K_{11}^{(n)}$ for
the singular cases as well. We will
not go into details here.

As part of our analysis we will express the minimizer $(\mu_1,\mu_2,\mu_3)$ in
terms of a meromorphic functions defined on a four sheeted Riemann surface. The
genus of the Riemann surface equals $N-1$ where $N$ is the number of intervals
in the support of $\mu_1$. Hence the genus is non-zero in the multi-cut case,
which complicates the construction of the outside parametrix in the steepest
descent analysis. For reasons of exposition we only deal here with the one-cut
case and plan to return to the multi-cut case in future work.

\subsection{Global eigenvalue regime}

Our first main theorem states that the measure $\mu_1$ that we obtain as part of the
minimizer $(\mu_1, \mu_2, \mu_3)$ for the equilibrium problem for
$E_V$ is the limiting mean  distribution for the eigenvalues of
$M_1$. We prove this in this paper for the one-cut regular case.

\begin{theorem} \label{th:eigenvaluedistribution}
Let $W(y) = y^4/4$, let $V$ be an even polynomial and $\tau > 0$.
Let $(\mu_1, \mu_2, \mu_3)$ be the minimizer of the equilibrium
problem described in Subsection
    \ref{subsec:EquilibriumProblem} above. Suppose that the measure
$\mu_1$ is one-cut regular.

Then we have as $n \to \infty$ with $n \equiv 0 (\mod 3)$
    \begin{align} \label{eq:limitingmeandensity}
    \lim_{n\to \infty}\frac{1}{n} K^{(n)}_{11}(x,x) =
    \frac{{\rm d}\mu_1}{{\rm d}x}(x),
    \end{align}
    uniformly for $x \in \R$.
\end{theorem}

% The assumption that the measure $\mu_1$ is one-cut regular can be
% removed in Theorem \ref{th:eigenvaluedistribution}.

\subsection{Local eigenvalue regime}

In the local eigenvalue regime we obtain the same universal
limiting behaviors that are known from one-matrix models.
This is in agreement with the universality conjecture in random
matrix theory, which says that local eigenvalue correlations  in
random matrix models with unitary symmetry do not depend on the
particular features of the model, but only on the global regime.

We again restrict ourselves to the one-cut regular case. So we
assume that the measure $\mu_1$ from the solution of the
equilibrium problem is supported on one interval and we write
$S(\mu_1)=[-a,a]$.

\begin{theorem} \label{th:universality}
Let $W(y)=y^4/4$, let $V$ be an even polynomial and $\tau > 0$.
Assume that the measure $\mu_1$ is one-cut regular.
\begin{enumerate}
\item[{\rm (a)}] Let $x^*\in (-a,a)$ and define
 \begin{align}
 \rho := \frac{{\rm d} \mu_1}{{\rm d} x} (x^*) > 0.
 \end{align}
 Then, as $n \to \infty$ with $n\equiv 0 (\mod 3)$ we have for every $k \in \N$
 that
 \begin{multline}\label{eq:universalitysine}
        \lim_{n\to\infty} \det\left(\frac{1}{\rho n} K^{(n)}_{11}
        \left(x^*+\frac{u_i}{\rho n}, x^*+\frac{u_j}{\rho n}\right)\right)_{i,j=1}^k \\
        =
        \det\left( \frac{\sin\pi(u_i-u_j))}{\pi(u_i-u_j)}\right)_{i,j=1}^k.
\end{multline}
uniformly for $(u_1,\ldots,u_k)$ in compact subsets of $\R^k$.

\item[{\rm (b)}] Let $\rho>0$ be such that
\begin{align*}
        \frac{{\rm d}\mu_1(x)}{{\rm d}x}=\frac{\rho}{\pi}(a-x)^{1/2}\left(1+\OO(a-x)\right),
        \qquad \textrm{as } x \nearrow a.
        \end{align*}
         Then, as $n\to \infty$ with $n\equiv 0 (\mod 3)$ we have
         for every $k \in \N$ that
\begin{multline}\label{eq:universalityairy}
    \lim_{n\to \infty} \det\left(\frac{1}{ (\rho n)^{2/3}} K^{(n)}_{11}\left(a+\frac{u_i}{(\rho n)^{2/3}}, a+\frac{u_j}{(\rho n)^{2/3}}\right)\right)_{i,j=1}^k \\
    =\det\left(\frac{\Ai(u_i)\Ai'(u_j)-\Ai'(u_i)\Ai(u_j)}{u_i-u_j}\right)_{i,j=1}^k,
        \end{multline}
      uniformly for $(u_1,\ldots,u_k)$ in compact subsets of $\R^k$. Here $\Ai$ denotes the usual Airy function.
\end{enumerate}
\end{theorem}

In Theorem \ref{th:universality} we only deal with the one-cut
regular case. The non-regular cases that can occur in our situation,
also appear in the unitary ensembles. In these cases,
we also obtain the same limits for  the  kernel $K_{11}^{(n)}$
as in the unitary ensembles. For example, if the density of the
measure $\mu_1$ vanishes quadratically at the an interior point
of the support, then the limiting behavior of the kernel $K_{11}^{(n)}$,
in a double scaling limit, is expressed in terms of the $\Psi$-functions
associated to the Hastings-McLeod solution of the  Painlev\'e II equation,
see also \cite{Bleher-Its2,Claeys-Kuijlaars}.

\subsection{Overview of the rest of the paper}

The proof of Theorem \ref{th:universality}  follows from a
Deift/Zhou steepest descent analysis applied to the RH problem for
biorthogonal polynomials. The Deift/Zhou steeepest descent
analysis consists of a sequence of explicit and  invertible
transformations
\begin{align}
   Y \mapsto X \mapsto U \mapsto T \mapsto S \mapsto R.
\end{align}
The RH problem for $Y$ is not in a suitable form for an immediate
application of the equilibrium measures from the equilibrium
problem. In the first transformation $Y \mapsto X$ we transform
the RH problem for $Y$ to a RH problem for $X$ that is more
suitable for further asymptotic analysis. This transformation
depends on a method that we learned from the authors of
\cite{BHI}. The construction involves Pearcey integrals and is
presented in Section 3.

In the second transformation $X\mapsto U$ we will use the
equilibrium measures $\mu_1$, $\mu_2$ and $\mu_3$ from the
equilibrium problem of Definition \ref{def:equilibriummeasure}
 and their
corresponding $g$-functions. In Section 4 we discuss the
equilibrium problem and we prove Theorem
\ref{prop:structureEqmeasure} and Propositions
\ref{prop:variationalconditions} and \ref{Th:convexity}.

The properties of the $g$-functions are conveniently expressed in
terms of functions that come from a four sheeted Riemann surface,
which in the one-cut case has genus zero. We introduce the Riemann
surface in Section 5. The second transformation of the RH problem
$X \mapsto U$ is given in Section 6.

In the transformation $U \mapsto T\mapsto S$ in Section 7 we open lenses.
In the transformation $U \mapsto T$ we open the lens around the supports of $\mu_2$ and $\mu_3$. These supports are
unbounded so special care has to be taken at infinity. In the transformation $T \mapsto S$ we open the lenses
around the support of $\mu_1$.

The next step is the construction of a parametrices in
Section 8. Here the Riemann surface is used again. We give a
rational parametrization which allows us to give explicit formulas
for the outside parametrix. In the multi-cut case the Riemann
surface has higher genus and the construction of the outside
parametrix is more complicated. We construct local Airy parametrices around the
branch points of the Riemann surface. These are the endpoints $\pm
a$ of the support of $\mu_1$ and the endpoints $\pm {\rm i}c$ of
the support of $\sigma-\mu_2$. The final transformation $S \mapsto R$ is also given in Section 8. It
leads to a RH problem for $R$ with jump matrices that are
uniformly close to the identity matrix. The RH problem for $R$ is
also  normalized at infinity.

The proofs of Theorems \ref{th:eigenvaluedistribution} and
\ref{th:universality} are in the final Section 9.

\section{The first transformation $Y\mapsto X$}

The starting point is the RH problem \eqref{RH-problem-pn} with
functions \eqref{eq:functionswj}.

In this section we introduce the first transformation $Y\mapsto X$
which makes use of the special form of the weight functions
$w_{j,n}$.

\subsection{The main idea}
The main idea behind the first transformation $Y\mapsto X$  can be
found in an unpublished manuscript of Bertola, Harnad and Its
\cite{BHI}.  The starting point is the observation that
$w_{0,n}(x){\rm e}^{nV(x)}$ satisfies a scaled version of the
so-called Pearcey differential equation
\begin{align}\label{eq:PearceyDE}
p'''(z)=zp(z).
\end{align}
Special solutions to the equation \eqref{eq:PearceyDE} are given
by the Pearcey integrals
   \begin{equation} \label{Pearcey-integrals}
      p_j (z) = \int_{\Gamma_j} {\rm e}^{-s^4/4+ sz} \ {\rm d} s, \qquad
      j=0,\ldots,5,
   \end{equation}
where the contours $\Gamma_j$ are
   \begin{equation}
        \begin{array}{ll}
            \Gamma_0= (-\infty, \infty), & \Gamma_1=({\rm i} \infty,0]\cup [0,\infty), \\
            \Gamma_2= ({\rm i}\infty,0]\cup [0,-\infty), &\Gamma_3= (-{\rm i}\infty,0]\cup [0,-\infty),\\
            \Gamma_4=(-{\rm i}\infty,0]\cup [0,\infty), &\Gamma_5=(-{\rm i}\infty, {\rm i} \infty),
        \end{array}
   \end{equation}
or homotopic deformations such as the ones shown Figure
\ref{fig:PearceyC}. Each $\Gamma_j$ is equipped with an
orientation as shown in Figure \ref{fig:PearceyC}.
    \begin{figure}
        \centering
          \includegraphics[angle=270,scale=0.4]{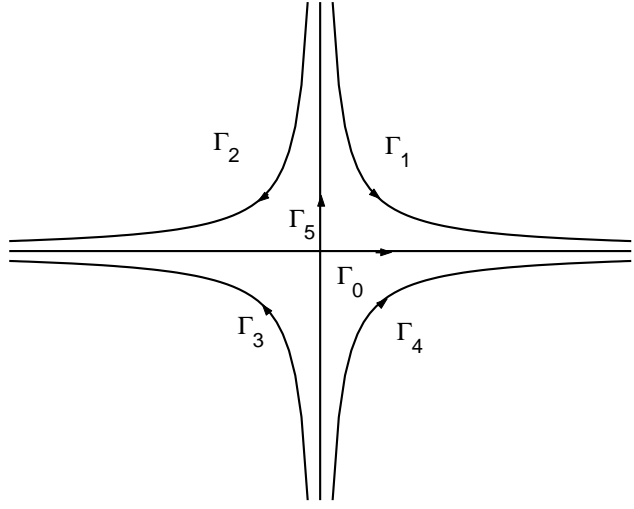}
         \caption{Contours $\Gamma_j$ in the definition of the Pearcey intergals.}
         \label{fig:PearceyC}
    \end{figure}
Since the  Pearcey equation \eqref{eq:PearceyDE} is a linear third
order equation, there must be a linear relation between any four
solutions. Indeed, from the integral representations
\eqref{Pearcey-integrals} one can find for example,
    \begin{align}\label{eq:p5p4p1}
        p_5(z)=p_4(z)-p_1(z).
    \end{align}

The functions $w_{j,n}$ from \eqref{eq:functionswj} are expressed
in terms of the Pearcey integral $p_0$ and its derivatives as
follows
    \begin{align} \label{eq:w0intermsofp5}
        w_{0,n}(z) &= n^{-1/4}{\rm e}^{-n V(z)}p_0(n^{3/4} \tau z),\\
        w_{1,n}(z) &= n^{-1/2}{\rm e}^{-n V(z)}p_0'(n^{3/4} \tau z),\\
        w_{2,n}(z) &= n^{-3/4}{\rm e}^{-n V(z)}p_0''(n^{3/4} \tau z)\label{eq:w2intermsofp5}.
\end{align}

Now we can illustrate the main idea behind the transformation
$Y\mapsto X$. Consider the matrix-valued function $\widehat X$
defined by
\begin{equation} \label{eq:widehatX}
\widehat X(z)= Y(z)
\begin{pmatrix}
1 & 0 \\
0 & D_n \widehat{P}^{-t}(n^{4/3} \tau z)
\end{pmatrix},
\end{equation}
where
\begin{align} \label{eq:defDn}
    D_n = \begin{pmatrix} n^{1/4} & 0 & 0\\
 0 & n^{1/2}  & 0\\
 0 & 0 &n^{3/4}\end{pmatrix}, \qquad
\widehat{P}=\begin{pmatrix}
p_0 & p_j & p_k\\
p_0' & p_j' & p_k'\\
p_0'' & p_j''& p_k''
\end{pmatrix},
\end{align}
and $j$ and $k$ are such that $p_0$, $p_j$, and $p_k$  are
linearly independent solutions of the Pearcey equation, so that
$\widehat{P}$ is indeed invertible. To compute the jumps for
$\widehat{X}$ we note that
\eqref{eq:w0intermsofp5}--\eqref{eq:w2intermsofp5} can be written
as
\begin{align}\label{eq:mainideatrickBHIpre}
D_n^{-1}\widehat{P}(n^{4/3} \tau z)  \begin{pmatrix}1\\ 0 \\
0\end{pmatrix}={\rm e}^{nV(z)} \begin{pmatrix} w_{0,n}(z)\\
w_{1,n}(z)\\
w_{2,n}(z)
\end{pmatrix},
\end{align}
so that for $x \in \R$, by \eqref{RH-problem-pn} and
\eqref{eq:widehatX},
\begin{align} \label{eq:mainideatrickBHI}
\widehat X_-(x)^{-1} \widehat X_+(x) &=
\begin{pmatrix}
1 & \begin{pmatrix} w_{0,n}(x) & w_{1,n}(x) & w_{2,n}(x)
\end{pmatrix} D_n \widehat{P}^{-t}(n^{4/3} \tau x)\\
0& I_3
\end{pmatrix}\\
\nonumber &=\begin{pmatrix}
1 & {\rm e}^{-n V(x)} & 0 & 0 \\
0 & 1 & 0  & 0 \\
0 & 0 & 1 & 0\\
0 & 0 & 0 & 1
\end{pmatrix}.
\end{align}
The jump matrix has only one non-trivial upper triangular $2\times
2$ block. This $2\times 2$ block has a varying exponential in its
$12$ entry and so is precisely of the form that appears in the RH
problem for orthogonal polynomials on the real line with respect
to  varying exponential weight \cite{DeiftBook,DKMVZuniform,FIK2}.
So the jump matrix for $\widehat X$ has a form that seems very
promising for asymptotic analysis.

The complication however, is that the asymptotic behavior of
$\widehat{X}$ becomes rather involved due to the nontrivial
asymptotic behavior of the Pearcey integrals $p_j(z)$ as $z \to
\infty$. We also have to deal with the Stokes phenomenon which
implies that the asymptotic behavior may be different as $z \to
\infty$ in different sectors in the complex plane. For example, by
classical steepest descent methods one can prove as in \cite{BK2}
that
    \begin{align} \label{eq:p0asymptotics-RHplane}
        p_0(z) = \sqrt{\frac{2\pi}{3}} z^{-1/3} {\rm e}^{\frac{3}{4} z^{4/3}}(1+\OO(z^{-4/3}))
    \end{align}
as $z\to \infty$ with $\Re z > 0$, and
    \begin{align} \label{eq:p0asymptotics-LHplane}
        p_0(z)&=\sqrt{\frac{2\pi}{3}} (-z)^{-1/3} {\rm e}^{\frac{3}{4} (-z)^{4/3}}(1+\OO(z^{-4/3}))
    \end{align}
as $z\to \infty$ with $\Re z < 0$. Here and throughout the paper
we use principal branches of the fractional powers, that is the
branch cut is along the negative real axis. So $p_0$ has a
different asymptotic behavior as $z \to \infty$ in the right
half-plane and in the left half-plane. Also the asymptotic
formulas
\eqref{eq:p0asymptotics-RHplane}--\eqref{eq:p0asymptotics-LHplane}
are only uniformly valid if we stay away from the imaginary axis.

To deal with the different asymptotic behaviors we define $P$
differently in each of the quadrants and define the transformation
$Y\mapsto X$ accordingly. This procedure will introduce a new
constant jump on the imaginary axis, while still simplifying the
jump matrix on the real axis. Also the asymptotic condition in the
RH problem for $X$ takes a nice form.

\subsection{A RH problem for Pearcey integrals}

We denote the four quadrants in the complex plane by $I$, $II$,
$III$ and $IV$, respectively. We define $P$ in the different
quadrants by
\begin{align} \label{eq:PinIandII}
P=\begin{pmatrix}
    p_0   &   -p_2   &   -p_5  \\
    p_0'  &   -p_2'  &   -p_5' \\
    p_0'' &   -p_2'' &   -p_5''
\end{pmatrix} \textrm{ in } I, \qquad &
P=\begin{pmatrix}
    p_0   &   -p_1   &   -p_5  \\
    p_0'  &   -p_1'  &   -p_5' \\
    p_0'' &   -p_1'' &   -p_5''
\end{pmatrix} \textrm{ in } II,\\
\label{eq:PinIIIandIV} P=\begin{pmatrix}
    p_0   &   -p_4   &   -p_5  \\
    p_0'  &   -p_4'  &   -p_5' \\
    p_0'' &   -p_4'' &   -p_5''
\end{pmatrix} \textrm{ in } III, \quad &
P=\begin{pmatrix}
    p_0   &   -p_3   &   -p_5  \\
    p_0'  &   -p_3'  &   -p_5' \\
    p_0'' &   -p_3'' &   -p_5''\\
\end{pmatrix} \textrm{ in } IV.
\end{align}
Hence each column of  $P$ contains a particular Pearcey integral
$p_j$, see \eqref{Pearcey-integrals}, and its derivatives. We use
$p_0$ in the first column and $p_5$ in the last column. In the
middle column we use different Pearcey integrals in the different
quadrants.

Then $P$ is analytic in $\C\setminus \big(\R\cup {\rm i}\R\big)$
with a jump on $\R \cup {\rm i}\R$. By linear relations such as
\eqref{eq:p5p4p1} one can easily obtain jump conditions for
$P$. In the following lemma we show the RH problem that
is satisfied by this function $P$. 

\begin{lemma} \label{lem:RHforP}
We have that $P$ satisfies the following RH problem:
\begin{equation}
\left\{
\begin{array}{l}
\multicolumn{1}{l}{P \textrm{ is analytic in } \C \setminus (\R \cup  {\rm i} \R)}, \\
P_+(z)= P_-(z) \begin{pmatrix}
                            1 & 0 & 0\\
                            0 & 1 & 0\\
                            0 & -1 & 1
                            \end{pmatrix},
                            \quad  z\in  \R, \\
P_+(z) = P_-(z) \begin{pmatrix}
                            1 & -1 & 0\\
                            0 & 1 & 0\\
                            0 & 0 & 1
                            \end{pmatrix},
                            \quad  z\in {\rm i} \R, \\
P(z)=\sqrt{2\pi}(I+\OO(z^{-2/3})) \begin{pmatrix}
                    z^{-1/3} & 0 & 0\\
                    0        & 1 & 0\\
                    0        & 0 & z^{1/3}
                \end{pmatrix} B_j \Theta_j(z),\\
                \multicolumn{1}{r}{\textrm{as } z\to \infty \textrm{ in the }$j$\textrm{th } quadrant,}
            \end{array}\right.
\end{equation}
where the constant matrices $B_j$ are given by
 (we use $\omega={\rm e}^{2\pi {\rm i}/3}$),
\begin{align} \label{eq:B1B2}
B_1=\frac{{\rm i}}{\sqrt{3}}\begin{pmatrix}
1& -\omega^2 & -\omega\\
1& -1 & -1\\
1& -\omega & -\omega^2
\end{pmatrix}, & \qquad
B_2=\frac{{\rm i}}{\sqrt{3}}\begin{pmatrix}
-\omega^2& -1 & -\omega\\
-1& -1 & -1\\
-\omega& -1 & -\omega
\end{pmatrix},\\ \label{eq:B3B4}
B_3=\frac{{\rm i}}{\sqrt{3}}\begin{pmatrix}
-\omega& -1 & \omega^2\\
-1& -1 & 1\\
-\omega^2& -1& \omega
\end{pmatrix}, & \qquad
B_4=\frac{{\rm i}}{\sqrt{3}}\begin{pmatrix}
1& -\omega & \omega^2\\
1& -1 & 1\\
1& -\omega^2 & \omega
\end{pmatrix},
\end{align}
and the diagonal matrices $\Theta_j(z)$ for $j=1,\ldots,4$ are
given by
\begin{align} \label{eq:Theta1Theta2}
    \Theta_1(z)=\begin{pmatrix}
        {\rm e}^{\theta_3(z)} &0 & 0\\
        0& {\rm e}^{\theta_1(z)} & 0\\
        0& 0 & {\rm e}^{\theta_2(z)}
    \end{pmatrix}, & \quad
\Theta_2(z)=\begin{pmatrix}
{\rm e}^{\theta_1(z)} &0 & 0\\
0& {\rm e}^{\theta_3(z)} & 0\\
0& 0 & {\rm e}^{\theta_2(z)}
\end{pmatrix}\\ \label{eq:Theta3Theta4}
\Theta_3(z)=\begin{pmatrix}
{\rm e}^{\theta_2(z)} &0 & 0\\
0& {\rm e}^{\theta_3(z)} & 0\\
0& 0 & {\rm e}^{\theta_1(z)}
\end{pmatrix}, & \quad
\Theta_4(z)=\begin{pmatrix}
{\rm e}^{\theta_3(z)} &0 & 0\\
0& {\rm e}^{\theta_2(z)} & 0\\
0& 0 & {\rm e}^{\theta_1(z)}
\end{pmatrix}
\end{align}
with functions $\theta_j$ defined as
\begin{equation} \label{eq:thetaj}
\theta_j(z)= \frac{3}{4}\omega^j z^{4/3}, \quad j=1,2,3, \qquad
 -\pi < \arg z < \pi.
\end{equation}

The asymptotics for $P$ are uniform as $z \to \infty$ in any
region provided we stay away from the axis, that is, it is uniform
as $z \to \infty$ in the region
\begin{equation}\label{eq:epsilonarea}
    \dist\left(\arg z, \frac{\pi}{2} \mathbb Z \right) > \eps,
\end{equation}
for some $\eps>0$.
\end{lemma}

\begin{proof} In  \cite{BK2}  a detailed discussion of a RH
problem for Pearcey integrals can be found. The lemma is proved by
following these arguments. However, we do wish to make a few
remarks on the asymptotics of $P$. From classical steepest descent
arguments one can prove as in \cite{BK2} that
\begin{align}
p_0(z)&=\sqrt{\frac{2\pi}{3}} z^{-1/3} {\rm e}^{\frac{3}{4} z^{1/3}}(1+\OO(z^{-4/3}))\\
p_2(z)&=\omega^2\sqrt{\frac{2\pi}{3}} z^{-1/3} {\rm e}^{\frac{3}{4}\omega  z^{1/3}}(1+\OO(z^{-4/3}))\\
p_5(z)&=\omega \sqrt{\frac{2\pi}{3}} z^{-1/3} {\rm e}^{\frac{3}{4}
\omega^2 z^{1/3}}(1+\OO(z^{-4/3}))
\end{align}
as $z\to \infty$ in the first quadrant. In \cite{BK2} these
formulas (with a different numbering of the functions $p_j$) are
stated with an error term $\OO(z^{-2/3})$. However, our case is  a
special case of the more general case considered in \cite{BK2}. By
analyzing the proof for our special case we obtain that the error
term is indeed of order $\OO(z^{-4/3})$. From these asymptotics
(and the corresponding ones for the derivatives of the Pearcey
integrals) we obtain by \eqref{eq:PinIandII} that
\begin{equation}
P(z)=\sqrt{\frac{2\pi}{3}} \begin{pmatrix} z^{-1/3} & 0 & 0\\
0 & 1 & 0\\
0 & 0 & z^{1/3}
\end{pmatrix} (1+\OO(z^{-4/3}))  B_1 \Theta_1(z),
\end{equation}
as  $z\to \infty$ in the first quadrant. Since
\begin{multline}
\begin{pmatrix} z^{-1/3} & 0 & 0\\
0 & 1 & 0\\
0 & 0 & z^{1/3}
\end{pmatrix} (1+\OO(z^{-4/3}))\\
= (1+\OO(z^{-2/3}))\begin{pmatrix} z^{-1/3} & 0 & 0\\
0 & 1 & 0\\
0 & 0 & z^{1/3}
\end{pmatrix}, \qquad \text{ as } z \to \infty,
\end{multline}
we obtain the asymptotic behavior of $P$ in the first quadrant.
The other quadrants can be dealt with in a similar way.

Finally, because of the Stokes phenomenon, the asymptotics for
$p_0$ are only uniform if we stay away from the imaginary axis.
Similarly, the asymptotics for $p_5$ are uniform if we stay away
from the real axis. Hence the asymptotic behavior is uniform if we
stay away from both the real and imaginary axis and let
$z\to\infty$ so that it remains in the region given by
\eqref{eq:epsilonarea} for some $\eps>0$.
\end{proof}

\subsection{The inverse transpose of $P$}
In the transformation $Y\mapsto X$ we will use
 \begin{equation} \label{eq:tildeQ}
    Q = P^{-t}.
 \end{equation}
Then $Q$ satisfies a RH problem which is easily obtained from the
RH problem satisfied by $P$.

\begin{lemma}\label{lem:RHforQ}
Define $Q= P^{-t}$. Then $Q$ satisfies the following RH problem:
\begin{equation}\label{eq:RHproblemtildeQ}
\left\{
\begin{array}{l}
\multicolumn{1}{l}{Q \textrm{ is analytic in } \C \setminus (\R \cup  {\rm i} \R)}, \\
Q_+(z) = Q_-(z) \begin{pmatrix}
                            1 & 0 & 0\\
                            0 & 1 & 1\\
                            0 & 0 & 1
                            \end{pmatrix},
                            \quad  z\in  \R,\\
Q_+(z) = Q_-(z) \begin{pmatrix}
                            1 & 0 & 0\\
                            1 & 1 & 0\\
                            0 & 0 & 1
                            \end{pmatrix},
                            \quad  z\in {\rm i} \R,\\
Q(z)=\frac{1}{\sqrt{2\pi}}(I+\OO(z^{-2/3}))
    \begin{pmatrix}
                    z^{1/3} & 0 & 0\\
                    0        & 1 & 0\\
                    0        & 0 & z^{-1/3}
                \end{pmatrix} A_j \Theta_j(z)^{-1},\\
                \multicolumn{1}{r}{\textrm{as }z\to \infty \textrm{ in the }j\textrm{th } quadrant,}
            \end{array}\right.,
\end{equation}
where $A_j = B_j^{-t}$, and the $B_j$ and $\Theta_j$, $j=1,2,3,4$,
are given by \eqref{eq:B1B2}--\eqref{eq:Theta3Theta4}.

The asymptotics for $Q$ are uniform as $z \to \infty$ in any
quadrant such that \eqref{eq:epsilonarea} holds for some $\eps
> 0$. \end{lemma}

\begin{proof} This follows immediately from \eqref{eq:tildeQ} and
Lemma \ref{lem:RHforP}. \end{proof}

The constant matrices $A_j$ are given explicitly by (with
$\omega={\rm e}^{2\pi {\rm i}/3}$),
\begin{align} \label{eq:A1A2}
A_1=\frac{{\rm i}}{\sqrt{3}}\begin{pmatrix}
-1& \omega & \omega^2\\
-1& 1 & 1\\
-1& \omega^2 & \omega
\end{pmatrix}, & \qquad
A_2=\frac{{\rm i}}{\sqrt{3}}\begin{pmatrix}
\omega& 1 & \omega^2\\
1& 1 & 1\\
\omega^2& 1 & \omega
\end{pmatrix}, \\ \label{eq:A3A4}
A_3=\frac{{\rm i}}{\sqrt{3}}\begin{pmatrix}
\omega^2& 1 & \omega\\
1& 1 & 1\\
\omega& 1& \omega^2
\end{pmatrix}, & \qquad
A_4=\frac{{\rm i}}{\sqrt{3}}\begin{pmatrix}
-1& \omega^2 & \omega\\
-1& 1 & 1\\
-1& \omega & \omega^2
\end{pmatrix}.
\end{align}
The prefactor ${\rm i}/\sqrt{3}$ is so that all matrices $A_j$
have determinant $1$.

\subsection{The transformation $Y \mapsto X$}

For clarity we define the transformation $Y \mapsto X$ in two
steps. First we define $Y \mapsto \widehat{X}$ and then
$\widehat{X} \mapsto X$.

We define $\widehat X$ by
\begin{equation}\label{eq:YtohatX}
    \widehat X(z)=\begin{pmatrix} 1 & 0 \\
    0 & C_n \end{pmatrix}
    Y(z)\begin{pmatrix} 1&0\\
    0& D_n Q(n^{3/4} \tau z)
    \end{pmatrix},
\end{equation}
where $D_n$ and $Q$ are given in \eqref{eq:defDn},
\eqref{eq:tildeQ}, and the constant prefactor $C_n$ is given by
\begin{equation}\label{eq:defCn}
C_n= \sqrt\frac{2\pi}{ n} \begin{pmatrix}
  \tau^{-1/3} &0 & 0\\
 0& 1  & 0\\
 0 & 0 & \tau^{1/3}
\end{pmatrix}. \end{equation}

We also define
\begin{equation} \label{eq:YtoX}
    X(z) = \widehat{X}(z) \begin{pmatrix} 1 & 0 \\
        0 & \Theta_j(n^{3/4} \tau z) \end{pmatrix}
        \end{equation}
        for $z$ in the $j$th
quadrant, where $\Theta_j$ is given by
\eqref{eq:Theta1Theta2}--\eqref{eq:Theta3Theta4}.

The $4 \times 4$ matrix valued functions $\widehat X$ and $X$ are
analytic in each of the four quadrants with jumps across the real
and imaginary axes.

\begin{lemma}
We have that $\widehat X$ satisfies the following RH problem:
\begin{equation} \label{eq:RHforhatX}
\left\{
\begin{array}{l}
\multicolumn{1}{l}{\widehat X \textrm{ is analytic in } \C \setminus (\R \cup  {\rm i} \R),}\\
\widehat X_+(z)=\widehat X_-(z)
    \begin{pmatrix} 1 & e^{-nV(z)} & 0 & 0 \\
        0 & 1 & 0 & 0 \\
        0 & 0 & 1 & 1 \\
        0 & 0 & 0 & 1
        \end{pmatrix}, \quad z \in \R, \\
\widehat X_+(z) = \widehat X_-(z)
    \begin{pmatrix} 1 & 0 & 0 & 0 \\
        0 & 1 & 0 & 0 \\
        0 & 1 & 1 & 0 \\
        0 & 0 & 0 & 1 \end{pmatrix}, \quad z \in {\rm i}\R, \\
\widehat X(z)=(I+\OO(z^{-2/3})) \begin{pmatrix}
    z^n&0 & 0 & 0 \\
    0 & z^{-\frac{n+1}{3}} & 0 & 0\\
    0 & 0 & z^{-\frac{n}{3}} & 0 \\
    0 & 0 & 0 & z^{-\frac{n-1}{3}}
                \end{pmatrix}
                \begin{pmatrix} 1 &  0 \\
                0& A_j \Theta_j^{-1}(n^{3/4} \tau z) \end{pmatrix} \\
                \multicolumn{1}{r}{\textrm{as }z\to \infty \textrm{ in the }j\textrm{th quadrant}}
            \end{array}\right.
\end{equation}
The asymptotics for $\widehat X$ are uniform as $z \to \infty$ in
any quadrant provided we stay away from the axis, that is, it
holds uniformly as $z \to \infty$ with \eqref{eq:epsilonarea} for
some $\eps > 0$.
\end{lemma}

\begin{proof}
Let $z \in \R$. Then we obtain from \eqref{RH-problem-pn} and
\eqref{eq:YtohatX} that
 \begin{multline}
    \widehat X_-^{-1}(z) \widehat X_+(z) =
    \begin{pmatrix}
    1 & \begin{pmatrix} w_{0,n}(z) & w_{1,n}(z) & w_{2,n} (z) \end{pmatrix} D_n Q_+(n^{3/4} \tau z)\\
    0& Q_-(n^{3/4}\tau z)^{-1} Q_+(n^{3/4} \tau z)
    \end{pmatrix}.
 \end{multline}
Since $Q = P^{-t}$, we can use \eqref{eq:mainideatrickBHI} to find
\[ \begin{pmatrix} w_{0,n}(z) & w_{1,n}(z) & w_{2,n} (z) \end{pmatrix} D_n Q_+(n^{3/4} \tau z)
    = \begin{pmatrix} e^{-nV(z)} & 0 & 0 \end{pmatrix}.
\]
Using also the jump condition on the real axis in the RH problem
\eqref{RH-problem-pn} for $Q$, we obtain the jump matrix for
$\widehat X$ on the real line as given in the RH problem
\eqref{eq:RHforhatX}.

The proof of the jump on the imaginary axis is similar (even
simpler).

 Finally we check the asymptotic formula for $\widehat X$. From
the asymptotic behavior of $Q$  given in
\eqref{eq:RHproblemtildeQ} and the definitions \eqref{eq:defDn}
and \eqref{eq:defCn},
 of $D_n$ and $C_n$, we obtain as $z \to \infty$ in the $j$th quadrant,
\begin{multline} \label{eq:asymptoticshatXstep1}
    D_n Q (n^{3/4} \tau z)
    = C_n^{-1} (I+\OO(z^{-2/3})) \begin{pmatrix}
    z^{1/3} & 0 & 0\\
    0 & 1 & 0 \\
    0 & 0  & z^{-1/3}
    \end{pmatrix} A_j \Theta_j^{-1}(n^{3/4} \tau z).
\end{multline}
 From the asymptotic condition in the RH problem
\eqref{RH-problem-pn} for $Y$ and the definition
\eqref{eq:YtohatX}, it follows that
\begin{align}
\widehat X(z)&=\begin{pmatrix} 1 & 0 \\
0 & C_n \end{pmatrix} \left(I+\OO(z^{-1})\right)\begin{pmatrix}
z^n & 0 & 0 & 0\\
0 & z^{-n/3} & 0 & 0 \\
0 & 0 & z^{-n/3} & 0 \\
0 & 0 & 0 & z^{-n/3}
\end{pmatrix}\\
\nonumber& \qquad \times
\begin{pmatrix} 1&0\\
0& D_n Q(n^{3/4} \tau z)
\end{pmatrix}.
\end{align}
Combining this with \eqref{eq:asymptoticshatXstep1} gives the
asymptotic condition for $\widehat X$ in \eqref{eq:RHforhatX}.
\end{proof}

The transformation \eqref{eq:YtoX} from $\widehat X$ to $X$ has
the effect of simplifying the asymptotic condition in the RH
problem, but the jumps on the real and imaginary axis are more
complicated.

\begin{lemma}
We have that $X$ satisfies the following RH problem:
\begin{equation} \label{eq:RHforX}
\left\{
\begin{array}{l}
\multicolumn{1}{l}{X \textrm{ is analytic in } \C \setminus (\R \cup  {\rm i} \R),}\\
    X_+(z)= X_-(z)J_{X}(z), \quad z\in \R\cup {\rm i}\R,\\
    X(z)=(I+\OO(z^{-2/3})) \begin{pmatrix}
    z^n&0 &0&0\\
    0&z^{-\frac{n+1}{3}} & 0 & 0\\
    0& 0& z^{-\frac{n}{3}} & 0\\
    0& 0 & 0 & z^{-\frac{n-1}{3}}
                \end{pmatrix}
                \begin{pmatrix}1 &  0 \\
                0& A_j\end{pmatrix}\\
                \multicolumn{1}{r}{\textrm{as }z\to \infty \textrm{ in the }j\textrm{th quadrant}}
            \end{array}\right.
\end{equation}
where the jump matrices $J_{X}$ are given by
\begin{align} \label{eq:jumpforXonR}
J_{X}(z)&= \begin{pmatrix}
1& {\rm e}^{-n (V(z)-\frac{3}{4}\tau^{4/3} |z|^{4/3}}) & 0 & 0\\
                            0& 1 & 0 & 0\\
                            0& 0 & {\rm e}^{\pm n \frac{3\sqrt{3}}{4}{\rm i} \tau^{4/3}|z|^{4/3}} & 1\\
                            0& 0 & 0 & {\rm e}^{\mp n \frac{3\sqrt{3}}{4}{\rm i}\tau^{4/3}|z|^{4/3}}
                            \end{pmatrix},
\end{align}
for  $z\in  \R_\pm$, and
\begin{align}\label{eq:jumpforXoniR}
    J_{X}(z)&= \begin{pmatrix}
               1&0 & 0 & 0\\
                        0 & {\rm e}^{\mp n \frac{3\sqrt{3}}{4} {\rm i} \tau^{4/3} |z|^{4/3}} & 0 & 0\\
                        0&  1 & {\rm e}^{\pm n \frac{3\sqrt{3}}{4}{\rm i} \tau^{4/3} |z|^{4/3}}  & 0\\
                        0&  0 & 0 & 1
                            \end{pmatrix},
\end{align}
for $ z\in {\rm i} \R_{\pm}$.

The asymptotics for $X$ are uniform as $z \to \infty$ in any
quadrant provided we stay away from the axis, that is, it holds
uniformly as $z \to \infty$ with \eqref{eq:epsilonarea} for some
$\eps > 0$.
\end{lemma}

\begin{proof} The asymptotic condition is clear from
\eqref{eq:YtoX} and the asymptotic condition in the RH problem
\eqref{eq:RHforhatX} for $\widehat X$.

Let $z \in \R_+$. Then it follows from \eqref{eq:YtoX}, the jump
condition in the RH problem \eqref{eq:RHforhatX}, and the
expressions for $\Theta_1$ and $\Theta_4$ from
\eqref{eq:Theta1Theta2}--\eqref{eq:Theta3Theta4} that
 \begin{multline}
 X_-^{-1}(z) X_+(z) =
 \begin{pmatrix}1 &0 \\
 0 & \Theta_{4}^{-1}(n^{3/4} z)
 \end{pmatrix}
 \begin{pmatrix}
1 & {\rm e}^{-nV(z)} & 0 & 0 \\
0 & 1 & 0 & 0 \\
0 & 0 & 1 & 1 \\
0 & 0 & 0 & 1
\end{pmatrix}
\begin{pmatrix} 1 &0 \\
 0 & \Theta_{1}(n^{3/4} z)
 \end{pmatrix} \\
 =
\begin{pmatrix}
    1 & {\rm e}^{-n V(z)+\theta_3(n^{3/4}n z)} & 0 &0\\
    0 & 1 & 0 & 0\\
    0 & 0   & {\rm e}^{\theta_1(n^{3/4} \tau z)-\theta_2(n^{3/4}\tau z)}  & 1 \\
    0 & 0 & 0 & {\rm e}^{\theta_1(n^{3/4} \tau z)-\theta_2(n^{3/4}\tau z)}
\end{pmatrix}.
\end{multline}
Inserting $\theta_j(z)=\frac{3}{4}\omega^j z^{4/3}$ we arrive at
the jump matrix \eqref{eq:jumpforXonR} for $z\in \R_+$.

The proof of \eqref{eq:jumpforXonR} for $z\in \R_-$ is similar,
taking into account that the functions $\theta_j$ have their
branch cuts along the negative real axis.

The proof of the jump matrix \eqref{eq:jumpforXoniR} for $z \in
{\rm i}\R$ is similar.
\end{proof}
The fact that the asymptotics for $X$ is not uniform up to the
axis is a complication that will be automatically resolved during
the steepest descent analysis.

The RH problem for $X$ reveals to some extent how the equilibrium
problem comes into play. We see that the jump matrices in the RH
problem locally have the structure of jump matrices in a $2\times
2$ matrix valued RH problem. The measure $\mu_1$ will be used to
handle the upper left $2 \times 2$ block in the jump matrix
\eqref{eq:jumpforXonR} on $\R$ and the measure $\mu_3$ will handle
the lower right block. The measure $\mu_2$ is supported on the
imaginary axis and will handle the middle $2 \times 2$ block in
the jump matrix \eqref{eq:jumpforXoniR} on the imaginary axis.

Thus $\mu_1$ will act on the first and second rows and columns,
$\mu_2$ and the second and third, and $\mu_3$ and the third and
fourth. Then we also see that $\mu_1$ and $\mu_2$ interact in some
way, since they both act on the second row and column. Similarly,
$\mu_2$ and $\mu_3$ will interact. On the other hand, the measures
$\mu_1$ and $\mu_3$ do not interact with each other, since they
act on different rows and columns. This is reflected in the
equilibrium problem where we have that the mutual logarithmic
energies $I(\mu_1, \mu_2)$ and $I(\mu_2, \mu_3)$ are present, but
not the mutual logarithmic energy of $\mu_1$ and $\mu_3$.

The potential $V(z)$ in \eqref{eq:jumpforXonR} comes with an extra
term $-\frac{3}{4}\tau^{4/3}|z|^{4/3}$. Together they appear in
the equilibrium problem as the external field on $\mu_1$. The
lower $2\times 2$ block in \eqref{eq:jumpforXonR} and the middle
$2\times 2$ block in \eqref{eq:jumpforXoniR} have rapidly
oscillating entries on the diagonal. This is reminiscent of what
happens in the steepest descent analysis of the RH problem for
orthogonal polynomials after the application of the $g$-functions
\cite{DKMVZuniform}.

We will explain at a later stage why the restriction $\mu_2 \leq
\sigma$ plays a role. Here we can only say that it is partly
explained by the triangular structure of the jump matrices
\eqref{eq:jumpforXoniR} and \eqref{eq:jumpforXonR}. Indeed the
non-trivial $2\times 2$ block in \eqref{eq:jumpforXoniR} is lower
triangular, while the $2\times 2$ blocks in \eqref{eq:jumpforXonR}
are upper triangular.

\section{Analysis of the equilibrium problem}

To prepare for the next transformation of the RH problem we need
the equilibrium problem for the energy functional $E_V$. In this
section we prove Theorem \ref{prop:structureEqmeasure} which gives
existence and properties of the equilibrium measures $\mu_1$,
$\mu_2$ and $\mu_3$. We also prove Propositions
\ref{prop:variationalconditions} and \ref{Th:convexity}. A basic
tool is the balayage (sweeping out) of a measure.

\subsection{Balayage}
The balayage of a finite measure $\nu$ onto a closed set $K$ with
positive capacity is a positive measure $\widehat \nu$ on $K$ with
$\|\nu\|=\|\widehat \nu\|$ and with the property that there exists
a constant $C$ such that
\begin{equation} \label{balayage}
      U^{\nu}(x) = U^{\widehat{\nu}}(x) + C, \qquad \text{q.e. } x \in
      K,
\end{equation}
where q.e.\ means quasi-everywhere, that is, with the exception of
a set of zero capacity. We will be interested in cases where $K$
is either the real or imaginary axis, or $K$ consists of two
unbounded intervals along the imaginary axis
\begin{equation} \label{eq:definitionKc}
    K_c = (-{\rm i} \infty, -{\rm i}c] \cup [{\rm i}c, {\rm i}\infty), \qquad c > 0.
\end{equation}
In these cases we have $C=0$ for the constant in \eqref{balayage}.

We use the notation $\widehat{\nu}$ for the balayage of $\nu$ if
$K$ is understood. To emphasize $K$ we write
\begin{equation} \label{eq:notationbalayage}
    \widehat{\nu} = \Bal(\nu, K).
\end{equation}

Each balayage measure can be written as the integral over balayage
measures of Dirac-delta measures as follows,
    \begin{align} \label{eq:integralofbalayage}
        \Bal(\nu,K)=\int \Bal(\delta_z,K) \ {\rm d}\nu(z).
    \end{align}
Hence it is sufficient (for our purposes) to calculate
$\Bal(\delta_z,K)$ for several $z$ and $K$.

\begin{example} \label{example1}
If $K = \mathbb R$ is the real line, and $y \neq 0$, then the
balayage measure of $\delta_{{\rm i}y}$ onto $\R$ has the Cauchy density
\begin{equation} \label{eq:balayageyonR}
         \frac{{\rm d} \Bal(\delta_{{\rm i}y}, \R)(x)}{{\rm d}x} = \frac{1}{\pi} \frac{|y|}{x^2 + y^2}, \qquad  x \in \R.
\end{equation}

For $K = {\rm i} \mathbb R$ and $x \neq 0$ we have similarly
\begin{equation}\label{eq:balayagexoniR}
    \frac{{\rm d} \Bal(\delta_x, {\rm i} \R)}{|{\rm d}z|} = \frac{1}{\pi} \frac{|x|}{|z|^2 + x^2}, \qquad z \in {\rm i} \R.
\end{equation}

These results are well-known. They are also the limiting case $c
\to 0+$ of \eqref{eq:balayagexonKc} in  the next example.
\end{example}

\begin{example} \label{example2}
Let $K_c = (-{\rm i}\infty,-{\rm i}c] \cup [{\rm i}c,{\rm  i} \infty)$ with $c >0$. If $y
\in (-c,c)$ then the balayage of $\delta_{{\rm i}y}$ onto $K_c$ has the
density
\begin{equation} \label{eq:balayageiyonKc}
    \frac{{\rm d} \Bal(\delta_{{\rm i} y}, K_c)}{|{\rm d}z|} =
               \frac{1}{\pi}\frac{\sqrt{c^2-y^2}}{|z - {\rm i}y| \sqrt{|z|^2-c^2}}, \qquad z \in K_c.
         \end{equation}

If $x \in \R$ then the balayage of $\delta_x$ onto $K_c$ has the
density
\begin{equation} \label{eq:balayagexonKc}
    \frac{{\rm d} \Bal(\delta_{x}, K_c)}{|dz|} =
               \frac{1}{\pi} \frac{|z| \sqrt{c^2+x^2}}{(|z|^2 + x^2) \sqrt{|z|^2-c^2}}, \qquad z \in K_c.
         \end{equation}
In all cases we have that the constant $C$ in \eqref{balayage}
vanishes.

These expressions can be proved by simple contour integration. We
omit the proofs. It will be important for us that the densities
\eqref{eq:balayageiyonKc} and \eqref{eq:balayagexonKc} are
decreasing as $|z|$ increases, with a decay rate $O(|z|^{-2})$ as
$|z| \to \infty$.
\end{example}

In the following we will see how the notion of balayage is related
to the equilibrium problem.

\subsection{Equilibrium problem for $\nu_3$}

If we fix two of the measures $\nu_j$, $j=1,2,3$ we can consider
the equilibrium problem with respect to the remaining measure
only.
If $\nu_1$ and $\nu_2$ are given, then the equilibrium problem for
$\nu_3$ is to minimize
\begin{equation} \label{eq:equilibriumnu3}
    I(\nu) - \int U^{\nu_2} {\rm d}\nu
\end{equation}
with respect to all measures $\nu$ on $\R$ with total mass
$1/3$. This problem only depends on $\nu_2$. Assuming that $\nu_2$
has finite logarithmic energy and total mass $2/3$ (as in conditions
(a) and (c) of the equilibrium problem), the minimizer for \eqref{eq:equilibriumnu3}
is given by $\nu_3 = \frac{1}{2} \Bal(\nu_2, \R)$
since it satisfies $U^{\nu_3} = \frac{1}{2} U^{\nu_2}$ on $\R$,
which is the variational condition for \eqref{eq:equilibriumnu3},
and it has the correct total mass $1/3$. So we have

\begin{lemma} \label{lemma:equilibriumnu3}
Suppose $\nu_1$ and $\nu_2$ are fixed so that conditions {\rm
(a)}, {\rm (b)}, and {\rm (c)} of the equilibrium problem are
satisfied. Then the measure $\nu_3$ that minimizes $E_V(\nu_1,
\nu_2,\nu_3)$ subject to the condition {\rm (d)} is given by
\begin{equation} \label{eq:nu3asbalayage}
    \nu_3 = \frac{1}{2} \Bal(\nu_2, \R).
\end{equation}
Thus by \eqref{eq:balayageyonR} and \eqref{eq:integralofbalayage},
we see that $\nu_3$ has a density given by
\[ \frac{{\rm d} \nu_3}{{\rm d}x}(x) = \frac{1}{\pi} \int \frac{|z|}{x^2 + |z|^2}
    {\rm d} \nu_2(z). \]
\end{lemma}

\subsection{Equilibrium problem for $\nu_1$}

For given $\nu_2$ and $\nu_3$ having finite logarithmic energy,
the equilibrium problem for $\nu_1$ is to minimize
\begin{equation} \label{eq:equilibriumfornu1}
    I(\nu) + \int \left(V(x) - \frac{3}{4} \tau^{4/3} |x|^{4/3} - U^{\nu_2}(x) \right) {\rm d}\nu(x)
\end{equation}
among all probability measures $\nu$ on $\R$. This is a usual equilibrium
problem on $\R$ with external field
\[ V(x) - \frac{3}{4} \tau^{4/3} |x|^{4/3} - U^{\nu_2}(x). \]
The term $-U^{\nu_2}$ in the external field attracts the measure
$\nu$ towards $0$.

\begin{lemma} \label{lemma:equilibriumnu1}
Suppose $\nu_2$ and $\nu_3$ are fixed so that conditions {\rm
(a)}, {\rm (c)}, {\rm (d)} and {\rm (e)} of the equilibrium
problem are satisfied. Then the measure $\nu_1$ that minimizes
$E_V(\nu_1, \nu_2,\nu_3)$ subject to the condition {\rm (b)} is
the minimizer for \eqref{eq:equilibriumfornu1}.

It has the following properties:
\begin{enumerate}
\item[\rm (a)] The support of $\nu_1$ is contained in the convex
hull of the support of the equilibrium measure on $\R$ in the
external field $V(x) -\frac{3}{4} \tau^{4/3} |x|^{4/3}$.
 \item[\rm (b)] If $\nu_2$ is so that there exists a constant $c > 0$ such
that $S(\sigma-\nu_2) \subset K_c$ where $K_c$ is given by \eqref{eq:definitionKc},
then
\[ \frac{3}{4} \tau^{4/3} |x|^{4/3} + U^{\nu_2}(x) \]
is real analytic on $\mathbb R$.
\end{enumerate}
\end{lemma}
\begin{proof} (a) We use the fact that the support of the equilibrium
measure with external field $Q$ on $\R$ is the closure of those
points where weighted polynomials $e^{-nQ} P_n^2$, $\deg P_n \leq
n$, $n \in \N$, take their maximum modulus \cite[Chapter
III.2]{ST}.

Assume the equilibrium measure in the external field $Q(x) = V(x)
- \frac{3}{4} \tau^{4/3} |x|^{4/3}$ is supported in $[-X,X]$ for
some $X > 0$. Hence the maximum modulus of $e^{-nQ} P_n^2$, $\deg
P_n \leq n$, is attained only in $[-X,X]$. Since $U^{\nu_2}(x) =
\int \log \frac{1}{|x-z|} {\rm d}\nu_2(z)$ is even on $\R$ and
strictly decreasing as $|x|$ increases, it then follows that the
maximum modulus of $e^{n U^{\nu_2}} e^{-nQ} P_n^2 = e^{-n(Q -
U^{\nu_2})} P_n^2$, $\deg P_n \leq n$ is attained on $[-X,X]$
only, which proves part (a) of the lemma.

(b) Suppose  $\nu_2$ and $c > 0$ are such that $S(\sigma-\nu_2)
\subset K_c$. Then
\begin{equation} \label{eq:realanalytic1}
    U^{\nu_2}(x) = - \int_{-{\rm i}c}^{{\rm i}c} \log|x-z| {\rm d}\sigma(z)
    - \int_{|z| \geq c} \log|x-z| {\rm d}\nu_2(z),
\end{equation}
and the second term is real analytic for $x \in \R$.
For the first term in \eqref{eq:realanalytic1} we have by
the definition \eqref{eq:densitysigma} of $\sigma$,
\begin{equation} \label{eq:realanalytic2}
    - \int_{-{\rm i}c}^{{\rm i}c} \log|x-z| {\rm d}\sigma(z)
    = - \frac{\sqrt{3}}{2 \pi} \tau^{4/3} \int_0^c \log(x^2+y^2) y^{1/3} {\rm d}y
\end{equation}
whose derivative with respect to $x$ is
\begin{multline} \label{eq:realanalytic3}
    -\frac{\sqrt{3}}{\pi} \tau^{4/3} \int_0^c \frac{xy^{1/3}}{x^2+y^2} {\rm d}y \\
    =
 -\frac{\sqrt{3} }{\pi} \tau^{4/3}\int_0^{\infty} \frac{x y^{1/3}}{x^2+y^2}  {\rm d}y
    + \frac{\sqrt{3}}{\pi} \tau^{4/3} \int_c^{\infty} \frac{x y^{1/3}}{x^2+y^2}  {\rm d}y.
    \end{multline}
The second term on the right-hand side of \eqref{eq:realanalytic3}
is real analytic on $\R$. In the first term we introduce the
change of variables $y = |x| s$ to obtain
\[  \int_0^{\infty} \frac{x y^{1/3}}{x^2+y^2} {\rm d}y = x |x|^{-2/3}
    \int_0^{\infty} \frac{ s^{1/3}}{1+s^2} {\rm d}s
\]
Because of the standard integral
\begin{equation} \label{eq:standardintegral}
    \int_0^{\infty} \frac{s^p}{1+s^2} {\rm d}s = \frac{\pi}{2 \cos (p\pi/2)},
    \qquad -1  < p < 1, \end{equation}
which for $p =1/3$ gives the value $\frac{1}{3} \sqrt{3} \pi$,
it then follows that
\[ -\frac{\sqrt{3}}{\pi} \tau^{4/3} \int_0^{\infty} \frac{x y^{1/3}}{x^2+y^2}  {\rm d}y
    = - \tau^{4/3} x |x|^{-2/3}. \]

Combining this with \eqref{eq:realanalytic1},
\eqref{eq:realanalytic2}, and \eqref{eq:realanalytic3}, we obtain
\[ \frac{{\rm d}}{{\rm d}x} U^{\nu_2}(x) = - \tau^{4/3} x |x|^{-2/3} + \text{``real analytic function''}, \]
which after integration completes the proof of part (b) of the
lemma. \end{proof}

Lemmas \ref{lemma:equilibriumnu3} and \ref{lemma:equilibriumnu1}
show how the solutions of the equilibrium problems for $\nu_1$ and
$\nu_3$ are determined by the measure $\nu_2$. Next we study the
converse: given $\nu_1$ and $\nu_3$, what are the properties of
the solution of the equilibrium problem for $\nu_2$?

\subsection{Equilibrium problem for $\nu_2$}
Now suppose that  we are given $\nu_1$ and $\nu_3$ satisfying
the conditions (a), (b), (d) of the equilibrium problem and we
wish to minimize $E_V(\nu_1,\nu_2,\nu_3)$ with respect to $\nu_2$.
The equilibrium problem for $\nu_2$ is to minimize
\[ I(\nu) - \int \left(U^{\nu_1} + U^{\nu_3}\right) {\rm d}\nu \]
with respect to all measures $\nu$ on ${\rm i}\R$ with total mass $2/3$
and satisfying the constraint $\nu \leq \sigma$.

If the constraint $\nu \leq \sigma$ were not present, then
simple balayage arguments would give that the minimizer is given
by
\[ \frac{1}{2}\Bal\left(\nu_1 + \nu_3,{\rm i}\R\right). \]
Note that this measure has indeed the correct total mass $2/3$.
However, this measure will violate the constraint. Indeed, the
density of the balayage  of any measure on $\R$ onto ${\rm
i}\R$ has its maximum at $0$, and  is strictly decreasing for $z \in {\rm i}\R$
if $|z|$ increases, see also \eqref{eq:balayagexoniR}.
Since  the density of $\sigma$ \eqref{eq:densitysigma} vanishes at $z=0$
and is strictly increasing if $|z|$ increases, it is
clear that the  density of the balayage measure lies strictly above the
density of $\sigma$ on a non-empty symmetr{\rm i}c interval around $0$ on ${\rm i}\R$.

\begin{lemma} \label{lemma:equilibriumnu2}
Suppose that $\nu_1$ and $\nu_3$ are fixed so that conditions {\rm
(a)}, {\rm (b)}, and {\rm (d)} of the equilibrium problem are
satisfied. Then the measure $\nu_2$ that minimizes $E_V(\nu_1,
\nu_2,\nu_3)$ subject to the conditions (c) and (e) exists, and
there is a constant $c >0$ such that
\[ \supp(\sigma - \nu_2) = K_c := (-{\rm i} \infty, -{\rm i}c] \cup [{\rm i}c, {\rm i} \infty). \]
Moreover, $\nu_2$ satisfies
\begin{align}\label{eq:variationalcondnu2}
            \begin{cases}
                U^{\nu_{2}}(z) = U^{\nu_1}(z)+U^{\nu_3}(z), & z\in {\rm i}\R\setminus (-{\rm i}c,{\rm i}c)\\
                U^{\nu_{2}}(z) < U^{\nu_1}(z)+U^{\nu_3}(z), & z\in  (-{\rm i}c,{\rm i}c).
            \end{cases}
    \end{align}
    and $\sigma-\nu_2$ vanishes as a square root at $\pm {\rm i}c$.
\end{lemma}
\begin{proof} We use the iterated balayage algorithm of
\cite{Kuijlaars-Dragnev}. Write
\[ \nu_{2,0} := \frac{1}{2}\Bal\left(\nu_1+\nu_3,{\rm i}\R\right). \]
As already noted above, there exists $c_0$ such that
\[ S((\nu_{2,0} - \sigma)^+) = [-{\rm i}c_0, {\rm i}c_0] \]
where $(\nu_{2,0} - \sigma)^+$ denotes the positive part of the
signed measure $\nu_{2,0} -\sigma$. It then follows from the
saturation principle of \cite{Dragnev-Saff} that the minimizer
$\nu_2$ (if it exists) is equal to the constraint $\sigma$ on
$[-{\rm i}c_0, {\rm i}c_0]$.

Then we define
\begin{align} \label{eq:iteratedbalayage1}
    \nu_{2,1} := \min(\sigma, \nu_{2,0}) +
    \Bal\left((\nu_{2,0}-\sigma)^+, K_{c_0}\right).
\end{align}
In other words, we balayage the part of $\nu_{2,0}$ that lies
above $\sigma$ onto the part of the imaginary axis $K_{c_0}$ where
the constraint is not (yet) active. Then $\nu_{2,1}$  has a
density for $z \in K_{c_0}$ that is strictly decreasing as $|z|$
increases, see also \eqref{eq:balayageiyonKc}. Then there exists
$c_1 > c_0$ so that
\[ S((\nu_{2,1} - \sigma)^+)= [-{\rm i}c_1, -{\rm i}c_0] \cup [{\rm i}c_0, {\rm i}c_1] \]
and we define
\begin{align} \label{eq:iteratedbalayage2}
 \nu_{2,2} := \min(\sigma, \nu_{2,1}) +
    \Bal\left((\nu_{2,1}-\sigma)^+, K_{c_1}\right).
\end{align}

Continuing this way we find an increasing sequence $(c_k)$ and a
sequence $(\nu_{2,k})$ of measures of total mass $2/3$. Clearly,
the $c_k$ must be such that
    \begin{align}
        \sigma([-{\rm i}c_k,{\rm i}c_k])\leq \frac{2}{3}.
    \end{align}
Hence the sequence converges and we denote the limit with $c$. It
follows as in \cite{Kuijlaars-Dragnev} that $(\nu_{2,k})$
converges weakly to a measure $\nu_2$. In fact, after each step in
the iterated balayage process we have that
    \begin{align}
            \begin{cases}
                U^{\nu_{2,k+1}}(x) = U^{\nu_{2,k}}(z), & z\in {\rm i}\R\setminus (-{\rm i}c_k,{\rm i}c_k), \\
                U^{\nu_{2,k+1}}(x) < U^{\nu_{2,k}}(z), & z\in  (-{\rm i}c_k,{\rm i}c_k).
            \end{cases}
    \end{align}
From these equations and the fact that
$2U^{\nu_{2,0}}(z)=U^{\nu_1}(z)+U^{\nu_3}(z)$ for all $z\in {\rm
i}\R$, we can easily deduce \eqref{eq:variationalcondnu2}. Hence
$\nu_2$  satisfies the variational conditions (with strict
inequality) that uniquely characterize the minimizer.

From the above arguments it also follows that the density of
$\nu_{2}$ for $z \in K_c$ is strictly decreasing as $|z|$
increases. Therefore  the density of $\nu_2$ lies strictly below
the density of $\sigma$ on $(-{\rm i}\infty, -{\rm i}c) \cup ({\rm
i}c,{\rm i}\infty)$ and there is a square root vanishing of the
density of $\sigma-\nu_2$ at $\pm {\rm i}c$.
\end{proof}

\subsection{Proof of uniqueness of the minimizer}

We start by rewriting the energy functional as
\begin{multline} \label{eq:energyrewrite}
    E_V(\nu_1,\nu_2,\nu_3) =
    \frac{1}{3} I(\nu_1)+\frac{1}{3} I(2\nu_1- 3\nu_2)+\frac{1}{4} I(\nu_2-2\nu_3)\\
    +\int \left( V(x)-\frac{3}{4}\tau^{4/3} |x|^{4/3} \right) {\rm d}\nu_1(x).
\end{multline}

If two measures $\nu$ and $\mu$ have finite logarithmic energy
and $\int {\rm d}\nu = \int {\rm d}\mu$, then
\begin{equation}\label{eq:positivelogenergy}
    I(\nu-\mu)\geq 0,
\end{equation}
with equality if and only if $\nu= \mu$. This is a well-known
result if $\nu_1$ and $\nu_2$ have compact supports \cite{ST}. For
measures with unbounded support, this is a more recent result of
Simeonov \cite{Simeonov}, who obtained this from a very elegant
integral representation for $I(\nu- \mu)$.

It follows from \eqref{eq:positivelogenergy} that the energy
functional \eqref{eq:energyrewrite} is strictly convex. The
minimizer is therefore uniquely characterized by the
Euler-Lagrange variational conditions listed in
Proposition~\ref{prop:variationalconditions}.

\subsection{Proof of the existence of the minimizer}

We rewrite $E_V$ again as in \eqref{eq:energyrewrite}. From
\eqref{eq:positivelogenergy} it follows that
\begin{align}
I(2\nu_1- 3\nu_2) > 0 \textrm{ and } I(\nu_2-2\nu_3) > 0.
\end{align}
Combining this with \eqref{eq:energyrewrite} leads to
\begin{align}
    E_V(\nu_1,\nu_2,\nu_3) &\geq \frac{1}{3} I(\nu_1) +
    \int \left(V(x)-\frac{3}{4}\tau^{4/3} |x|^{4/3} \right) {\rm d}\nu_1(x),
\end{align}
and this implies that $E_V$ is bounded from below.

A natural approach to prove the existence of the minimizer is now
the following. Since $E_V$ is bounded from below we can approach
the infimum by a sequence of measures satisfying the conditions of
the equilibrium problem. If we can prove that this sequence is
tight, then we have that the sequence converges weakly to the
minimizer of the $E_V$. However, it seems complicated to follow
this procedure. The main difficulty is the fact that the measures
$\nu_2$ and $\nu_3$ will have unbounded supports with rather fat
tails. It turns out that their densities behave like
$O(|z|^{-5/3})$ as $|z| \to \infty$.

Our method of proof consists of restricting the measure $\nu_2$
by an additional constraint
\begin{equation} \label{eq:extraconstraint}
    \nu_2 \leq K \sigma_{p}
\end{equation}
where $\sigma_{p}$ is the measure on ${\rm i}\R$ with density
\begin{equation} \label{eq:definitionsigmap}
    d\sigma_{p}(z) = \frac{1}{|z|^p} |dz|, \qquad z\in {\rm i}\R,
\end{equation}
with $p \in (1, 5/3)$  and $K > 0$ is a suitable constant that
will be determined later on. So we consider the equilibrium
problem with extra constraint \eqref{eq:extraconstraint} and we
prove that for the minimizer the extra constraint is not active.

The special role of the value $p = 5/3$ is made
clear in the following lemma.

\begin{lemma} \label{lemma:estimateonnu3}
Let $1 < p < 2$.
If $\nu_2 \leq K \sigma_{p}$ and $\nu_3 = \frac{1}{2} \Bal(\nu_2, \R)$
then
\begin{equation} \label{eq:estimateonnu3}
    \frac{{\rm d}\nu_3}{{\rm d}x} \leq C_p \frac{K}{|x|^p}, \qquad x \in \R,
\end{equation}
with constant
\begin{equation} \label{eq:definitionCp}
    C_p = \frac{1}{2 \sin(p \pi/2)}.
\end{equation}
We have $C_p < 1$ if and only if $p < 5/3$.
\end{lemma}
\begin{proof}
By direct calculation
\begin{align*}
    \frac{{\rm d}\nu_3}{{\rm d}x} & = \frac{1}{2\pi} \int \frac{|z|}{x^2 + |z|^2} {\rm d}\nu_2(z)
    \leq \frac{1}{2\pi} \int_{-\infty}^{\infty} \frac{|y|}{x^2 + y^2} \frac{K}{|y|^p} {\rm d}y \\
    & = \frac{K}{\pi} \int_0^{\infty} \frac{y^{1-p}}{x^2 + y^2} {\rm d}y,
    \qquad x \in \R.
\end{align*}
Make the change of variables $y = |x| s$ and use
\eqref{eq:standardintegral} to obtain
\begin{align*}
    \frac{{\rm d}\nu_3}{{\rm d}x} & \leq \frac{K}{\pi |x|^p} \int_0^{\infty} \frac{s^{1-p}}{1+s^2} {\rm d}s
     = \frac{K}{\pi |x|^p}  \frac{\pi}{2 \cos((1-p)\pi/2)} \\
    & = \frac{K}{|x|^p} \frac{1}{2 \sin(p \pi/2)}, \qquad x \in \R.
\end{align*}
This proves the lemma. Note that indeed $C_p < 1$ if and only if
$p < 5/3$.
\end{proof}

We may iterate the above argument and we obtain
\begin{corollary} \label{cor:estimateonBalnu3}
If $\nu_2 \leq K \sigma_{p}$ and $\nu_3 = \frac{1}{2} \Bal(\nu_2,
\R)$, then
\begin{equation} \label{eq:estimateonBalnu3}
    \frac{1}{2} \Bal(\nu_3, {\rm i}\R) \leq C_p^2 K \sigma_{p}.
\end{equation}
\end{corollary}

We fix $p \in (1,5/3)$ and we show first what constant $K > 0$ to
take. Let $X > 0$ be such that the support of the equilibrium
measure in the external field $V(x) - \frac{3}{4} \tau^{4/3}
|x|^{4/3}$ is contained in $[-X,X]$.

From the explicit formulas \eqref{eq:balayageiyonKc} and
\eqref{eq:balayagexonKc}
 for the balayage of point
masses in $[-{\rm i}c, {\rm i}c] \cup \R$ onto $K_{c}$, we can easily
find the bounds
\[ \frac{{\rm d} \Bal(\delta_{{\rm i}y}, K_{c})}{|{\rm d}z|} \leq \frac{1}{\pi} \frac{c}{|z|^2 - c^2},
    \qquad z \in K_{c}, \]
and
\[ \frac{{\rm d} \Bal(\delta_{x}, K_{c})}{|{\rm d}z|} \leq \frac{1}{\pi} \frac{\sqrt{c^2+x^2}}{|z|^2 - c^2},
    \qquad z \in K_{c}. \]
From this it follows that
\begin{equation} \label{eq:boundbalayagerho}
    \frac{{\rm d} \Bal(\rho, K_{c})}{|{\rm d}z|} \leq \frac{1}{\pi} \frac{\sqrt{c^2 + X^2}}{|z|^2 - c^2},
    \qquad z \in K_{c}. \end{equation}
for every probability measure $\rho$ on $[-X,X] \cup [-{\rm i}c, {\rm i}c]$.

Let $c_0$ be such that $\sigma([-{\rm i}c_0, {\rm i}c_0]) = 2/3$.
It follows from \eqref{eq:densitysigma} and \eqref{eq:boundbalayagerho}
that there exists a constant $L = L_X$ so that
\begin{equation} \label{eq:boundbalayagerho2}
    \min \left(\frac{{\rm d}\sigma}{|{\rm d}z|}, \frac{{\rm d} \Bal(\rho, K_{c_0})}{|{\rm d}z|}\right) \leq \frac{L}{|z|^p},
    \qquad z \in K_{c_0} \end{equation}
for every probability measure $\rho$ on $[-X,X] \cup [-{\rm i}c_0, {\rm i}c_0]$.
Since $C_p < 1$ we  can then find a constant $K$ large enough so that
\begin{equation} \label{eq:constantK}
    C_p^2 K + \frac{2}{3} L < K.
\end{equation}
This is how we choose $K$ and this determines the extra
constraint $K \sigma_{p}$.\\

%\Proof[Proof of existence of the minimizer]
We now claim that there exists a  vector of measures
$(\nu_1^*,\nu_2^*,\nu_3^*)$ that minimizes $E_V(\nu_1,\nu_2,\nu_3)$
subject to the conditions $(a)-(e)$ of the equilibrium problem of
Definition \ref{def:equilibriummeasure} and subject to the additional constraint
\begin{align}\label{eq:extracondition2}
 \nu_2 \leq K \sigma_{p}
 \end{align}
with $K$ satisfying \eqref{eq:constantK}.

We already found that $E_V$ is bounded from below. Therefore there
exists a sequence $(\nu_{1,n},\nu_{2,n},\nu_{3,n})$ such that
    \begin{align}
        E_V(\nu_{1,n},\nu_{2,n},\nu_{3,n})\leq \frac{1}{n}+\inf E_V
    \end{align}
where both the sequence and the infinimum are taken with respect to
conditions (a)-(e) of the equilibrium problem and \eqref{eq:extracondition2}.
By Lemma \ref{lemma:equilibriumnu3} we may as well assume that
\begin{align}
\nu_{3,n}=\frac{1}{2}\Bal(\nu_{2,n},\R).
\end{align}
Moreover, we take $\nu_{1,n}$ as the minimizer of $E(\nu_1,\nu_{2,n},\nu_{3,n})$
subject to condition (b). By \eqref{eq:extraconstraint} and Lemma \ref{lemma:estimateonnu3}
the sequences $\nu_{2,n}$ and $\nu_{3,n}$ are tight. Moreover,
by Lemma \ref{lemma:equilibriumnu1} and the choice of $X$ we have that
the measures $\nu_{1,n}$ are supported in $[-X,X]$.  Therefore there
exists a subsequence $(\nu_{1,k_n},\nu_{2,k_n},\nu_{3,k_n})$ that
converges weakly to a vector of measures $(\nu_{1}^*,\nu_2^*,\nu_3^*)$.
This is the minimizer that we seek for and hence we proved the claim.

Our final goal is to show that the extra upper constraint $K
\sigma_{p}$ is not active.

Given the measures $\nu_1^*$ and $\nu_3^*$ from the minimizer of
the equilibrium problem with extra constraint, we denote by
$\nu_2^{**}$ the measure that minimizes the energy functional
$E_V(\nu_1^*, \nu_2, \nu_3^*)$ subject to the constraint $\nu_2
\leq \sigma$ only. By Lemma \ref{lemma:equilibriumnu2} we then
have that  $\nu_2^{**}$ is equal to $\sigma$ on some interval
$[-{\rm i}c, {\rm i}c]$. We clearly have $c < c_0$. It also
follows from properties of balayage that
\begin{align} \nonumber
    \left. \nu_2^{**} \right|_{K_{c_0}} & \leq \frac{1}{2} \Bal(\nu_1^* + \nu_3^*, K_{c_0}) \\
    & \leq \frac{1}{2} \Bal(\nu_3^*, {\rm i}\mathbb R)
    + \Bal(\rho, K_{c_0})
    \label{eq:boundonnu2starstar}
\end{align}
where
\[ \rho = \frac{1}{2} \nu_1^* + \frac{1}{2} \left. \Bal(\nu_3^*, {\rm i}\R) \right|_{[-{\rm i}c_0,{\rm i}c_0]} \]
is a measure on $[-X,X] \cup [-{\rm i}c_0, {\rm i} c_0]$ with total mass $\leq 2/3$.
Recalling also that
\[ \nu_2^{**} \leq \sigma \]
we find by Corollary \ref{cor:estimateonBalnu3},
\eqref{eq:boundbalayagerho2}, and \eqref{eq:boundonnu2starstar}
that
\[ \frac{{\rm d}\nu_2^{**}}{|{\rm d}z|} \leq C_p^2 \frac{K}{|z|^p}  + \frac{2}{3} \frac{L}{|z|^p}
    \qquad z \in {\rm i}\R. \]
Then it follows from \eqref{eq:constantK} that
\begin{equation} \label{eq:strictbound}
    \frac{{\rm d}\nu_2^{**}}{|{\rm d}z|} <  \frac{K}{|z|^p}, \qquad z \in {\rm i}\R.
\end{equation}
Thus $\nu_2^{**}$ would also be the minimizer if we impose the
extra constraint $\nu_2 \leq K\sigma_{p}$. It follows that
\[ \nu_2^* = \nu_2^{**} \]
and since the inequality in \eqref{eq:strictbound} is strict,
it follows that $\nu_2^*$ does not feel the extra constraint.

Thus $(\nu_1^*, \nu_2^*, \nu_3^*)$ is also the minimizer of the
equilibrium problem of Definition \ref{def:equilibriummeasure}
(without extra constraint), and existence of the minimizer is now
proved.

\subsection{Proof of Theorem \ref{prop:structureEqmeasure}
 and proof of Proposition \ref{prop:variationalconditions}}

\begin{proof}
We have shown that a unique minimizer $(\mu_1, \mu_2, \mu_3)$ for
the equilibrium problem exists.

Proposition \ref{prop:variationalconditions}  and parts (a)-(c) of
Theorem \ref{prop:structureEqmeasure} follow by taking the
minimizer $(\mu_1,\mu_2,\mu_3)$ for the vector
$(\nu_1,\nu_2,\nu_3)$ in  Lemmas \ref{lemma:equilibriumnu2},
\ref{lemma:equilibriumnu3} and \ref{lemma:equilibriumnu1}.

Part (d) of Theorem \ref{prop:structureEqmeasure} follows from the
symmetry in the equilibrium problem (note that $V$ is even) and
the uniqueness of the minimizer.
\end{proof}

\subsection{Proof of Proposition \ref{Th:convexity}}

\begin{proof}
The measure $\mu_1$ is the unique minimizer of the energy
functional given in \eqref{eq:equilibriumfornu1}. Since  the
external field is symmetric, it follows from \cite[Theorem IV.1.10
(f)]{ST} that ${\rm d}\mu_1(t)={\rm d}\widetilde \mu(t^2)/2$ where
$\widetilde \mu$ is the unique minimizer of the energy functional
\begin{align}
\iint \log \frac{1}{|x-y|} {\rm d} \nu(x) {\rm d}\nu(y)
 + 2 \int \left( V(\sqrt{x})-\tau^{4/3} x^{2/3} -U^{\mu_2}(\sqrt{x}) \right) {\rm d} \nu(x),
\end{align}
among all probability measures $\nu$ with support on $[0,\infty)$.
From
\begin{align}
    -U^{\mu_2}(\sqrt{x})=-\frac{1}{2}\int \log(x+|z|^2) {\rm d}\mu_2(z)
\end{align}
it follows that $x\mapsto -U^{\mu_2}(\sqrt{x})$ is a convex
function on $[0,\infty)$. Also $x \mapsto -\tau^{4/3}x^{2/3}$ is
convex on $[0,\infty)$. Hence, since we assumed that $x\mapsto
V(\sqrt x)$ is convex as well, the external field in the
equilibrium problem for $\tilde \mu$ is convex. Thus the support
of $\tilde{\mu}$ consists of one interval $[a,b]$ with $0 \leq a <
b$, see e.g.~\cite[Theorem IV.1.10]{ST}. Hence the support of
$\mu_1$ is either one  or two intervals, depending on whether
$a=0$ or $a > 0$. \end{proof}

\section{A Riemann surface}

\subsection{A four-sheeted Riemann surface}
We use the solution of the equilibrium problem to construct a
Riemann surface as follows. Start with the four sheets $\mathcal
R_j$ defined as
\begin{align}\begin{array}{ll}
\mathcal R_1 = \C \setminus  S(\mu_1), & \mathcal R_2=\C\setminus (S(\mu_1) \cup S(\sigma-\mu_2)),\\
\mathcal R_3 = \C \setminus (S(\sigma-\mu_2) \cup S(\mu_3)), &
\mathcal R_4 = \C \setminus S(\mu_3). \end{array}
\end{align}
The four sheets are connected as follows: $\mathcal R_1$ is
connected to $\mathcal R_2$ via $S(\mu_1)$, $\mathcal R_2$ is
connected to $\mathcal R_3$ via $S(\sigma-\mu_2)$ and $\mathcal
R_3$ is connected to $\mathcal R_4$ via $S(\mu_3)$, every
connection is  in the usual crosswise manner.  See also Figure
\ref{RiemannS} for a picture of the Riemann surface in the case
where the support of $\mu_1$ consists of one interval.

The Riemann surface is compactified by adding two points at
infinity: one is on the first sheet, and the other point at
infinity is common to the other three sheets. The Riemann surface
has genus $N-1$ if $N$ is the number of intervals in the support
of $\mu_1$.

\begin{figure}[t]
\centering
  \begin{overpic}[scale=0.4]{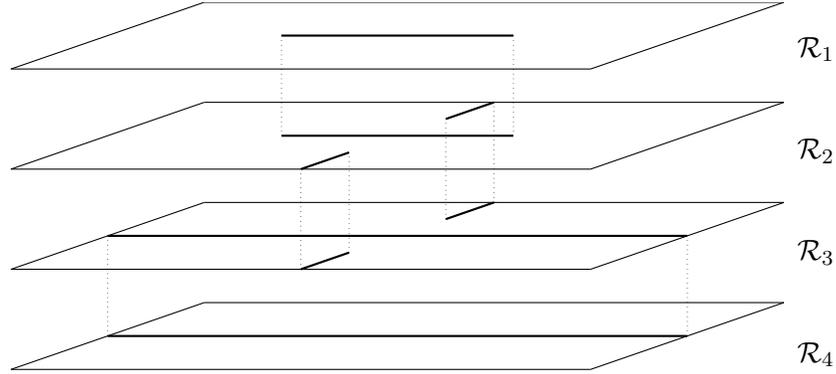}
  \put(90,32){\mbox{$\mathcal R_1$}}
   \put(90,22){\mbox{$\mathcal R_2$}}
    \put(90,12){\mbox{$\mathcal R_3$}}
     \put(90,2){\mbox{$\mathcal R_4$}}
  \end{overpic}
  \caption{The Riemann surface $\mathcal R$ for the one-cut case}

  \label{RiemannS}
\end{figure}

\subsection{The Cauchy transforms}

The Cauchy transform $F$ of a measure $\mu$  is defined as
\begin{equation}
  F(z)= \int \frac{1}{z-x} \ {\rm d}\mu(x), \qquad z\in  \C \setminus S(\mu).
\end{equation}
Note that we slightly abuse the notion of Cauchy transform since
\[ C\mu(z) = \frac{1}{2\pi {\rm i}} \int \frac{1}{x-z} d\mu(x)
    = -\frac{1}{2\pi {\rm i}} F(z) \]
is also called the Cauchy transform of $\mu$.

We are interested in particular in the Cauchy transforms of the
measures $\mu_1$, $\mu_2$ and $\mu_3$ that are the solution of the
equilibrium problem of Definition \ref{def:equilibriummeasure}. We
denote their Cauchy transforms by $F_1$, $F_2$, $F_3$,
respectively. We use the Cauchy transforms $F_j$ to construct a
meromorphic function on the Riemann surface $\mathcal R$. We
recall the Sokhotski-Plemelj formulas according to which
 \begin{align*}
    F_{j,+}(z) + F_{j,-}(z) & = 2 PV \int \frac{1}{z-x} d\mu_j(x), \\
    F_{j,+}(z) - F_{j,-}(z) & = -2 \pi {\rm i} \frac{d \mu_j}{dz},
    \end{align*}
    where PV denotes the Cauchy principal value.

\begin{lemma} \label{lem:analyticityxi}
The function $\xi: \bigcup_{j=1}^4 \mathcal R_j \to \C$ defined by
 \begin{align} \label{eq:defxi}
\xi(z)=\left\{\begin{array}{ll} V'(z) -F_1(z), & z\in \mathcal R_1,\\
   F_1(z) - F_2(z)+ \tau^{4/3} z^{1/3}, & z\in \mathcal R_2,\ \Re z > 0,\\
   F_1(z) - F_2(z)-\tau^{4/3} (-z)^{1/3}, & z\in \mathcal R_2, \ \Re z <0,  \\
   F_2 (z) - F_3(z) -\tau^{4/3} (-z)^{1/3},& z\in \mathcal R_3,\ \Re z  > 0,\\
   F_2 (z) - F_3(z)+ \tau^{4/3} z^{1/3}, & z\in \mathcal R_3,\ \Re z < 0,\\
      F_3(z)+
      {\rm e}^{4\pi{\rm i}/3}\tau^{4/3} z^{1/3},& z\in \mathcal R_4, \ \Im z > 0, \\
       F_3(z)+{\rm e}^{2\pi{\rm i}/3}\tau^{4/3} z^{1/3}, & z\in \mathcal R_4, \ \Im z < 0,
    \end{array}\right.
  \end{align}
has an extension to a meromorphic function (also denoted by $\xi$)
on $\mathcal R$. The meromorphic function has a pole of order
$\deg V - 1$ at infinity on the first sheet, and a simple pole at
the other point at infinity.
\end{lemma}
\begin{proof} The analyticity of $\xi$ on $\mathcal R_1$ and $\mathcal
R_4$ is immediate. To obtain the analyticity of $\xi$ on $\mathcal
R_2$ and $\mathcal R_3$ we note that from the Sokhotski-Plemelj
formula and from the fact that $\mu_2=\sigma$ on $[-{\rm i}c,{\rm
i}c]$ it follows that
\begin{align} \nonumber
    F_{2,+}(z)- F_{2,-}(z) & =
    - 2\pi{\rm i} \frac{{\rm d} \mu_2(z)}{{\rm d}z}
      =  -2\pi{\rm i} \frac{{\rm d} \sigma(z)}{{\rm d}z}\\
    & = - \tau^{4/3} \sqrt{3} |z|^{1/3} = -\tau^{4/3}(z^{1/3}+(-z)^{1/3})
\end{align}
for $z\in (-{\rm i}c,{\rm i}c)$. By the definition of $\xi$ this
implies that $\xi$ has an analytic continuation to $(-{\rm
i}c,{\rm i}c)$ on both $\mathcal R_2$ and $\mathcal R_3$. Thus
$\xi$ is analytic on each individual sheet $\mathcal R_j$.

It remains to check the analyticity if we cross a cut and move
from one sheet to another. That this holds, follows from the
variational conditions for the equilibrium problem and the
symmetry in the problem. We will only show  that $\xi$ is analytic
when one crosses $S(\mu_1)$ since the other cases follow from
similar arguments.

Differentiating $U^{\mu_1}$ gives
  \begin{align}
    2 \frac{{\rm d}}{{\rm d} x } U^{\mu_1}(x) = -2 PV \int \frac{1}{x-s}{\rm d}\mu_1(s),
    \qquad x \in S(\mu_1),
  \end{align}
  which by the Sokhotski-Plemelj formulas  can be written as
   \begin{align} \label{Sokhotski-Plemelj2}
   - 2 \frac{{\rm d}}{{\rm d} x } U^{\mu_1}(x)= F_{1,+}(x)+ F_{1,-}(x),
    \qquad x \in S(\mu_1).
  \end{align}
On the other hand, we can differentiate the right-hand side of the
variational condition \eqref{eq:variationalcondmu1a} to obtain
\begin{align} \label{diffUmu1}
2 \frac{{\rm d}}{{\rm d} x }& U^{\mu_1}(x)= \frac{{\rm d}}{{\rm d} x } U^{\mu_2}(x)-V'(x)+\tau |x|^{1/3} \sign x\\
\nonumber &=-\frac{{\rm d}}{{\rm d} x }\int \log|x-s| {\rm d}\mu_2(s) -V'(x)+\tau |x|^{1/3} \sign x\\
\nonumber  &=-\frac{{\rm d}}{{\rm d} x }\frac{1}{2}\int \log (x^2-s^2) {\rm d}\mu_2(s) -V'(x)+\tau |x|^{1/3} \sign x\\
\nonumber  &=-\int \frac{x}{x^2-s^2} {\rm d}\mu_2(s) -V'(x)+\tau |x|^{1/3} \sign x\\
\nonumber  &=-\frac{1}{2}\int \frac{1}{x-s}{\rm
d}\mu_2(s)-\frac{1}{2}\int\frac{1}{x+s} {\rm d}\mu_2(s)
-V'(x)+\tau |x|^{1/3} \sign x,
\end{align}
for $x\in S(\mu_1)$.  By symmetry of $\mu_2$ we obtain
\begin{align} \label{symmCauchymu2}
  \frac{1}{2}\int \frac{1}{x-s}{\rm d}\mu_2(s)+\frac{1}{2}\int\frac{1}{x+s} {\rm d}\mu_2(s)
  =\int \frac{1}{x-s} {\rm d} \mu_2(s)
  =F_2(x)
\end{align}
for $x\in \R$. Combining \eqref{symmCauchymu2}, \eqref{diffUmu1} and
\eqref{Sokhotski-Plemelj2}  leads to
\begin{align}
  F_{1,+}(x)+ F_{1,-}(x) = F_2(x) + V'(x) - \tau |x|^{1/3} \sign x,
\end{align}
for $x\in S(\mu_1)$. This is exactly what is needed to prove that
$\xi$ is analytic when one crosses $S(\mu_1)$ and moves from
$\mathcal R_1$ to $\mathcal R_2$. \end{proof}

We use $\xi_j$ to denote the restriction of $\xi$ to $\mathcal
R_j$.

As a corollary to the lemma we find the following result about the
asymptotic behavior of the Cauchy transforms.
\begin{corollary} \label{cor:asymptoticsCmuj}
    There exists a constant $\alpha$ such that (with $\omega={\rm e}^{2\pi{\rm i}/3}$)
\begin{align}
F_1(z)&=\frac{1}{z} +\OO(z^{-2}),\\
F_2(z)&=\begin{cases} \frac{2}{3z} -
\frac{\alpha}{z^{5/3}} + \OO(z^{-2}), &\Re z >0,\\
\frac{2}{3z} -\frac{\omega \alpha}{z^{5/3}} +\OO(z^{-2}), &\Re z<0 , \ \Im z>0,\\
\frac{2}{3z} -\frac{\omega^2 \alpha}{z^{5/3}} +\OO(z^{-2}),  & \Re
z<0,\ \Im z<0,
\end{cases}\\
F_3(z)&=
\begin{cases} \frac{1}{3z} +
\frac{\omega^2 \alpha}{z^{5/3}} +\OO(z^{-2}),  & \Im z>0,\\
\frac{1}{3z}+\frac{\omega \alpha}{z^{5/3}} +\OO(z^{-2}), & \Im z<0,
\end{cases}
\end{align}
uniformly for $z\to \infty$.
\end{corollary}

\begin{proof} The statement about $F_1$ is obvious, since $\mu_1$ has
a compact support.

For $z$ large enough define $\Xi$ as
\[ \Xi(z) = \xi_2(z^3) \qquad - \frac{\pi}{6} < \arg z < \frac{\pi}{6}, \]
and then continued analytically into other sectors.  Explicit
calculations  show that
\begin{align}
\Xi(z)=\begin{cases}
\xi_2(z^3), & \begin{cases} -\pi/6 \leq \arg z \leq \pi/6,\\
5\pi/6 \leq \arg z \leq \pi,\\
-\pi \leq \arg z \leq -5\pi/6,
\end{cases} \\
\xi_3(z^3), & \begin{cases} \pi/6\leq \arg z \leq \pi/3,\\
-\pi/3 \leq \arg z \leq -\pi/6,\\
2\pi/3\leq \arg z \leq 5\pi/6,\\
-5\pi/6 \leq \arg z \leq -2\pi/3,
\end{cases} \\
\xi_4(z^3), & \begin{cases} \pi/3\leq \arg z \leq 2\pi/3,\\
-2\pi/3 \leq \arg z \leq -\pi/3.
\end{cases} \\
\end{cases}
\end{align}
Then $\Xi$ is analytic for $|z|>R$ for $R$ large enough with a
Laurent expansion
\begin{equation}
    \Xi(z)=\tau^{4/3} z+\sum_{j=0}^\infty c_j z^{-j}.
\end{equation}
Since $\xi_j(-z)=-\xi_j(z)$ we also have $\Xi(-z)=-\Xi(z)$, and
hence the even powers vanish in this expansion.

Next we claim that we can reconstruct the $\xi_j$, $j=2,3,4$  out
of $\Xi$ again in the following way
\begin{align}
\xi_2(z)&=\begin{cases}
\Xi(z^{1/3} ), &\Re z >0,\\
\Xi(z^{1/3} \omega), &\Re z<0 , \ \Im z>0,\\
\Xi(z^{1/3} \omega^2),  & \Re z<0,\ \Im z<0,
\end{cases}\\
\xi_3(z)&=\begin{cases}
\Xi(z^{1/3}), &\Re z <0,\\
\Xi(z^{1/3} \omega), &\Re z>0 , \ \Im z>0,\\
\Xi(z^{1/3} \omega^2),  & \Re z>0,\ \Im z<0,
\end{cases}\\
\xi_4(z)&=\begin{cases}
\Xi(z^{1/3} \omega^2), &\Im z>0,\\
\Xi(z^{1/3}\omega), & \Im z<0.
\end{cases}
\end{align}
We leave the verification to the reader. Since we have that by \eqref{eq:defxi}
\begin{align}
F_2(z)= -\xi_1(z)-\xi_2(z) + V'(z)+ \tau^{4/3} z^{1/3}
\end{align}
for $\Re z>0$, we obtain
\begin{align}
    F_2(z)=-\frac{c_1}{z^{1/3}}+\frac{1-c_3}{z} -\frac{c_5}{z^{5/3}}+\OO(z^{-1}),
\end{align}
for $z\to \infty$ and $\Re z>0$.  Since $F_2(z)=\int \frac{1}{z-x}
{\rm d}\mu_2(x)$ and since $\mu_2$ has total mass $2/3$, we see
that $c_1=0$ and $c_3=1/3$.  This proves the asymptotic behavior
of $F_2$ with $\alpha=c_5$ for $\Re z>0$. The asymptotics for the
other sectors and other Cauchy transforms follow from similar
arguments.
\end{proof}

The solution to the equilibrium problem can be described in terms
of the functions $\xi_j$ as described in the following
proposition.

\begin{proposition}\label{prop:EqmeasureXi}
With $\xi_j=\xi|_{\mathcal R_j}$, we have that
           \begin{align}\label{eq:mu1inxi}
    {\rm d} \mu_1(z)&=\frac{1}{2\pi {\rm i}}
    \left(\xi_{1,+}(z)-\xi_{1,-}(z)\right) {\rm d}z, \qquad z \in \R, \\
    \label{eq:mu2inxi}
    {\rm d} \mu_2(z)&=\frac{1}{2\pi {\rm i}}
    \left(\xi_{2,+}(z)-\xi_{2,-}(z)\right) {\rm d}z+ d \sigma(z), \qquad z \in {\rm i}\R, \\
     {\rm d} \mu_3(z)&=\frac{1}{2\pi {\rm i}}
    \left(\xi_{3,+}(z)-\xi_{3,-}(z)\right) {\rm d}z-
    \frac{\tau^{4/3} \sqrt{3}}{2\pi} |z|^{1/3} {\rm d}z, \qquad z \in \R. \label{eq:mu3inxi}
  \end{align}
\end{proposition}
\begin{proof} This follows by Lemma \ref{lem:analyticityxi} and the
Sokhotski-Plemelj formula
\begin{align}
    -\frac{1}{2\pi {\rm i}} \left( F_{j,+}(z)- F_{j,-}(z) \right)
    =\frac{{\rm d}}{{\rm d}z} \mu_{j}(z)
\end{align}
for $z\in S(\mu_j)$. \end{proof}

\section{The second transformation $X \mapsto U$}

From now on, we assume that $\mu_1$ is one-cut regular. That is,
$\mu_1$ is supported on one interval, and there are no singular
points in the sense of Definition \ref{def:regular}. We write
\[ S(\mu_1) = [-a,a] \]
with $a > 0$.

\subsection{Definition of the transformation}

In the second transformation in the steepest descent analysis, we
use the $g$-functions associated with the measures $\mu_j$ that
minimize the energy functional $E_V$.  For $j=1$ and $j=3$, we
define $g_j$ as
\begin{equation}
    g_j(z)=\int \log(z-s) \ {\rm d} \mu_j(s), \qquad j=1,3,
\end{equation}
with the principal branch of the logarithm
 \[ \log(z-s) = \log|z-s| + {\rm i} \arg (z-s), \qquad \arg (z-s) \in (-\pi, \pi). \]
We also define $g_2$ by a similar integral, but since $\mu_2$ is
supported on the imaginary axis, we make a different choice for
the branch of the logarithm. We define $g_2$ by
\begin{align} \label{eq:defg2}
  g_2(z)=\int \log(z-s) \ {\rm d}\mu_2(s).
\end{align}
where we take the logarithm such that
\begin{align} \label{eq:logbranchg2}
\log(z-s)=\log|z-s| +{\rm i}\arg(z-s), \quad \arg(z-s) \in
(-\pi/2, 3\pi/2),
\end{align}
for $z\in \C\setminus {\rm i}\R$ and $s \in {\rm i}\R$. Thus $g_2$
is defined and analytic in $\C \setminus {\rm i} \R$.

The transformation $X \mapsto U$ is now defined by
\begin{equation}\label{eq:XtoU}
    U(z)=L^{-n} X(z) G(z)^n L^n,
\end{equation}
where $G$ is given by
\begin{align} \label{eq:defG}
    G(z)= \diag \begin{pmatrix}
    {\rm e}^{-g_1(z)} & {\rm e}^{g_1(z)-g_2(z)} & {\rm e}^{g_2(z)-g_3(z)} & {\rm e}^{g_3(z)}
\end{pmatrix}
\end{align}
and $L$ is given by
\begin{equation}\label{eq:defL}
L= \diag \begin{pmatrix} {\rm e}^{-\ell} & 1 & 1 & 1
\end{pmatrix},
\end{equation}
with $\ell$  the variational constant from
\eqref{eq:variationalcondmu1a}.

Since $X$ and the $g_j$ are analytic in $\C\setminus (\R\cup {\rm
i}\R)$, it is clear that $U$ is also analytic in $\C\setminus
(\R\cup {\rm i}\R)$. The jump matrices $J_U$ in the RH problem for
$U$ are obtained from the jump matrices $J_{X}$ by
 \begin{equation} \label{eq:JU}
    J_{U} = L^{-n} G_-^{-n} J_{X} G_+^n L^n,
    \end{equation}
 and we have to calculate them on
the various parts of the real and imaginary axis. We will see that
the jump matrices take a simple form that is suitable for further
analysis.

Also the behavior at infinity simplifies under the transformation.
However, in contrast to other works (see e.g.\
\cite{DKMVZuniform}), the RH problem is not normalized at infinity
after the use of the $g$-functions. So $U(z)$ does not tend to the
identity matrix as $z \to \infty$, but instead there is a more
complicated behavior at infinity, which however does not depend on
$n$ anymore.

Before we state the RH problem that is satisfied by $U$, we first
collect the properties of the functions $g_j$ that will be needed.

\subsection{The $g$-functions}

Observe that  the functions $g_j$ are anti-derivatives of the
Cauchy transforms $F_j$, i.e.,
 \begin{align} \label{eq:fromgtoCmu}
 g_j'(z)=F_j(z).
 \end{align}
The asymptotic behavior of $g_j$ can be obtained by integrating
the asymptotic behavior of $F_j$ as given in Corollary
\ref{cor:asymptoticsCmuj}.
\begin{lemma}\label{lem:asymptoticsg}
With the constant $\alpha$ as in Corollary {\rm
\ref{cor:asymptoticsCmuj}}, we have that
\begin{align}
g_1(z)&=\log z  +\OO(z^{-2}), \label{eq:asymptoticsg1}\\
g_2(z)&=\begin{cases}\frac{2}{3}\log z -
\frac{3\alpha}{2z^{2/3}} +\OO(z^{-1}), &\Re z >0, \\
\frac{2}{3}\log z-\frac{3\omega \alpha}{2z^{2/3}} +\OO(z^{-1}), &\Re z<0 , \ \Im z>0,\\
\frac{2}{3}\log z + \frac{4\pi{\rm i}}{3} -\frac{3\omega^2
\alpha}{2z^{2/3}} +\OO(z^{-1}), & \Re z<0,\ \Im z<0,
\end{cases}\label{eq:asymptoticsg2} \\
g_3(z)&=
\begin{cases}\frac{1}{3}\log z+
\frac{3\omega^2 \alpha}{2z^{2/3}} +\OO(z^{-1}),  & \Im z>0,\\
\frac{1}{3}\log z +\frac{3\omega \alpha }{2z^{2/3}} +\OO(z^{-1}),
& \Im z<0,
\end{cases} \label{eq:asymptoticsg3}
\end{align}
as $z\to \infty$ uniformly. In
\eqref{eq:asymptoticsg1}--\eqref{eq:asymptoticsg3} every $\log z$ is
defined with the principal branch.
\end{lemma}
\begin{proof} The asymptotics \eqref{eq:asymptoticsg1} for $g_1$
follows from the fact that $\mu_1$ is a symmetric probability
measure with compact support.

The asymptotics \eqref{eq:asymptoticsg2}--\eqref{eq:asymptoticsg3}
for $g_2$ and $g_3$ can be found by integrating the asymptotic
results for $F_2$ and $F_3$ as given in Corollary
\ref{cor:asymptoticsCmuj}. The constants of integration vanish,
which can be seen from the fact that
\[ \int \log(z-s) d\mu_j(s) = \log(z) \| \mu_j \| + \int
 \log(1-s/z) d\mu_j(s) \]
and $\log(1-s/z) d\mu_j(s) \to 0$ as $z \to \infty$.

The additional constant $\frac{4\pi {\rm i}}{3}$ in
\eqref{eq:asymptoticsg2} in the third quadrant comes from the fact
that we used the branch \eqref{eq:logbranchg2} of the logarithm to
define $g_2$ which does not correspond to the principal branch of
$\log z$ as $z \to \infty$ in the third quadrant. \end{proof}

\subsection{The $\phi$-functions}

In the statement of the RH problem for $U$ it turns out to be
convenient to use certain functions $\phi_j$. We recall that
$\xi_j$ is defined as the restriction of $\xi$ to the sheet
$\mathcal R_j$ of the Riemann surface. We define
\begin{align}
    \phi_j:\C\setminus(\R\cup \big(-{\rm i}\infty,-{\rm i}c]\cup[{\rm i}c,{\rm i}\infty)\big)\to \C,
\end{align}
as
\begin{align}\label{eq:defphi1}
\phi_1(z)&= - \frac{1}{2}\int_{a}^z  \left(\xi_1(y)-\xi_2(y)\right) \ {\rm d}y, \\
\phi_2(z)&=\begin{cases}
-\frac{1}{2}\int_{{\rm i}c}^z \left(\xi_{2}(y)-\xi_3(y)\right) \ {\rm d} y, & \Im z>0, \\
-\frac{1}{2}\int_{-{\rm i}c}^z \left(\xi_{2}(y)-\xi_{3}(y) \right) \ {\rm d}
y, & \Im z<0,
\end{cases}\label{eq:defphi2} \\
\phi_{3}(z)&= \frac{\pi {\rm i}}{6}-\frac{1}{2}\int_0^z
\left(\xi_3(y)-\xi_4(y)\right)\ {\rm d} y. \label{eq:defphi3}
\end{align}
All paths of integration lie entirely (except for their starting
values) in $\mathbb C \setminus (\mathbb R \cup \big(-{\rm
i}\infty,-{\rm i}c]\cup[{\rm i}c,{\rm i}\infty)\big)$.

The importance of the $\phi$-functions is that jump properties of
the function $g_j$ can be conveniently written in terms of the
function $\phi_j$. We start with the jump properties of $g_1$ and
the connection with $\phi_1$.

\begin{lemma} \label{lem:propertiesg1}
\begin{enumerate}
\item[(a)] We have for $z \in [-a,a]$,
    \begin{align} \label{eq:g1+g1-phi1}
    g_{1,+}(z) - g_{1,-}(z) = 2 \phi_{1,+}(z) = - 2\phi_{1,-}(z),
    \end{align} and
    \begin{multline} \label{eq:variationalcondmu1complex}
    g_{1,+}(z) + g_{1,-}(z) - g_2(z) - V(z) + \frac{3}{4} \tau^{4/3}
    |z|^{4/3} + \ell \\
    = \begin{cases} 0, & z \in (0,a], \\
     -2 \pi {\rm i}/3, & z \in [-a,0).
    \end{cases}
    \end{multline}

\item[(b)] For $z \in \R \setminus [-a,a]$, we have
    \begin{align}\label{eq:onecut}
        g_{1,+}(z)-g_{1,-}(z)= \begin{cases}
            0, & z \in (a,\infty), \\
        2\pi {\rm i}, & z \in (-\infty, -a).
        \end{cases}
    \end{align}
    \begin{multline} \label{eq:variationalphi1}
    g_{1,+}(z) + g_{1,-}(z) - g_2(z) - V(z) +
    \frac{3}{4}\tau^{4/3} |z|^{4/3} + \ell \\
    = \begin{cases}
        2 \phi_{1}(z), & z \in (a, \infty), \\
        2 \phi_{1,+}(z) -8 \pi {\rm i}/3, & z \in (-\infty, -a).
        \end{cases} \end{multline}
 \end{enumerate}
\end{lemma}

\begin{proof} (a)
Let $z \in [-a,a]$. Then we have
\begin{align} \label{eq:g1jump}
g_{1,+}(z)-g_{1,-}(z)&={\rm i} \int \left(\arg_+(z-s)-\arg_-(z-s)\right) \ {\rm d} \mu_1(s)\\
\nonumber &=2\pi {\rm i} \int_z^{a} \ {\rm d}\mu_1(s).
\end{align}
By \eqref{eq:mu1inxi}, \eqref{eq:defphi1}, and the fact that
$\xi_{1,-}(z)=\xi_{2,+}(z)$ we obtain
\begin{align}
    g_{1,+}(z)-g_{1,-}(z) & = \int_z^{a} \left(\xi_{1,+}(y)-\xi_{1,-}(y)\right) \ {\rm d} y\\
    \nonumber & =\int_z^{a} \left(\xi_{1,+}(y)-\xi_{2,+}(y)\right) \ {\rm d}y \\
    \nonumber & =2\phi_{1,+}(z)=-2\phi_{1,-}(z)
\end{align}
which proves \eqref{eq:g1+g1-phi1}.

For $z\in [-a,a] = S(\mu_1)$, it follows from the variational
condition \eqref{eq:variationalcondmu1a} that
\begin{align}
    \Re \left(g_{1,+}(z)+g_{1,-}(z)-g_2(z)-V(z)+\frac{3}{4} \tau^{4/3} |z|^{4/3}+\ell\right)=0
\end{align}
and hence
\begin{multline}\label{eq:variationalcondmu1complexstep1}
    g_{1,+}(z)+g_{1,-}(z)-g_2(z)-V(z)+\frac{3}{4} \tau^{4/3} |z|^{4/3}+\ell\\
    ={\rm i} \int \big(\arg_+(z-s)+\arg_-(z-s)\big) {\rm d}
    \mu_1(s) -{\rm i} \int_{{\rm i}\R} \arg(z-s) \ {\rm d}\mu_2(s),
\end{multline}
where by the choice of branches in the definition of $g_1$ and $g_2$,
we have to use  $-\pi < \arg(z-s) <
\pi$ in the first integral and $-\pi/2 < \arg(z-s) < 3\pi/2$
in the second integral.
Clearly,
\begin{align}\label{eq:variationalcondmu1complexstep2}
\arg_+(z-s)+\arg_-(z-s)=0
\end{align}
for $z,s\in \R$. From the  symmetry of $\mu_2$ with respect to the
real axis,  we obtain
\begin{equation} \label{eq:conditiong2}
\int_{{\rm i}\R} \arg(z-s) \ {\rm d}\mu_2(s)=\int_0^{{\rm
i}\infty} \left(\arg(z-s)+\arg(z+s)\right) {\rm d}\mu_2(s).
\end{equation}
Since by \eqref{eq:logbranchg2}
\begin{equation}
 \arg(z-s)+\arg(z+s) = \left\{ \begin{array}{cc}
 0 & \textrm{if }  z > 0, \\
 2\pi & \textrm{if } z < 0, \end{array} \right.
\end{equation}
for $s\in (0,{\rm i}\infty)$, and since $\mu_2({\rm i}\R_+) = 1/3$, we
get from \eqref{eq:conditiong2}
\begin{align}
\int_{{\rm i}\R} \arg(z-s) \ {\rm d}\mu_2(s) = \left\{
\begin{array}{cc}
 0 & \textrm{if }  z > 0, \\
 2\pi/3 & \textrm{if } z < 0. \end{array} \right.
\label{eq:variationalcondmu1complexstep3}
\end{align}
Inserting \eqref{eq:variationalcondmu1complexstep3} and
\eqref{eq:variationalcondmu1complexstep2} into
\eqref{eq:variationalcondmu1complexstep1} leads to
\eqref{eq:variationalcondmu1complex}.

\medskip

(b) Let $z \in \R \setminus [-a,a]$. Then we continue to have
\eqref{eq:g1jump}, and from this the relations \eqref{eq:onecut}
immediately follow. By \eqref{eq:fromgtoCmu} and \eqref{eq:defxi}
we further have
\begin{align}\label{eq:diffcomplexvariationalcondmu1}
 \frac{{\rm d}}{{\rm d} z} \left( g_{1,+}(z)+g_{1,-}(z)-g_2(z)-V(z)+
    \frac{3}{4}\tau^{4/3}|z|^{4/3}+\ell \right)=\xi_2(z)-\xi_1(z).
\end{align}
We obtain \eqref{eq:variationalphi1} from \eqref{eq:defphi1},
\eqref{eq:diffcomplexvariationalcondmu1}, and the fact that
equality holds in \eqref{eq:variationalcondmu1complex} for $z= \pm
a$. Note also that $\phi_{1,+}(-a) = \pi {\rm i}$.
\end{proof}

For the functions $g_2$ and $\phi_2$ we have in a similar way.

\begin{lemma} \label{lem:propertiesg2}
\begin{enumerate}
\item[(a)] For  $z \in S(\sigma-\mu_2)$, we have
\begin{align}\label{eq:variationalcondmu2complex}
    g_{2,+}(z)+g_{2,-}(z)-g_1(z)-g_3(z) =
        \begin{cases}
        0, & \textrm{if } z \in [{\rm i}c, {\rm i}\infty), \\
        4\pi {\rm i}/3, & \textrm{if } z \in (-{\rm i}\infty, -{\rm i}c],
        \end{cases}
\end{align}
and for $z \in S(\sigma-\mu_2) \cap {\rm i} \R_{\pm}$,
\begin{multline}\label{eq:middleblockentrydiag}
    g_{2,+}(z)-g_{2,-}(z)\pm\frac{{\rm i} 3\sqrt{3}}{4} \tau^{4/3}
    |z|^{4/3} \\ =
    2\phi_{2,+}(z) + 2\pi {\rm i}/3
    =  -2\phi_{2,-}(z) + 2 \pi {\rm i}/3.
\end{multline}

\item[(b)] For $z\in {\rm i}\R \setminus S(\sigma-\mu_2) = (-{\rm
i}c, {\rm i}c)$, we have
\begin{align} \label{eq:conditionong2}
    g_{2,+}(z)-g_{2,-}(z)\pm\frac{{\rm i} 3\sqrt{3}}{4} \tau^{4/3}
    |z|^{4/3} = 2\pi {\rm i}/3, \qquad z \in {\rm i} \R_{\pm},
    \end{align}
and
\begin{align} \label{eq:conditionong2b}
 g_{2,+}(z)+g_{2,-}(z)-g_1(z)-g_3(z) =
    \begin{cases}
    2 \phi_2(z) & \textrm{if } z \in (0,{\rm i} c), \\
    2 \phi_2(z) + 4 \pi {\rm i}/3 & \textrm{if } z\in (-{\rm i}c, 0).
    \end{cases}
 \end{align}
\end{enumerate}
\end{lemma}

\begin{proof}
(a) Let $z \in S(\sigma-\mu_2)$.
From the variational condition \eqref{eq:variationalcondmu2a} it
then follows that
\begin{equation}
\Re \left(g_{2,+}(z)+g_{2,-}(z)-g_1(z)-g_3(z)\right)= 0.
\end{equation}
 Hence
\begin{multline} \label{eq:variationalcondmu2complexstep1}
    g_{2,+}(z) +g_{2,-}(z)-g_1(z)-g_3(z)\\
    = {\rm i} \int \left(\arg_+(z-s)+\arg_-(z-s)\right)  {\rm d} \mu_2(s) \\
    - {\rm i} \int \left(\arg(z-s)\right) {\rm d} (\mu_1+\mu_3)(s).
\end{multline}
 By symmetry we find
\begin{multline}
    \int \arg(z-s) {\rm d} (\mu_1+\mu_3)(s)
    =\int_0^\infty \left(\arg(z-s)+\arg(z+s)\right) {\rm d}
    (\mu_1+\mu_3)(s).
\end{multline}
Since $\arg(z-s) + \arg(z+s) = \pm \pi$ for $z \in {\rm i} \R_{\pm}$ and $s \in \R_+$,
and since $(\mu_1 + \mu_3)(\R_+) = 2 \pi {\rm i}/3$,
we then obtain
\begin{align}\label{eq:variationalcondmu2complexstep2}
    \int \arg(z-s) {\rm d} (\mu_1+\mu_3)(s)
    = \pm 2 \pi/3, \qquad
        z \in {\rm i} \R_{\pm}.
\end{align}
 Moreover,
\begin{align}\label{eq:variationalcondmu2complexstep3}
\int \left(\arg_+(z-s)+\arg_-(z-s)\right)  {\rm d}
\mu_2(s)=2\pi/3, \qquad z \in {\rm i}\R.
\end{align}
Inserting \eqref{eq:variationalcondmu2complexstep2} and
\eqref{eq:variationalcondmu2complexstep3} into
\eqref{eq:variationalcondmu2complexstep1} leads to
\eqref{eq:variationalcondmu2complex}.

By symmetry we have $\mu_2({\rm i} \R_+) = 1/3$. Then by the
definition of $g_2$ and \eqref{eq:mu2inxi} we obtain
\begin{align} \nonumber
    g_{2,+}(z)-g_{2,-}(z) & =
    2\pi{\rm i}\int_z^{{\rm i}\infty} {\rm d}\mu_2(y)
    = \frac{2\pi{\rm i}}{3}-2\pi{\rm i} \int_0^z \ {\rm d}\mu_2(y)\\
    \label{eq:diffg2}
    & = \frac{2\pi{\rm i}}{3} - \int_0^z \left(\xi_{2,+}(y)-\xi_{2,-}(y)
        \right) {\rm d} y - 2\pi {\rm i} \int_0^z {\rm d} \sigma(y).
        \end{align}
By \eqref{eq:densitysigma} we have
\begin{equation} \label{eq:integralofsigma}
    2 \pi {\rm i} \int_0^z {\rm d} \sigma(y) =
    \pm \frac{{\rm i} 3\sqrt{3}}{4} \tau^{4/3} |z|^{4/3},
    \quad \textrm{for } z \in {\rm i}\R_{\pm}.
    \end{equation}
For $z \in S(\sigma-\mu_2) \cap {\rm i}\R_{\pm}$, we also find
since $\xi_2$ is analytic in $(-{\rm i}c,0) \cup (0,{\rm i}c)$
\begin{align}  \nonumber
     \int_0^z \left(\xi_{2,+}(y)-\xi_{2,-}(y) \right) {\rm d}y
    & = \int_{\pm {\rm i}c}^z \left(\xi_{2,+}(y)-\xi_{2,-}(y) \right) {\rm d}y \\
    \nonumber
    & = \int_{\pm {\rm i}c}^z \left(\xi_{2,+}(y)-\xi_{3,+}(y) \right) {\rm d}y \\
    \label{eq:2phi2integral}
    & = -2 \phi_{2,+}(z) = 2 \phi_{2,-}(z),
    \end{align}
where we used the definition \eqref{eq:defphi2} of $\phi_2$.
Putting  \eqref{eq:integralofsigma} and
\eqref{eq:2phi2integral} into \eqref{eq:diffg2} leads to \eqref{eq:middleblockentrydiag}.

\medskip

(b) Let $z\in (-{\rm i}c ,{\rm i}c)$.  Then
as before
\begin{align} \label{eq:integralofmu2-1}
    g_{2,+}(z)-g_{2,-}(z)&=
    2\pi{\rm i}\int_z^{{\rm i}\infty} {\rm d}\mu_2(y)
    =
    \frac{2\pi{\rm i}}{3}-2\pi{\rm i} \int_0^z {\rm d}\mu_2(y).
    \end{align}
Now we have
\[ \int_0^z {\rm d} \mu_2(y) =
    \int_0^z {\rm d}\sigma(y), \]
    and so by \eqref{eq:integralofsigma}
\begin{align} \label{eq:integralofmu2-2}
     2 \pi {\rm i} \int_0^z {\rm d} \mu_2(y) =
    \pm \frac{{\rm i} 3\sqrt{3}}{4} \tau^{4/3} |z|^{4/3},
    \quad \textrm{for } z \in (-{\rm i}c, {\rm i}c) \cap {\rm i}\R_{\pm}.
    \end{align}
Then \eqref{eq:conditionong2} follows by \eqref{eq:integralofmu2-1}
and \eqref{eq:integralofmu2-2}.

Next, from \eqref{eq:fromgtoCmu} and \eqref{eq:defxi} we
obtain
\begin{multline} \label{eq:diffcomplexvariationalcondmu2}
\frac{{\rm d}}{{\rm d}z} \left(g_{2,+}(z)+g_{2,-}(z)-g_{1}(z)-g_3(z)\right) \\ =
    F_{2,+}(z) + F_{2,-}(z) - F_1(z) - F_3(z) = -\xi_2(z)+\xi_3(z).
\end{multline}
Integrating \eqref{eq:diffcomplexvariationalcondmu2} and noting
that the equality \eqref{eq:variationalcondmu2complex}
holds for $z = \pm {\rm i}c$,
 we obtain \eqref{eq:conditionong2b}
by \eqref{eq:defphi2}.
\end{proof}

Finally, for $g_3$ and $\phi_3$ we have the following results.
\begin{lemma} \label{lem:propertiesg3}
We have for $z \in \R$,
\begin{align} \label{eq:conditiong3}
    g_{3,+}(z)+g_{3,-}(z)-g_2(z) =
        \begin{cases} 0, &  z > 0, \\
            -2 \pi {\rm i}/3, &  z < 0.
            \end{cases}
\end{align}
and
\begin{multline} \label{eq:lowerrightblockentrydiag}
    g_{3,+}(z)-g_{3,-}(z) \mp \frac{{\rm i}3\sqrt{3}}{4}\tau^{4/3}
    |z|^{1/3} \\
    = 2\phi_{3,+}(z) = - 2\phi_{3,-}(z)
        + 2 \pi {\rm i}/3, \quad \textrm{for } z \in \R_{\pm}.
\end{multline}
\end{lemma}

\begin{proof}
From the variational condition \eqref{eq:variationalcondmu3} it
follows that for $z \in \R$,
\begin{align}
    \Re\left(g_{3,+}(z)+g_{3,-}(z)-g_2(z)\right)=0,
\end{align}
and hence
\begin{multline*}
g_{3,+}(z)+g_{3,-}(z)-g_2(z) ={\rm i} \int
\big(\arg_+(z-s)+\arg_-(z-s)\big) \ {\rm d} \mu_3(s)\\
-{\rm i} \int \arg(z-s) \ {\rm d} \mu_2(s).
\end{multline*}
Combining this with \eqref{eq:variationalcondmu1complexstep2} and
\eqref{eq:variationalcondmu1complexstep3} gives
\eqref{eq:conditiong3}.

Further, by symmetry we find $\mu_3(\R_+) = 1/6$ and so
\begin{align}
g_{3,+}(z)-g_{3,-}(z) & =2\pi{\rm i} \int_z^\infty {\rm d} \mu_3(s)
= \frac{\pi {\rm i}}{3}  -2\pi{\rm i}\int_0^z {\rm d}\mu_3(s).
\end{align}
Now inserting \eqref{eq:mu3inxi}  and using
$\xi_{3,-}(z)=\xi_{4,+}(z)$ gives
\begin{align*}
    g_{3,+}(z)-g_{3,-}(z)&=\frac{\pi {\rm i}}{3}
     -\int_0^z \left(\xi_{3,+}(y)-\xi_{3,-}(y)-{\rm i}\sqrt{3}\tau^{4/3} |y|^{1/3}\right){\rm d}y\\
    \nonumber &=\frac{\pi {\rm i}}{3}  -\int_0^z
    \left(\xi_{3,+}(y)-\xi_{4,+}(y) \right){\rm d}y \pm\frac{{\rm i}3\sqrt{3}}{4}\tau^{4/3} |z|^{1/3}
\end{align*}
for $z\in \R_{\pm}$. Combining this with \eqref{eq:defphi3} leads
to \eqref{eq:lowerrightblockentrydiag}.
\end{proof}

\subsection{The RH problem for $U$}

In the following lemma we state the RH problem that is satisfied
by $U$.

\begin{lemma} \label{lem:RHforU}
We have that $U$ satisfies the RH problem
\begin{equation} \label{eq:RHforU}
\left\{
    \begin{array}{l}
    \multicolumn{1}{l}{U \textrm{ is analytic in }\C\setminus(\R \cup {\rm i}\R)},\\
    U_+(z)=U_-(z) J_{U}(z),\qquad  z\in \R\cup {\rm i}\R, \\
U(z)=(I+\OO(z^{-1/3}))\begin{pmatrix}    1 & 0 & 0 &0\\
    0 & z^{1/3} & 0 & 0\\
    0 & 0 & 1  & 0 \\
    0 & 0 & 0 & z^{-1/3}
    \end{pmatrix} \begin{pmatrix} 1 & 0 \\
    0 & A_j
    \end{pmatrix}\\
    \multicolumn{1}{r}{\textrm{as } z\to \infty \textrm{ in the }j\textrm{th quadrant}}
    \end{array}
    \right.,
\end{equation}
where the $3 \times 3$ matrices $A_j$ are given in \eqref{eq:A1A2}
and \eqref{eq:A3A4}. The asymptotics for $U$ are uniform as $z \to
\infty$ in any region such that \eqref{eq:epsilonarea} holds for
some $\eps > 0$.

The jump matrix $J_{U}$ is given by
\begin{align}
J_{U}&=\begin{pmatrix} \label{eq:JumpUSmu1}
    {\rm e}^{-2 n\phi_{1,+}} & 1 &0 &0\\
    0 & {\rm e}^{- 2 n \phi_{1,-}} &0 &0\\
    0 & 0 & {\rm e}^{-2n\phi_{3,+}} & 1 \\
    0 & 0 & 0 &{\rm e}^{-2n\phi_{3,-}}
    \end{pmatrix}, \quad  \textrm{on }  S(\mu_1),
\end{align}
\begin{align}
J_{U}&=\begin{pmatrix}
    1 & {\rm e}^{2n \phi_{1,+(x)}} &0 &0\\
    0 & 1  &0 &0\\
    0 & 0 & {\rm e}^{-2n\phi_{3,+}} & 1 \\
    0 & 0 & 0 &{\rm e}^{-2n\phi_{3,-}}
    \end{pmatrix},\label{eq:JumpUnotinSmu1}
    \quad \textrm{on } \R \setminus S(\mu_1),
\end{align}
\begin{align}
J_{U}&=\begin{pmatrix}
    1& 0 & 0  & 0\\
    0& {\rm e}^{2n\phi_{2,-}} & 0  & 0 \\
    0 & 1 & {\rm e}^{2n\phi_{2,+}} & 0 \\
    0 & 0 & 0 &1
    \end{pmatrix},\qquad \textrm{on } S(\sigma-\mu_2),\label{eq:JumpUSsigmamu2}\\
J_{U}&=\begin{pmatrix}
    1& 0 & 0  & 0\\
    0& 1 & 0  & 0 \\
    0 & {\rm e}^{-2 n \phi_2} & 1 & 0 \\
    0 & 0 & 0 &1
    \end{pmatrix}, \qquad \textrm{on } {\rm i}\R \setminus S(\sigma-\mu_2),\label{eq:JumpUic}
\end{align}
\end{lemma}

\subsection{Proof of the jump matrices $J_U$ in the RH problem for $U$}

\begin{proof}

From \eqref{eq:jumpforXonR} we see that the jump matrix for $X$ on
the real line can be decomposed into  $2\times 2$ blocks. The jump
matrix \eqref{eq:JU} for $U$ can also be decomposed into $2\times
2$ blocks, with two non-trivial diagonal blocks. The upper left
block is given by
\begin{multline} \label{eq:upperleftblock}
\begin{pmatrix}
    {\rm e}^{n\left(g_{1,-}(z)+\ell \right)}& 0\\
    0 & {\rm e}^{n\left(g_{2,-}(z)-g_{1,-}(z)\right)}
\end{pmatrix}
\begin{pmatrix}
1 &{\rm e}^{-n\left(V(z)-\frac{3}{4}\tau^{4/3}|z|^{4/3}\right)}\\
0 & 1
\end{pmatrix}\\
\times
\begin{pmatrix}
{\rm e}^{-n\left(g_{1,+}(z)+\ell\right)}& 0\\
0 & {\rm e}^{-n\left(g_{2,+}(z)-g_{1,+}(z)\right)}
\end{pmatrix}\\
=\begin{pmatrix}
{\rm e}^{-n (g_{1,+}(z)-g_{1,-}(z))} & {\rm e}^{n\left(g_{1,+}(z)+g_{1,-}(z)-g_2(z)-V(z)+\frac{3}{4}\tau^{4/3} |z|^{4/3}+\ell \right)}\\
0 &  {\rm e}^{n \left(g_{1,+}(z)-g_{1,-}(z)\right)}
\end{pmatrix},
\end{multline}
which by the relations in Lemma \ref{lem:propertiesg1} and the
fact that $n$ is a multiple of three, gives the equalities of the
upper left blocks in \eqref{eq:JumpUSmu1} and
\eqref{eq:JumpUnotinSmu1}.

For the lower right block in the jump matrix \eqref{eq:JU} on the
real line we find after  a simple calculation
\begin{equation}\label{eq:lowerrightblock1}
    \begin{pmatrix}
    {\rm e}^{-n(g_{3,+}(z)-g_{3,-}(z)-\frac{{\rm i}3\sqrt{3}}{4} \tau^{4/3} z^{4/3})} & {\rm e}^{n(g_{3,+}(z)+g_{3,-}(z)-g_2(z))}\\
    0 & {\rm e}^{n(g_{3,+}(z)-g_{3,-}(z)-\frac{{\rm i}3\sqrt{3}}{4}
    \tau^{4/3} z^{4/3})}
\end{pmatrix}
\end{equation}
for $z\in \R_+$ and
\begin{equation}\label{eq:lowerrightblock2}
\begin{pmatrix}
    {\rm e}^{-n(g_{3,+}(z)-g_{3,-}(z)+\frac{{\rm i}3\sqrt{3}}{4} \tau^{4/3} |z|^{4/3})} & {\rm e}^{n(g_{3,+}(z)+g_{3,-}(z)-g_2(z))}\\
    0 & {\rm e}^{n(g_{3,+}(z)-g_{3,-}(z)+\frac{{\rm i}3\sqrt{3}}{4}
    \tau^{4/3} |z|^{4/3})}
\end{pmatrix}
\end{equation}
for  $z\in \R_-$. This proves the equality of the lower right
blocks in \eqref{eq:JumpUSmu1} and \eqref{eq:JumpUnotinSmu1} in
view of Lemma \ref{lem:propertiesg3}. It is important again that
$n$ is a multiple of three.

\medskip

We finally come to the  middle block in the jump matrix on the
imaginary axis. By \eqref{eq:jumpforXoniR} and \eqref{eq:JU} we
find after a calculation
\begin{equation}\label{eq:middleblock1}
\begin{pmatrix}
{\rm e}^{-n(g_{2,+}(z)-g_{2,-}(z)+\frac{{\rm i}3\sqrt{3}}{4} \tau^{4/3} z^{4/3})} & 0\\
{\rm e}^{-n(g_{2,+}(z)+g_{2,-}(z)-g_1(z) - g_3(z))} & {\rm
e}^{n(g_{2,+}(z)-g_{2,-}(z)+\frac{{\rm i}3\sqrt{3}}{4} \tau^{4/3}
z^{4/3})}
\end{pmatrix}
\end{equation}
for $z\in {\rm i} \R_+$ and
\begin{equation}\label{eq:middleblock2}
\begin{pmatrix}
{\rm e}^{-n(g_{2,+}(z)-g_{2,-}(z)-\frac{{\rm i}3\sqrt{3}}{4} \tau^{4/3} z^{4/3})} & 0\\
{\rm e}^{-n(g_{2,+}(z)+g_{2,-}(z)-g_3(z)-g_1(z))} & {\rm
e}^{n(g_{2,+}(z)-g_{2,-}(z)-\frac{{\rm i}3\sqrt{3}}{4} \tau^{4/3}
z^{4/3})}
\end{pmatrix}\end{equation}
for  $z\in {\rm i}\R_-$. Then by Lemma \ref{lem:propertiesg2} and
the fact that $n$ is a multiple of three, we indeed obtain
\eqref{eq:JumpUSsigmamu2} and \eqref{eq:JumpUic}.
\end{proof}

\subsection{Proof of the asymptotics for $U$}

\begin{proof} We deal with the asymptotic condition in the RH problem for $U$. Define
matrices
\begin{align}
H_1 &=\frac{3\alpha}{2}\begin{pmatrix}
1  & 0 & 0 \\
0 & \omega &0\\
0 & 0 & \omega^2
\end{pmatrix},
\qquad
H_2=\frac{3\alpha}{2}\begin{pmatrix}
\omega  & 0 & 0 \\
0 & 1 &0\\
0 & 0 & \omega^2
\end{pmatrix}, \\
H_3&=\frac{3\alpha}{2}\begin{pmatrix}
\omega^2  & 0 & 0 \\
0 & 1 &0\\
0 & 0 & \omega
\end{pmatrix}, \qquad
H_4=\frac{3\alpha}{2}\begin{pmatrix}
1  & 0 & 0 \\
0 & \omega^2 &0\\
0 & 0 & \omega
\end{pmatrix}.
\end{align}
From Lemma \ref{lem:asymptoticsg} and \eqref{eq:defG} we obtain
\begin{equation}
G(z)=\begin{pmatrix}
z^{-1} & 0 & 0  & 0\\
0 & z^{1/3}& 0  & 0\\
0 & 0 & z^{1/3} & 0\\
0 & 0 & 0 & z^{1/3}
\end{pmatrix} \left( I+\begin{pmatrix}1&0\\
0&H_j\end{pmatrix} z^{-2/3}+\OO(z^{-1}) \right)
\end{equation}
for $z\to \infty$ and in the $j$th quadrant with $j=1,3,4$. In the
second quadrant the situation is slightly different. Then we have
\begin{equation}
G(z)=\begin{pmatrix}
z^{-1} & 0 & 0  & 0\\
0 & \omega^2z^{1/3}& 0  & 0\\
0 & 0 & \omega z^{1/3} & 0\\
0 & 0 & 0 & z^{1/3}
\end{pmatrix} \left( I +\begin{pmatrix}1&0\\
0&H_2\end{pmatrix} z^{-2/3}+\OO(z^{-1}) \right)
\end{equation}
for $z\to \infty$ in the second quadrant. The extra factors
$\omega^2$ in the $(2,2)$ entry and $\omega$ in the $(3,3)$ entry
are due to the extra term $-\frac{4\pi{\rm i}}{3}$ in the
asymptotic behavior of $g_2$ in the second quadrant as given in
\eqref{eq:asymptoticsg2}. However these extra factors play no role
in the asymptotic behavior of $G(z)^n$, since $n$ is a multiple of
three.

Then by the asymptotics \eqref{eq:RHforX} for $X$,
\begin{align}\label{eq:asymptoticstildeXGn}
X(z)G(z)^n & =(I+\OO(z^{-{2/3}})) \begin{pmatrix}
1&0 &0&0\\
0&                   z^{-1/3} & 0 & 0\\
0&                    0        & 1 & 0\\
0&                    0        & 0 & z^{1/3}
                \end{pmatrix}\\
            & \qquad    \times
                \begin{pmatrix}1 &  0 \\
                0& A_j\end{pmatrix} \left(I+ n\begin{pmatrix}1&0\\
    0& H_j\end{pmatrix} z^{-2/3}+\OO(z^{-1})\right)\nonumber\\
\nonumber &= (I+\OO(z^{-{2/3}}))\begin{pmatrix}
1&0 &0&0\\
0&                   z^{-1/3} & 0 & 0\\
0&                    0        & 1 & 0\\
0&                    0        & 0 & z^{1/3}
                \end{pmatrix}\\
                &\qquad \times
                \left(I+ n\begin{pmatrix}1&0\\
0& A_jH_jA_j^{-1}\end{pmatrix} z^{-2/3}+\OO(z^{-1})\right)\begin{pmatrix}1 &  0 \\
                0& A_j\end{pmatrix}\nonumber
\end{align}
as $z\to \infty$ in the $j$th quadrant. A simple calculation then
shows that
\begin{align}
A_jH_j A_j^{-1} =\frac{3\alpha}{2}\begin{pmatrix}
0 & 0 & 1\\
1 &0 & 0\\
0 & 1 & 0
\end{pmatrix}, \qquad j=1,2,3,4,
\end{align}
and hence
\begin{align}
\begin{pmatrix}
1&0 &0&0\\
0&                   z^{-1/3} & 0 & 0\\
0&                    0        & 1 & 0\\
0&                    0        & 0 & z^{1/3}
                \end{pmatrix}
                \left(I+ n\begin{pmatrix}1&0\\
0& A_jH_jA_j^{-1}\end{pmatrix} z^{-2/3} \right)\\
\nonumber =(I+\OO(z^{-{1/3}}))\begin{pmatrix}
1&0 &0&0\\
0&                   z^{-1/3} & 0 & 0\\
0&                    0        & 1 & 0\\
0&                    0        & 0 & z^{1/3}
                \end{pmatrix}
\end{align}
as   $z \to \infty$. Combining this with
\eqref{eq:asymptoticstildeXGn} leads to
\begin{align}
X(z)G(z)^n=(1+\OO(z^{-1/3}))\begin{pmatrix}
1&0 &0&0\\
0&                   z^{-1/3} & 0 & 0\\
0&                    0        & 1 & 0\\
0&                    0        & 0 & z^{1/3}
                \end{pmatrix}
                \begin{pmatrix}1 &  0 \\
                0& A_j\end{pmatrix}
\end{align}
as $z\to \infty$ in the the $j$th quadrant. Then the asymptotic
condition in the RH problem \eqref{eq:RHforU} for $U$ follows by
\eqref{eq:XtoU} and the fact $L$ is a diagonal matrix of the form
\eqref{eq:defL}. This completes the proof of Lemma
\ref{lem:RHforU}.
\end{proof}

Let us summarize the progress we made so far. Although the RH
problem for $U$ is not normalized at infinity, an important step
is that the asymptotic behavior of $U$ does not depend on $n$
anymore.

The jump structure for $U$ is somewhat involved. Let us take a
closer look at the $2\times 2$ lower right blocks of the jump
matrix $J_{U}$ on $\R$, that is, at the lower right blocks of
\eqref{eq:JumpUSmu1} and \eqref{eq:JumpUnotinSmu1}. We see that
the diagonal entries are highly oscillating. Indeed, by
\eqref{eq:defphi3}, the fact that $\xi_{4,+}(z)=\xi_{3,-}(z)$ for
$z\in S(\sigma-\mu_2)$ and \eqref{eq:mu3inxi} we obtain
\begin{align}\label{eq:oscillatingjumps3}
\phi_{3,+}(z)= - \phi_{3,-}(z) = - \pi{\rm i}\int_0^z {\rm
d}\mu_3 \ \mp \ \frac{{\rm i}3\sqrt{3}}{8}\tau^{4/3} |z|^{4/3}
\end{align}
for $z\in \R$, and hence $\phi_{3,\pm}$ is purely imaginary on
$\R$.

An important step in the steepest analysis is the so-called
opening of the lens. In this step jump matrices with highly
oscillatory diagonal entries, are turned into a constant (or
otherwise nice behaved) jump matrix on the original contour and
jump matrices with exponentially decaying off-diagonal entries on
new contours on the lips of the lens. We do not go into details
here, but let us mention that this construction works under the
condition
\begin{align} \label{eq:conditiontoopenthelens3}
    \frac{{\rm d}}{{\rm d}x}\Im(\phi_{3,+}(x))<0
\end{align}
which holds in our case by \eqref{eq:oscillatingjumps3} and the
fact that  $\mu_3$ is a measure with a density that is strictly positive
on $\R$.

From \eqref{eq:JumpUSsigmamu2} we see that the jump matrix for $U$
on the the part $S(\sigma-\mu_2) = (-{\rm i}\infty,-{\rm i}c] \cup
[{\rm i}c,{\rm  i}\infty)$ of the imaginary axis has highly
oscillating terms on the diagonal. Indeed, by  \eqref{eq:defphi2},
the fact that $\xi_{3,+}(z)=\xi_{2,-}(z)$ for $z\in
S(\sigma-\mu_2)$ and \eqref{eq:mu2inxi} we obtain
    \begin{align}\label{eq:oscillatingjumps2}
        \phi_{2,+}(z)=-\pi{\rm i} \int_{{\rm i} c}^z {\rm d}(\mu_2-\sigma)
    \end{align}
for $z\in S(\sigma-\mu_2) \cap {\rm i}\R$. In order to open the
lenses successfully we now must have the condition
    \begin{align}\label{eq:conditiontoopenthelens2}
        \frac{{\rm d}}{{\rm d}y} \Im \phi_{2,+}({\rm i} y) > 0,
    \end{align}
for  $y>c$ and $y < c$. The difference in sign in
\eqref{eq:conditiontoopenthelens2} when compared to
\eqref{eq:conditiontoopenthelens3} is due to the different
triangularity structure of the jump matrices. Indeed the middle
$2\times 2$ block in \eqref{eq:JumpUSsigmamu2} is lower
triangular, in contrast to the non-trivial blocks in
\eqref{eq:JumpUSmu1} and \eqref{eq:JumpUnotinSmu1}, which are
upper triangular. From \eqref{eq:oscillatingjumps2} we see that
    \begin{align}
            \frac{{\rm d}}{{\rm d}y} \Im \phi_{2,+}({\rm i} y)=-\pi
            \frac{{\rm d}\mu_2}{|{\rm d}z|}({\rm i}y)+\pi \frac{{\rm d}\sigma}{|{\rm d}z|}({\rm i}y)
    \end{align}
for  $y\in \R$. And since ${\rm i}y\in S(\sigma-\mu_2)$ for $y>c$
and $y<c$ we have that condition
\eqref{eq:conditiontoopenthelens2}  is  satisfied. Here we see the
importance of the upper constraint $\mu_2\leq \sigma$.

\section{The third and fourth transformations $U\mapsto T \mapsto S$}

\subsection{Definition of the transformation $U \mapsto T$}
The next step in the Deift/Zhou steepest descent analysis is the
opening of lenses. The aim of this step is to turn the oscillating
diagonal entries in the jump matrices into exponentially small
off-diagonal entries.

We have to open a lens around each of the sets $S(\mu_1)$,
$S(\sigma-\mu_2)$ and $S(\mu_3)$. The latter two are unbounded and
we will treat the opening of lenses around these two sets in this
section. The opening of the lens around the bounded set $S(\mu_1)$
is more standard and it is deferred to the next section.

The lens around $S(\sigma-\mu_2)$ is opened as follows. In the
discussion at the end of the last section we have seen that
\eqref{eq:conditiontoopenthelens2} holds for every $y>c$ and
$y<-c$. Then it follows from the Cauchy-Riemann equations that
$\Re \phi_2(z)<0$ for $z$ in region around $(-{\rm i}\infty,-{\rm
i}c)\cup({\rm i}c,{\rm i}\infty)$. In particular, one can show
that, for some $r > 0$, it contains a cone
\begin{align}\label{eq:cone1}
 \{ z = x+{\rm i}y \in \mathbb C \mid 0 < |x| < 2r(|y| - c) \}
 \end{align}
in its interior.

\begin{figure}[t]
 \centering
    \begin{overpic}[width=9.6cm,height=8cm]{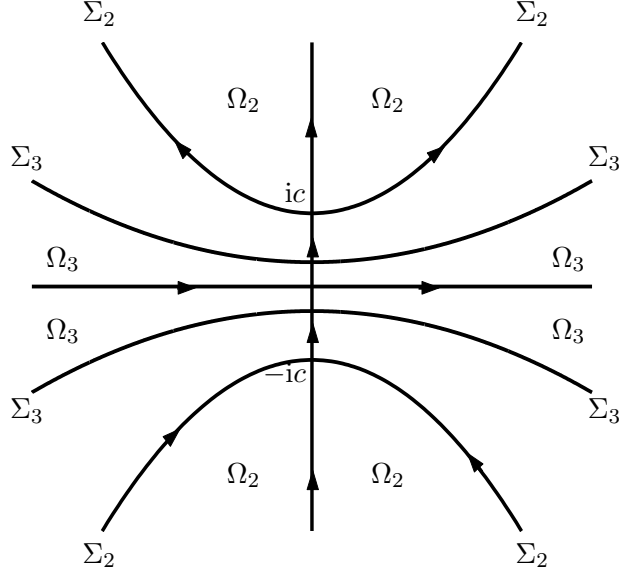}
      \put(85,46){\mbox{$\Omega_{3}$}}
      \put(60,68){\mbox{$\Omega_{2}$}}
      \put(40,68){\mbox{$\Omega_{2}$}}
      %\put(25,60){\mbox{$\Omega_5$}}
      \put(15,46){\mbox{$\Omega_{3}$}}
      \put(15,36){\mbox{$\Omega_{3}$}}
      %\put(25,22){\mbox{$\Omega_8$}}
      \put(40,16){\mbox{$\Omega_{2}$}}
      \put(60,16){\mbox{$\Omega_{2}$}}
      %\put(75,22){\mbox{$\Omega_{11}$}}
      \put(85,36){\mbox{$\Omega_{3}$}}
      \put(20,5){\mbox{$\Sigma_{2}$}}
      \put(80,5){\mbox{$\Sigma_{2}$}}
      \put(90,25){\mbox{$\Sigma_{3}$}}
      \put(90,60){\mbox{$\Sigma_{3}$}}
      \put(80,80){\mbox{$\Sigma_{2}$}}
      \put(20,80){\mbox{$\Sigma_{2}$}}
      \put(10,60){\mbox{$\Sigma_{3}$}}
      \put(10,25){\mbox{$\Sigma_{3}$}}
    \put(45,30){\mbox{$-{{\rm i}c}$}}
    \put(48,55){\mbox{${\rm i}c$}}
%    \put(38,41){\mbox{$-a$}}
%    \put(60,41){\mbox{$a$}}
    \end{overpic}
    \caption{The contours $\Sigma_2$ and $\Sigma_3$ enclose the inner
    regions $\Omega_2$ and $\Omega_3$.
    Together with $\R\cup{\rm i}\R$ they form the contour  $\Sigma_T$.}
    \label{GlobalOpeningLens}
  \end{figure}

Then we take the contour $\Sigma_2$ around $S(\sigma -\mu_2)$  as
shown in Figure \ref{GlobalOpeningLens}, so that
\begin{itemize}
    \item $\Re \phi_2(z) < 0$ for $z \in \Sigma_2 \setminus \{\pm {\rm i}c\}$, and
    \item there exists an $r >0$ such that $|x| > r (|y|-c)$ for every $z = x+{\rm i}y \in
    \Sigma_2 \setminus \{\pm {\rm i}c\}$.
\end{itemize}

The lens around $S(\mu_3) = \R$ is opened as follows. Since
\eqref{eq:conditiontoopenthelens3} holds for every $x \in \R$, we
have that $\Re \phi_3(z) > 0$ for $z$ in a region around $\R$. The
region is unbounded and it can be shown that it contains the  cone
\[ \{ z = x+{\rm i}y \in \mathbb C \mid 0 < |y| < 2r(|x|+1) \} \]
for some $r>0$. We take the contour $\Sigma_3$, as shown in Figure \ref{GlobalOpeningLens}, such that
\begin{itemize}
    \item $\Re \phi_3(z) > 0$ for $z \in \Sigma_3$, and
    \item there exists an $r > 0$ so that $|y| > r (|x|+1)$ for every $z = x+{\rm i}y \in
    \Sigma_3$,
\end{itemize}
see also Figure \ref{GlobalOpeningLens}. We may (and do) assume
that $\Sigma_2$ and $\Sigma_3$ do not intersect.

The contours $\Sigma_2$ and $\Sigma_3$ give rise to a partitioning
of the complex plane  as in Figure \ref{GlobalOpeningLens}. The
inner part of the lens around $S(\sigma-\mu_2)$ enclosed by the
contour $\Sigma_2$ is denoted by $\Omega_2$, and the inner part of
the lens around $\R$ enclosed by $\Sigma_3$ is denoted by
$\Omega_3$.

We define the $4\times 4$ matrix valued function $T$ by
\begin{align}
T(z)& = U(z) \begin{pmatrix} 1 & 0 & 0 & 0\\
0 & 1 & {\rm e}^{2n\phi_2(z)} & 0\\
0 & 0 & 1 & 0\\
0 & 0 & 0& 1
\end{pmatrix}, \qquad  z\in \Omega_{2}, \ \Re z > 0, \label{eq:Tomegal2}
 \\
T(z)& = U(z) \begin{pmatrix} 1 & 0 & 0 & 0\\
0 & 1 & -{\rm e}^{2n\phi_2(z)} & 0 \\
0 & 0 & 1 & 0\\
0 & 0 & 0&1
\end{pmatrix}, \qquad  z\in \Omega_{2}, \ \Re z < 0, \label{eq:Tomegau2}\\
 T(z)& =  U(z) \begin{pmatrix} 1 & 0 & 0 & 0\\
0 & 1 &  0 & 0\\
0 & 0 & 1 & 0\\
0 & 0 & -{\rm e}^{-2n\phi_3(z)}&1
\end{pmatrix}, \qquad   z\in \Omega_{3}, \ \Im z > 0, \label{eq:Tomegau3} \\
T(z)& =  U(z) \begin{pmatrix} 1 & 0 & 0 & 0\\
 0 & 1 & 0 & 0\\
0 & 0 & 1 & 0\\
0 & 0 & {\rm e}^{-2n\phi_3(z)}&1
\end{pmatrix}, \qquad  z\in \Omega_{3}, \ \Im z < 0,  \label{eq:Tomegal3}
\end{align}
and
 \begin{align} \label{eq:Telsewhere}
    T(z) & = U(z) \qquad \textrm{ elsewhere.}
\end{align}

Then the $4\times 4$ matrix valued function $T$ is defined and
analytic in $\C\setminus \Sigma_T$ where $\Sigma_T$ is given by
    \begin{align}
        \Sigma_T = \R \cup{\rm i}\R \cup \Sigma_{2}\cup \Sigma_{3}.
    \end{align}

\subsection{RH problem for $T$}

In the following lemma we state the RH problem that is satisfied
by $T$.

\begin{lemma}
$T$ is the unique solution of the following RH problem
\begin{equation} \label{eq:RHproblemT}
\left\{
\begin{array}{l}
T \textrm{ is analytic in }  \C\setminus \Sigma_{T},  \\
 T_+(z)= T_-(z) J_{T}(z), \qquad z\in \Sigma_{T}, \\
T(z)=(I+\OO(z^{-1/3}))\begin{pmatrix}
1& 0 & 0 &0 \\
0& z^{1/3} & 0 & 0\\
0 & 0 &1 & 0\\
0 & 0 & 0 & z^{-1/3}
\end{pmatrix}\begin{pmatrix}
1 & 0 \\
0 & A_j
\end{pmatrix} \\
 \multicolumn{1}{r}{\textrm{as } z\to \infty \textrm{ in the }j\textrm{th quadrant}.}
\end{array}
\right.
\end{equation}
The matrices $A_j$ are given by \eqref{eq:A1A2} and
\eqref{eq:A3A4}, and the asymptotic condition in
\eqref{eq:RHproblemT} holds uniformly as $z \to \infty$ in each
quadrant.

The jump matrix $J_{T}$ is given by
\begin{align} \label{eq:JT1}
J_{T} &=
\begin{pmatrix}
{\rm e}^{-2n\phi_{1,+}} & 1 & 0 & 0 \\
0 & {\rm e}^{-2n\phi_{1,-}} & 0 & 0 \\
0 & 0 & 0 & 1\\
0 & 0 & -1 & 0\\
\end{pmatrix}, \quad \textrm{on } S(\mu_1),\\
\label{eq:JT2} J_{T} &=\begin{pmatrix}
1 & {\rm e}^{2n\phi_{1,+}} & 0 & 0 \\
0 & 1 & 0 & 0 \\
0 & 0 & 0 & 1\\
0 & 0 & -1 & 0\\
\end{pmatrix}, \quad \begin{array}{c} \textrm{on } \R \setminus S(\mu_1), \end{array}\\
\label{eq:JT3} J_{T}&=\begin{pmatrix}
1 & 0 & 0 & 0 \\
0 & 1 & 0 & 0 \\
0 & 0 & 1 & 0\\
0 & 0 & {\rm e}^{-2n\phi_3} & 1\\
\end{pmatrix}, \quad \textrm{on } \Sigma_3, \\
\label{eq:JT4} J_{T}&=\begin{pmatrix}
1 & 0 & 0 & 0 \\
0 & 0 & -1 & 0 \\
0 & 1 & 0 & 0\\
0 & 0 & 0 & 1\\
\end{pmatrix}, \quad \textrm{on } S(\sigma-\mu_2), \\
\label{eq:JT5} J_{T}&=\begin{pmatrix}
1 & 0 & 0 & 0 \\
0 & 1 & {\rm e}^{2n\phi_2} & 0 \\
0 & 0 & 1 & 0\\
0 & 0 & 0 & 1\\
\end{pmatrix}, \quad \textrm{on } \Sigma_{2},
\end{align}
and
\begin{align}
\label{eq:JT6} J_{T} &=\begin{pmatrix}
1 & 0 & 0 & 0 \\
0 & 1 & 0 & 0 \\
0 & {\rm e}^{-2n\phi_2} & 1 & 0\\
0 & 0 & 0 & 1\\
\end{pmatrix}, \quad \textrm{on } (-{\rm i}c,{\rm i}c) \setminus \Omega_3, \\
\label{eq:JT7} J_{T} &=\begin{pmatrix}
1 & 0 & 0 & 0 \\
0 & 1 & 0 & 0 \\
0 & {\rm e}^{-2n\phi_{2}} & 1 & 0\\
0 &  {\rm e}^{-2n(\phi_{2} +\phi_3)} & 0 & 1\\
\end{pmatrix}, \quad \textrm{on } (0,{\rm i}c) \cap \Omega_3, \\
J_{T} &=\label{eq:JT8}\begin{pmatrix}
1 & 0 & 0 & 0 \\
0 & 1 & 0 & 0 \\
0 & {\rm e}^{-2n\phi_{2}} & 1 & 0\\
0 &  -{\rm e}^{-2n(\phi_{2} +\phi_3)} & 0 & 1\\
\end{pmatrix}, \quad \textrm{on } (-{\rm i}c,0) \cap \Omega_3.
\end{align}
\end{lemma}

\begin{proof} Each of the jump matrices
\eqref{eq:JT1}--\eqref{eq:JT7} follows from straightforward
calculations based on the definitions
\eqref{eq:Tomegal2}--\eqref{eq:Tomegal3} and the jump matrices in
the RH problem for $U$. Then the jump matrices
\eqref{eq:JT1}--\eqref{eq:JT3} are based on the factorization
\[ \begin{pmatrix} e^{-2n \phi_{3,+}} & 0 \\ 1 & e^{2n\phi_{3,-}}
    \end{pmatrix}
    = \begin{pmatrix} 1 & 0 \\ e^{-2n \phi_{3,-}}  & 1 \end{pmatrix}
    \begin{pmatrix} 0 & 1 \\ -1 & 0 \end{pmatrix}
    \begin{pmatrix} 1 & 0 \\ e^{-2n \phi_{3,+}} & 1 \end{pmatrix}
    \]
of the $2\times 2$ lower right block in the jump matrix $J_{U}$ on
$\R$, see \eqref{eq:JumpUSmu1} and \eqref{eq:JumpUnotinSmu1}. The
jump matrices \eqref{eq:JT4}--\eqref{eq:JT5} are similarly based
on the factorization
\[ \begin{pmatrix} e^{2n \phi_{2,-}} & 0 \\ 1 & e^{2n\phi_{2,+}}
\end{pmatrix}
    = \begin{pmatrix} 1 & e^{2n \phi_{2,-}} \\ 0 & 1 \end{pmatrix}
    \begin{pmatrix} 0 & -1 \\ 1 & 0 \end{pmatrix}
    \begin{pmatrix} 1 & e^{2n \phi_{2,+}} \\ 0 & 1 \end{pmatrix}
    \]
of the $2 \times 2$ middle block in the jump matrix $J_{U}$ on
$S(\sigma-\mu_2)$, see \eqref{eq:JumpUSsigmamu2}.

The jump matrix \eqref{eq:JT6} is the same as the corresponding
jump matrix \eqref{eq:JumpUic} in the RH problem for $U$. The jump
matrix \eqref{eq:JT7} follows from conjugating \eqref{eq:JumpUic}
with either \eqref{eq:Tomegau3} or \eqref{eq:Tomegal3}.

The asymptotic condition in the RH problem for $T$ follows from
the definition \eqref{eq:Tomegal2}--\eqref{eq:Tomegal3}, the
asymptotic condition in the RH problem for $U$ in Lemma
\ref{lem:RHforU}, and the fact that $\Re \phi_2 < 0$ in $\Omega_2$
and $\Re \phi_3 > 0$ in $\Omega_3$. For the latter facts, see also
Lemmas \ref{lem:Rephi2} and \ref{lem:Rephi3} below.

A more detailed analysis would show that the asymptotics for $T$
actually holds uniformly up to the axes, in contrast to the
asymptotics for $U$. To show this we would have to trace back the
transformations $Y \mapsto X \mapsto U \mapsto T$ to see in
particular what combination of Pearcey integrals is actually
involved in the regions $\Omega_2$ and $\Omega_3$ near the axes.
It turns out that the asymptotics of the relevant combinations of
Pearcey integrals is uniformly valid up to the axes. We will not
give details here.
\end{proof}

\subsection{Large $n$ behavior of $J_T$}

In the following lemmas we state result about the sign of $\Re
\phi_2$ and $\Re \phi_3$ on various parts of the contour
$\Sigma_T$.

\begin{lemma}\label{lem:Rephi2}
For every  neighborhoods $U_{{\rm i}c}$ and $U_{-{\rm i}c}$ around
${\rm i}c$ and $-{\rm i}c$ there exists an $\eps_2>0$ such that
    \begin{enumerate}
        \item[\rm (a)] $\Re \phi_2(z) <-\eps_2 |z|^{4/3}$ for $z\in \Sigma_2\setminus U_{\pm {\rm i}c}$,
        \item[\rm (b)] $\Re \phi_2(z) >\eps_2$ for $z\in [-{\rm i}c,{\rm i}c]\setminus U_{\pm {\rm i}c}$.
    \end{enumerate}
\end{lemma}
\begin{proof} (a) By definition of $\Sigma_2$ we have $\Re
\phi_2(z)<0$ for $z\in \Sigma_2\setminus \{\pm{\rm  i}c\}$. Let us
consider the behavior for $z\in \Sigma_2\cap I$ near infinity.
Then
\begin{align}   \label{eq:asymptoticsphi2first}
        \phi_2(z)=-\frac{1}{2}\int_{{\rm i } c}^z \left(\xi_{2}(y)-\xi_3(y)\right) \ {\rm d}y=\frac{3(\omega-1)\tau^{4/3}}{8} z^{4/3}+\OO(\log(z)),
    \end{align}
as $z\to \infty$ remaining in the first quadrant. Since $\Omega_2$ contains the cone \eqref{eq:cone1} for some $r>0$, we have that there exists a $\delta>0$ such that
\begin{align}\label{eq:deltacone}
0<\arg z<\pi/2-\delta
\end{align}
for large enough $z\in \Sigma_2\cap I$.  Then we have from \eqref{eq:asymptoticsphi2first} that
\begin{align}
\Re \phi_2(z) &= \frac{3}{8}\tau^{4/3} |z|^{4/3} \Re \left((\omega-1){\rm e}^{4{\rm i} \arg z/3}\right)\left(1+\OO(|z|^{-4/3}\log|z|)\right)
\end{align}
as $z\to \infty$ and $z\in \Sigma_2\cap I$. By \eqref{eq:deltacone} we have that
\begin{multline}
    \Re \left((\omega-1){\rm e}^{4{\rm i} \arg z/3}\right)=\sqrt{3} \sin \left(\frac{4}{3}(\arg z-\pi/2)\right)\\
    \leq-\sqrt{3} \min\left(\sin \frac{4\delta}{3},\frac{\sqrt{3}}{2}\right)
\end{multline}
for $z\in \Sigma_2\cap I$ large enough. Hence there exists an $\eps_2>0$ such that
\begin{align}\label{eq:rephi2large}
\Re\phi_2(z)<-\eps_2 |z|^{4/3}
\end{align}
for $z\in \Sigma_2 \cap I$ large enough.  By similar arguments in
the other quadrants, we see that we can choose $\eps_2$ such that
\eqref{eq:rephi2large} holds for $z\in \Sigma_2$ large enough.
Finally, by continuity we can choose $\eps_2$ small enough such
that \eqref{eq:rephi2large} holds for  $z\in \Sigma_2\setminus
U_{{\pm {\rm i}}c}$.  This proves the first property

(b) The second property follows from the variational conditions
for $\mu_2$. From \eqref{eq:defxi} and
\eqref{eq:variationalcondmu2b} we have
\begin{multline}
\Re \phi_2(z)=\Re (2g_2(z)-g_1(z)-g_3(z))\\=-2 U^{\mu_2}(z)+U^{\mu_1}(z)+U^{\mu_3}(z)>0
\end{multline}
for $z\in (-{\rm i}c,{\rm i}c)$. It is also clear that $\Re
\phi_2$ is continuous on $(-{\rm i}c,{\rm i}c)$. Hence the
statement follows. \end{proof}

\begin{lemma}\label{lem:Rephi3}
There exists an $\eps_3>0$ such that $\Re \phi_3(z)>\eps_3|z|^{4/3}$ for all $z\in \Sigma_3$.
\end{lemma}
\begin{proof} The statement can be proved in the same way as we
proved property (a) in the proof of Lemma \ref{lem:Rephi2}.
\end{proof}

By  Lemma  \ref{lem:Rephi2} we see that the jump matrices $J_T$ in
\eqref{eq:JT5} and \eqref{eq:JT6} converge pointwise to the
identity matrix $I$ at an exponential rate as $n\to \infty$. In
fact, it shows that this convergence is uniform as long as we stay
away from the points $\pm {\rm i}c$.  Since $\Re \phi_3>0$ in
$\Omega_3$ we also have that \eqref{eq:JT7} and \eqref{eq:JT8}
converge uniform to the identity matrix  $I$ at an exponential
rate as $n\to\infty$.

By Lemma \ref{lem:Rephi3}, we see that the jump matrix $J_T$ in
\eqref{eq:JT3} converges uniformly to the identity matrix $I$ at
an exponential rate as $n\to \infty$.

\subsection{The fourth transformation $ T\mapsto  S$}

In the next transformation  we also open the lens $\Omega_1$
around $S(\mu_1) = [-a,a]$.

We use $\Sigma_1$ to denote the outer boundary of $\Omega_1$ so
that $\Sigma_1$ consists of two contours from $-a$ to $a$, one in
the upper half-plane and one in the lower half-plane. Both
contours are oriented from $-a$ to $a$. See also Figure
\ref{fig:jumps}.

Note that
\begin{align}
    \phi_{1,+}(z) = -\pi{\rm i}\int_{a}^z {\rm d}\mu_1, \qquad z
    \in [-a,a],
\end{align}
Since we assume that $\mu_1$ is regular we have
\begin{align}
    \frac{{\rm d}}{{\rm d}x}\Im \phi_{1,+}(x)<0 \quad
    \textrm{   and   } \quad \frac{{\rm d}}{{\rm d}x}\Im \phi_{1,-}(x)>0,
\end{align}
for $x \in (-a,a)$, and we see that  by the Cauchy-Riemann
equations, there exists a region around $(-a, a)$ so that $\Re
\phi_1(z) > 0$ for every $z \not\in (-a,a)$ in that region. Hence
the contours $\Sigma_1$ can be taken such that
\begin{itemize}
\item $\Re \phi_1(z) > 0$ for $z \in \Sigma_1 \setminus
 \{-a, a \}$,
 \item $\Sigma_1$ is contained in $\Omega_3$.
\end{itemize}

We define the $4 \times 4$ matrix valued function $S$ by
\begin{align} \label{eq:defS1}
 S(z)& =T(z) \begin{pmatrix} 1 & 0 & 0 & 0\\
 -{\rm e}^{-2n\phi_1(z)} & 1 &  0 & 0\\
 0 & 0 & 1 & 0\\
 0 & 0 & 0&1
 \end{pmatrix}, \qquad  z\in \Omega_{1}, \ \Im z > 0, \\
    \label{eq:defS2}
    S(z) & = T(z) \begin{pmatrix} 1 & 0 & 0 & 0\\
 {\rm e}^{-2n\phi_1(z)} & 1 & 0 & 0\\
0 & 0 & 1 & 0\\
0 & 0 & 0&1
\end{pmatrix}, \qquad  z\in \Omega_{1}, \ \Im z < 0,
\end{align}
and
\begin{align} \label{eq:defS3}
    S(z) &=T(z) \qquad \textrm{ elsewhere}.
\end{align}

Then $S$ is defined and analytic in $\C \setminus \Sigma_S$ where
\[ \Sigma_S = \R \cup {\rm i} \R \cup \Sigma_1 \cup \Sigma_2 \cup \Sigma_3, \]
see Figure \ref{fig:jumps} for a sketch of $\Sigma_S$. The next
lemma gives the RH problem that is satisfied by $S$.

\begin{lemma}
$S$ is the unique solution of the following RH problem
\begin{equation}\label{eq:RHproblemS}
\left\{
\begin{array}{l}
S \textrm{ is analytic in }  \C\setminus \Sigma_S,\\
 S_+(z)=S_-(z) J_{ S}(
 z),\qquad  z\in \Sigma_S,\\
S(z)=(I+\OO(z^{-1/3}))\begin{pmatrix}
1& 0 & 0 &0 \\
0& z^{1/3} & 0 & 0\\
0 & 0 &1 & 0\\
0 & 0 & 0 & z^{-1/3}
\end{pmatrix}\begin{pmatrix}
1 & 0 \\
0 & A_j
\end{pmatrix} \\
 \multicolumn{1}{r}{\textrm{uniformly as } z\to \infty \textrm{ in the }j\textrm{th quadrant.}}\end{array} \right.
\end{equation}
The matrices $A_j$ are given by \eqref{eq:A1A2} and
\eqref{eq:A3A4} and the jump matrices  $J_{S}$ are given by
\begin{align}\label{eq:JS1}
J_S(z)&=\begin{pmatrix}
0 & 1 & 0 & 0 \\
-1 & 0 & 0 & 0 \\
0 & 0 & 0  &1\\
0 & 0 & -1 & 0\\
\end{pmatrix}, \quad  z\in  S(\mu_1), \\
\label{eq:JS2}
J_S(z)&=\begin{pmatrix}
1 & 0 & 0 & 0 \\
{\rm e}^{-2n\phi_1(z)} & 1 & 0 & 0 \\
0 & 0 & 1  &0\\
0 & 0 & 0 & 1\\
\end{pmatrix}, \quad  z\in  \Sigma_{1},
\end{align}
\begin{align}\label{eq:JS3}
J_{S}(z)&=\begin{pmatrix}
1 & 0 & 0 & 0 \\
0 & 1 & 0 & 0 \\
-{\rm e}^{-2n(\phi_1(z)+\phi_{2}(z))} & {\rm e}^{-2n\phi_{2}(z)} & 1 & 0\\
-{\rm e}^{-2n(\phi_1(z)+\phi_{2}(z)+\phi_3(z))}  & {\rm e}^{-2n(\phi_{2}(z)+\phi_3(z))} & 0 & 1\\
\end{pmatrix},
\end{align}
for $z \in (0, {\rm i}c) \cap \Omega_1,$
\begin{align}\label{eq:JS4}
J_{S}(z)&=\begin{pmatrix}
1 & 0 & 0 & 0 \\
0 & 1 & 0 & 0 \\
{\rm e}^{-2n(\phi_1(z)+\phi_{2}(z))} & {\rm e}^{-2n\phi_{2}(z)} & 1 & 0\\
-{\rm e}^{-2n(\phi_1(z)+\phi_{2}(z)+\phi_3(z))}  & -{\rm e}^{-2n(\phi_{2}(z)+\phi_3(z))} & 0 & 1\\
\end{pmatrix},
\end{align}
for $z \in (-{\rm i}c, 0) \cap \Omega_1,$
 and
 \begin{align}
 J_S=J_{T} \qquad \textrm{ elsewhere}.
 \end{align}

\end{lemma}
\begin{proof} Since  $S$ differs from $T$ only in a bounded domain, it
is clear that the asymptotic behavior of $S$ is the same as that
of $T$. Hence the asymptotic condition in \eqref{eq:RHproblemS}
follows from the asymptotic condition in the RH problem
\eqref{eq:RHproblemT}.

The calculations that lead to the jump matrices are based on the
factorization
 \[ \begin{pmatrix} e^{-2n \phi_{1,+}} & 1 \\ 0 & e^{-2n\phi_{1,-}} \end{pmatrix}
    = \begin{pmatrix} 1 & 0 \\ e^{-2n\phi_{1,-}} & 1 \end{pmatrix}
    \begin{pmatrix} 0 & 1 \\ -1 & 0 \end{pmatrix}
    \begin{pmatrix} 1 & 0 \\ e^{-2n \phi_{1,+}} & 1 \end{pmatrix}
\]
of the $2\times 2$ upper left block of \eqref{eq:JT1}.  We will
not give further details here as this step in the RH steepest
descent analysis is similar to the corresponding step in the RH
analysis for orthogonal polynomials considered in
\cite{DKMVZuniform}. \end{proof}

\begin{figure}[t]
\centering
\begin{overpic}[width=9.6cm,height=8cm]{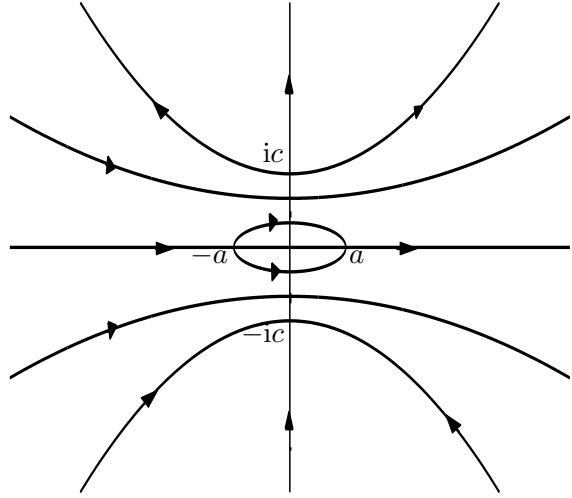}

\put(45,30){\mbox{$-{{\rm i}c}$}}
\put(48,55){\mbox{${\rm i}c$}}
\put(38,41){\mbox{$-a$}}
\put(60,41){\mbox{$a$}}
\end{overpic}
\caption{The jump contour $\Sigma_S$ in the RH problem for $S$ in
the one-cut case.} \label{fig:jumps}
\end{figure}

\bigskip

The following lemma deals with the sign of $\Re \phi_1$ at parts
of the contour $\Sigma_S$.

\begin{lemma} \label{lem:Rephi1}
Let $U_{-a}$ be a neighborhood of $-a$ and $U_{a}$ a neighborhood
of $a$.  Then there exists an $\eps_1>0$ such that
    \begin{enumerate}
        \item[\rm (a)] $\Re \phi_1(z)>\eps_1$ for $z \in \Sigma_{1}\setminus
        \left( U_{-a} \cup U_{a}\right) $,
        \item[\rm (b)] $\Re \phi_{1,+}(x)<-\eps_1 (|x|+1)$ for
        $x\in \R\setminus \left(S(\mu_1) \cup U_{-a}\cup U_{a}\big)\right)$.
    \end{enumerate}
\end{lemma}
\begin{proof} (a) This follows immediately from the continuity of
$\phi_{1}$ and the fact that $\Re \phi_1>0$ on
$\Sigma_1^{(k)}\setminus\{a_k,b_k\}$.

(b) Part (b) follows from the variational condition
\eqref{eq:variationalcondmu1b}. Indeed, by
\eqref{eq:variationalphi1} we have that
\begin{align} \label{eq:Rephi1final}
    \Re \phi_{1,+}(x)=\frac{1}{2} \Re \left(2g_1(x)-g_2(x)-V(x)+\frac{3}{4} \tau^{4/3} x^{4/3}+k_1\right)
\end{align}
for $x \in \R\setminus S(\mu_1)$. By
\eqref{eq:variationalcondmu1b} and regularity of the measure we
see that $\Re \phi_1(x)<0$ for $x\in \R \setminus S(\mu_1)$.
Moreover,  for large $x$ we have that the dominant term at the
right-hand side of \eqref{eq:Rephi1final} is $-V(x)$, which is a
polynomial of even degree. Combining this with the continuity of
$\Re \phi_1$ we see that there exists an $\eps_1>0$ such that
\begin{align}
    \Re \phi_{1,+}(x)<-\eps_1 (|x|+1)
\end{align}
for all $x\in \R\setminus S(\mu_1)$. \end{proof}

By Lemma \ref{lem:Rephi1} the matrix $J_T$ in \eqref{eq:JS2}
converges to the identity matrix $I_4$ at an exponential rate as
$n\to \infty$. If one stays away from the endpoints of $S(\mu_1)$,
the convergence is uniform. Combining Lemmas \ref{lem:Rephi1},
\ref{lem:Rephi2} and \ref{lem:Rephi3} we also see that the jump
matrix $J_T$ given in \eqref{eq:JS3} and \eqref{eq:JS4} converge
uniformly to the identity matrix $I_4$ at an exponential rate as
$n\to \infty$.

We also see that the $(1,2)$ entry of $J_T$ in \eqref{eq:JT2}
converges to zero at an exponential rate as $n\to \infty$. Again
if we stay away from the endpoints of $S(\mu_1)$ the  convergence
is uniform.

\section{Construction of parametrices and the transformation $S\mapsto R$}

\subsection{The RH problem for $M$}

If we ignore all exponentially small entries in the jump matrices $J_S$
in the RH problem for $S$, then we obtain the following model RH
problem for a $4 \times 4$ matrix valued function $M$.
    \begin{align}
        \label{eq:RHforM}
            \left\{
            \begin{array}{l}
                    \multicolumn{1}{l}{M \textrm{ is analytic in }
                    \C \setminus ( \R\cup S(\sigma-\mu_2)),} \\
                    M_+(z)=M_-(z)
                        \begin{pmatrix}
                                0& 1 & 0 & 0\\
                                -1 & 0 & 0 & 0\\
                                0 & 0 & 0 & 1\\
                                0 & 0 & -1 & 0
                        \end{pmatrix}, \quad  z\in S(\mu_1), \\
                    M_+(z)=M_-(z)
                        \begin{pmatrix}
                                1 & 0 & 0 & 0\\
                                0 & 1 & 0 & 0\\
                                0 & 0 & 0 & 1\\
                                0 & 0 & -1 & 0
                        \end{pmatrix}, \quad z \in \R \setminus S(\mu_1), \\
                    M_+(z)=M_-(z)
                        \begin{pmatrix}
                            1 & 0 & 0 & 0\\
                            0 & 0 & -1 & 0\\
                            0 & 1 & 0 & 0\\
                            0 & 0 & 0 & 1
                        \end{pmatrix}, \quad z\in S(\sigma-\mu_2),\\
                    M(z)=(I+\OO(z^{-1}))
                        \begin{pmatrix}
                            1& 0 & 0 &0\\
                            0 & z^{1/3} & 0 & 0\\
                            0 & 0 & 1 & 0 \\
                            0 & 0 & 0 & z^{-1/3}
                        \end{pmatrix}
                        \begin{pmatrix}
                            1 & 0 \\
                            0 & A_j
                        \end{pmatrix}\\
                    \multicolumn{1}{r}{\text{uniformly as } z\to \infty \text{ in the $j$th quadrant.}}
                \end{array}\right.
    \end{align}
Here, the matrices $A_j$ are given by \eqref{eq:A1A2} and
\eqref{eq:A3A4}.

The solution of \eqref{eq:RHforM} is not unique.  To ensure
uniqueness, we impose the additional conditions
\begin{align}
        \label{eq:assumptionfourthroot}
    \left\{ \begin{array}{ll}
            M(z)=\OO((z\mp a)^{-1/4}), & \qquad z\to \pm a, \\
            M(z)=\OO((z\mp {\rm i}c)^{-1/4}), & \qquad z\to \pm {\rm i}c.
            \end{array} \right.
    \end{align}

\subsection{Construction of the outside parametrix $M$}

The RH problem for $M$ can be solved in the one-cut case by using a
rational parametrization of the Riemann surface, which provides a
conformal map to the Riemann sphere.

We have $S(\mu_1)=[-a,a]$ and  $S(\sigma-\mu_2) = {\rm i}\R
\setminus (-{\rm i}c, {\rm i}c)$. The Riemann surface thus depends
on two parameters $a$ and $c$. It has genus zero. An explicit
rational parametrization is given by the equation
    \begin{equation} \label{algebraicequ}
        w+\frac{s^2-t^2}{w}+\frac{s^2t^2}{3w^3}=z,
    \end{equation}
where the constants $s > 0$, $t > $ are the unique positive
solutions of the equations
    \begin{equation}\label{equ-for-st}
        \left\{ \begin{array}{l}
            2s-\frac{2t^2}{3s}=a, \\
            2t-\frac{2s^2}{3t}=-c.
       \end{array}\right.
    \end{equation}
In this parametrization the branch points at $z=\pm a$ and $z=\pm
{\rm i} c$ correspond to $w = \pm s$ and $w=\mp {\rm i}t$
respectively.

We introduce the following function
\begin{align} \label{eq:squarerootscut}
    w\mapsto \big((w^2-s^2)(w^2+t^2)\big)^{1/2},
\end{align}
defined and analytic in the complex $w$-plane cut along
\begin{align} \label{eq:defGammaN}
    \Gamma_N :=  w_{1,+}(S(\mu_1)) \cup w_{2,-}(S(\sigma-\mu_2))
    \cup w_{3,+}(S(\mu_3)).
\end{align}
The square root is  taken such that
\eqref{eq:squarerootscut} behaves like $w^2$ as $w\to \infty$, and
such that it changes sign when crossing the cuts. See also Figure
\ref{imagrieman2}.

\begin{figure}[t]
\centering
\begin{overpic}[width=10cm,height=6cm]{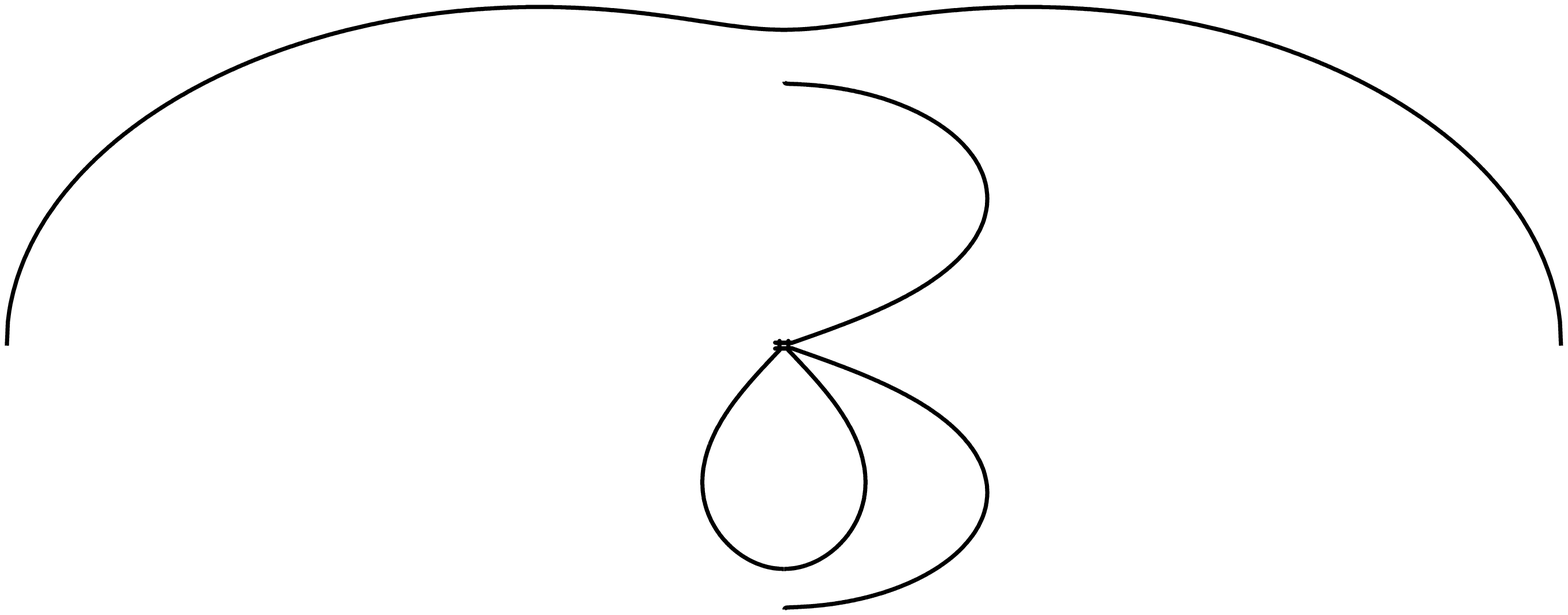}

\put(8,28){\mbox{$-s$}} \put(90,28){\mbox{$s$}}
\put(45,8){\mbox{$-{{\rm i}t}$}} \put(48,50){\mbox{${\rm i}t$}}
\put(57,30){\mbox{$w_{2,-}(S(\sigma-\mu_2))$}}
\put(25,19){\mbox{$w_{3,+}(S(\mu_3))$}}
\put(17,45){\mbox{$w_{1,+}(S(\mu_1))$}}
\end{overpic}
\caption{The cut $\Gamma_N$ in the definition of the square root
of the function $w\mapsto \big((w^2-s^2)(w^2+t^2)\big)^{1/2}$.}
\label{imagrieman2}
\end{figure}

\begin{proposition}
The RH problem \eqref{eq:RHforM}, \eqref{eq:assumptionfourthroot}
for $M$ has a unique solution.
\end{proposition}

\begin{proof}
The uniqueness of the solution follows by standard arguments for
uniqueness of RH problems, see e.g.\ \cite{DeiftBook}. We have to
prove existence of a solution.

Choose any basis $Q_k$, $k=1, \ldots, 4$ of the space of
polynomials of degree $\leq 3$, and define
 \begin{align} \label{eq:defNk}
    N_k(w) = \frac{Q_k(w)}{w\big((w^2-s^2)(w^2+t^2)\big)^{1/2}}.
    \end{align}
Then each $N_k$ is analytic in $\C \setminus \Gamma_N$ and
satisfies
 \begin{align} \label{eq:hatNk1}
    N_{k,+}(w) & = - N_{k,-}(w), \qquad w \in \Gamma_N,
    \\ \label{eq:hatNk2}
    N_k(w) & =  \OO(1), \qquad \textrm{as } w \to \infty,
    \\ \label{eq:hatNk3}
    N_k(w) & = \OO(w^{-1}), \qquad \textrm{as } w \to 0,
    \\ \label{eq:hatNk4}
    N_k(w) &= \OO((w \mp s)^{-1/2}), \quad \textrm{as } w \to \pm s, \\
    \label{eq:hatNk5}
    N_k(w) & = \OO((w \mp {\rm i}t)^{-1/2}), \quad \textrm{as } w \to  \pm {\rm i}t.
    \end{align}

Then the $4\times 4$ matrix valued function $z \mapsto
\widehat{M}(z)$ with entries
    \begin{align} \label{eq:hatM}
    \widehat{M}_{kj}(z) = N_k(w_j(z)), \qquad k,j=1, \ldots, 4,
    \end{align}
where $w_j$, $j=1,\ldots, 4$, are the mapping functions for
\eqref{algebraicequ}, is analytic in $\C \setminus (\R \cup
S(\sigma-\mu_2))$, satisfies the jump conditions in the RH problem
for $M$ (due to the property \eqref{eq:hatNk1} of $N_k$), as well
as the fourth root condition \eqref{eq:assumptionfourthroot} (due
to \eqref{eq:hatNk4}--\eqref{eq:hatNk5}). In addition we have as
$z \to \infty$,
 \begin{align} \label{eq:hatM1}
 \widehat{M}_{kj}(z) = \begin{cases}
    \OO(1), & \textrm{ for } j =1, \\
    \OO(z^{1/3}), & \textrm{ for } j=2,3,4,
    \end{cases}
    \end{align}
which follows from \eqref{eq:hatNk2}--\eqref{eq:hatNk3} and
\eqref{eq:hatM}.

Since the jump matrices have determinant one, it follows by
standard arguments that $z \to \det \widehat{M}(z)$ extends to an
entire function. From \eqref{eq:hatM1} it follows that $\det
\widehat{M}(z) = \OO(z)$ as $z \to \infty$, so that
 \begin{align} \label{eq:dethatM}
 \det \widehat{M}(z) = \alpha z + \beta,
 \end{align}
for some constants $\alpha$ and $\beta$. Since the polynomials
$Q_k$, $k=1,\ldots, 4$ are linearly independent, $\alpha$ and
$\beta$ cannot both be zero.

Define
    \begin{align} \label{eq:defA}
    A(z)=
      \begin{pmatrix}
                1& 0 & 0 &0\\
                0 & z^{1/3} & 0 & 0\\
                0 & 0 & 1 & 0 \\
                0 & 0 & 0 & z^{-1/3}
        \end{pmatrix}
            \begin{pmatrix}
                1 & 0 \\
                0 & A_j
            \end{pmatrix}
    \end{align}
for $z$ in the $j$th quadrant. Then $A$ is analytic in $\C
\setminus\big(\R\cup {\rm i}\R\big)$ and from the formulas
\eqref{eq:A1A2} and \eqref{eq:A3A4} for the matrices $A_j$ it
easily follows that on  $\R\cup{\rm i}\R$ it satisfies the jump
conditions
    \begin{align}
        A_+(z)&=A_-(z)
            \begin{pmatrix}
                1 & 0 & 0 &0\\
                0 & 1 & 0 &0\\
                0 & 0 & 0 &1\\
                0 & 0 & -1&0
            \end{pmatrix}, \qquad z\in \R, \\
        A_+(z)&=A_-(z)
            \begin{pmatrix}
            1 & 0 & 0 & 0 \\
            0 & 0 & -1 & 0 \\
            0 & 1 & 0 & 0\\
            0 & 0 & 0 & 1
            \end{pmatrix}, \qquad z\in {\rm i}\R.
    \end{align}
Comparing this with the jumps in the RH problem for $M$, we see
that $A$ and $\widehat{M}$ satisfy the same jump conditions for
$z$ large enough. Then $\widehat{M} A^{-1}$ is analytic in a
neighborhood of infinity.

Due to \eqref{eq:hatM1} and \eqref{eq:defA}, we have
$\widehat{M}(z) A^{-1}(z) = \OO(z^{2/3})$ as $z \to \infty$.
Therefore, by analyticity at infinity, we have
 \begin{align}
 \widehat{M}(z) A^{-1}(z) = C + \OO(z^{-1}), \qquad z \to \infty,
 \end{align}
for a constant matrix $C$. Then $\det \widehat{M}(z) = \det C +
\OO(z^{-1})$, so that by \eqref{eq:dethatM}, $\alpha = 0$ and
$\beta = \det C$. Since $\alpha$ and $\beta$ are not both zero, we
find $\det C \neq 0$. Thus $C^{-1}$ exists, and then it easily
follows that
\begin{equation}
    M = C^{-1} \widehat{M}
    \end{equation}
    satisfies all conditions in the RH problem for $M$, as well as
    the fourth root condition \eqref{eq:assumptionfourthroot}.

This completes the construction of a solution of the  RH problem for $M$
$M$ in the one-cut case.
\end{proof}

\subsection{Construction of local parametrices}

The next step is the construction of local parametrices near the
branch points. Since we are in the one-cut regular case, the
density of $\mu_1$ vanishes as a square root at the endpoints $\pm
a$. As in the case of orthogonal polynomials
\cite{DKMVZuniform,DKMVZstrong} the local parametrix will then be
constructed with the help of Airy functions. Also for larger size
RH problems Airy parametrices have been constructed, see e.g.\
\cite{ABK,BK1,DKV,KVW}. The situation in the present case is
similar, and so we will not give all details of the construction
here.

\subsubsection{The model RH problem: Airy functions}

Airy functions solve a model RH problem. Let $y_0$, $y_1$ and
$y_2$ be defined by
\begin{align}
y_0(s)=\Ai(s), \quad y_1(s)=\omega \Ai(\omega s), \quad y_2(s)=\omega^2 \Ai(\omega^2 s).
\end{align}
where $\Ai$ is the Airy function and $\omega={\rm e}^{2\pi{\rm
i}/3}$. Consider the $2\times 2$ matrix valued function $\Psi$
\begin{align}\label{eq:Airy1}
\Psi(s)&=\begin{pmatrix}
y_0(s) & -y_2(s)\\
y_0'(s) & -y_2'(s)
\end{pmatrix}, \qquad \arg s\in (0,2\pi/3),\\
\Psi(s)&=\begin{pmatrix}
-y_1(s) & -y_2(s)\\
-y_1'(s) & -y_2'(s)
\end{pmatrix}, \qquad  \arg s\in (2\pi/3,\pi),\label{eq:Airy2}\\
\Psi(s)&=\begin{pmatrix}
-y_2(s) & y_1(s)\\
-y_2'(s) & y_1'(s)
\end{pmatrix}, \qquad \arg s\in (-\pi,-2\pi/3),\\
\Psi(s)&=\begin{pmatrix}
y_0(s) & y_1(s)\\
y_0'(s) & y_1'(s)
\end{pmatrix}, \qquad  \arg s\in (-2\pi/3,0),\label{eq:Airy4}
\end{align}
Then $\Psi$ is analytic in the complex $s$-plane with a jump
discontinuity along the rays $\arg s=0$, $\arg=\pm 2\pi/3$ and
$\arg s=\pi$. We equip these rays with an orientation as shown in
Figure \ref{fig:airy}. This figure also shows the jump matrices,
which can be easily obtained from the definition
\eqref{eq:Airy1}--\eqref{eq:Airy4} and the linear relation $y_0 +
y_1 + y_2 =0$.

\begin{figure}[t]
\centering
    \begin{overpic}[width=9cm,height=7cm]{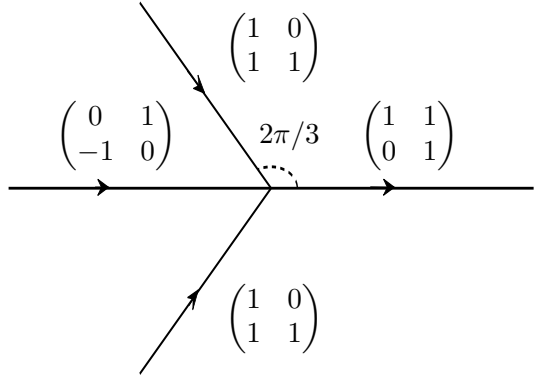}
    \put(65,47){\mbox{$\begin{pmatrix}
     1& 1\\
     0& 1
    \end{pmatrix}$}}
    \put(20,47){\mbox{$\begin{pmatrix}
     0& 1\\
     -1& 0
    \end{pmatrix}$}}
    \put(45,20){\mbox{$\begin{pmatrix}
     1& 0\\
     1& 1
    \end{pmatrix}$}}
    \put(50,47){\mbox{$2\pi/3$}}
        \put(45,60){\mbox{$\begin{pmatrix}
     1& 0\\
     1& 1
    \end{pmatrix}$}}
    \end{overpic}
    \caption{The contour and the jump matrices in the RH problem for $\Psi$}
    \label{fig:airy}
\end{figure}

The asymptotic behavior of $\Psi(s)$ is given by
        \begin{align}\label{eq:asymptAiryparam}
            \Psi(s)=\frac{1}{2\sqrt \pi}
                    \begin{pmatrix}
                        s^{-1/4} & 0 \\
                        0 & s^{1/4}
                    \end{pmatrix}
                    \begin{pmatrix}
                        1 & {\rm i}\\
                        -1 & {\rm i}
                    \end{pmatrix}(I+\OO(s^{-2/3}))
            \begin{pmatrix}
{\rm e}^{-\frac{2}{3} s^{3/2}} & 0 \\
0 & {\rm e}^{\frac{2}{3} s^{3/2}}
\end{pmatrix},
\end{align}
as $s\to \infty$.

\subsubsection{The local parametrix near $\pm a$}

The local parametrix $P_a$ is constructed in a small disk $U_a$
around $a$. It should satisfy the jump conditions in the RH
problem for $S$ exactly and match with $M$ on the boundary
$\partial U_a$.

Recall the function $\phi_1$ from \eqref{eq:defphi1}, which
behaves near $z = a$ as
    \begin{align}\label{eq:orderphi1}
            -\frac{3}{2}\phi_{1}(z)&=\rho(z-a)^{3/2}\left(1+\OO(z-a)\right), \qquad z\to
            a,
    \end{align}
where $\rho > 0$ is such that
    \begin{align}
     \frac{{\rm d}\mu_1}{{\rm d}x}(x)=\frac{ \rho}{\pi} (a-x)^{1/2}(1+\OO((a-x)),
     \qquad  x \nearrow a.
    \end{align}
Then
    \begin{align}\label{eq:deff1}
        f_1(z)=\left(-\frac{3}{2}\phi_1(z)\right)^{2/3},
        \qquad z \in U_a,
    \end{align}
is a conformal map from $U_a$ onto a neighborhood of the origin,
so that $f_1(z) > 0$ for $z \in U_a\cap[a,\infty)$. We use the
freedom we have in opening the lens around $[-a,a]$ so that $f_1$
maps the part of $\Sigma_1 \cap U_a$ in the upper half-plane into
into the ray $\arg s=2\pi/3$ and the part in the lower half-plane
into the the ray $\arg s=-2\pi/3$.

Now define $P_a$ by
    \begin{align}\label{eq:defPa}
        P_a(z)=E_a(z) \begin{pmatrix}\Psi(n^{2/3} f_1(z))  & 0   \\
            0 & I_2 \end{pmatrix} \begin{pmatrix}
            {\rm e}^{-n\phi_1(z) \sigma_3} & 0 \\
            0 & I_2
        \end{pmatrix},
    \end{align}
for $z\in U_a$, where $\sigma_3 = \left(\begin{smallmatrix} 1 & 0
\\ 0 & - 1 \end{smallmatrix}\right)$ and
$E_a$ is given by
    \begin{align}\label{eq:E1upperhalf}
        E_a(z)=\sqrt{\pi} M(z)
            \begin{pmatrix}
                1  & -1 & 0  \\
            -{\rm i}& -{\rm i} &   0 \\
                0 & 0  & I_2
            \end{pmatrix}
        \begin{pmatrix} \big( n^{2/3} f_1(z)\big)^{\sigma_3/4} & 0 \\
            0 & I_2
        \end{pmatrix}.
    \end{align}
for $z\in U_a$.

In  the following lemma we prove that $P_a$ exactly solves the
jump condition in the RH problem for $S$ in $U_a$ and matches with
the outside parametrix $M$ on $\partial U_a$.
\begin{lemma} \label{lem:parama}
    The function $P_a$ satisfies
        \begin{align}\label{eq:paramplusa}
            \left\{
                \begin{array}{ll}
                \multicolumn{2}{l}{P_a \textrm{ is analytic  in } U_a\setminus J_S,}\\
                P_{a_+}(z)=P_{a_-}(z)J_S(z), & z\in U_a\cap J_S, \\
                P_a(z)=M(z)(1+\OO(1/n)), & \textrm{as } n \to \infty,
                \textrm{ uniformly for } z\in \partial U_a.
            \end{array}
            \right.
        \end{align}
\end{lemma}
\begin{proof}
The proof is standard.
\end{proof}

The construction of the local parametrix $P_{-a}$ around $-a$ can
be done similarly. We can also use the symmetry to define $P_{-a}$
directly in terms of $P_a$ as follows
\begin{align}
P_{-a}(z)=\begin{pmatrix}
1 & 0 & 0 & 0 \\
0 & -1 & 0 & 0\\
0 & 0 & 1 & 0\\
0 & 0 & 0 & -1
\end{pmatrix}P_a(-z)
\begin{pmatrix}
 1& 0 & 0 &0 \\
    0 & -1 & 0 & 0\\
        0 & 0 & 1 & 0\\
     0 & 0 & 0 & -1
    \end{pmatrix}
\end{align}

\subsubsection{The local parametrix near $\pm {\rm i} c$}

The local parametrices near the branch cuts $\pm {\rm i} c$ are
constructed in essentially the same way. We will focus on $-{\rm
i}c$. The jump matrices in the RH problem for $S$ that are
relevant near $-{\rm i}c$ \eqref{eq:JT4}, \eqref{eq:JT5} and
\eqref{eq:JT6}. They are non-trivial only in the $2\times 2$
middle block.

We again construct the local parametrix by means of the Airy model
RH problem and a conformal map that maps a disk $U_{-{\rm i}c}$
around $-{\rm i}c$ onto a neighborhood of the origin. There is a
small difference in the fact that the $2\times 2$ middle block in
\eqref{eq:JT5} is upper triangular, whereas the jump matrix for
$\Psi$ on the rays $\arg s =\pm 2\pi/3$ is lower triangular. We
can deal with the different triangularity structure by using
    \begin{align} \label{eq:AiryPhi}
\Phi=\Psi\begin{pmatrix}
 0 & 1 \\
 1 & 0
\end{pmatrix}
\end{align}
instead of $\Psi$ in the construction of $P_{-{\rm i}c}$.

The conformal map $f_2$ is constructed out of $\phi_2$, which
behaves like
    \begin{align}
            \frac{3}{2}\phi_{2}(z)= \widetilde{\rho}
            (z+{\rm i}c)^{3/2}{\rm e}^{-3\pi{\rm i}/4}(1+\OO(z+{\rm i}c)),
    \end{align}
    as $z \to -{\rm i}c$. The fracional power is defined here with
    a branch cut along $(-{\rm i}\infty, -{\rm i}c]$ and so that
    $\phi_2$ takes positive values on $(-{\rm i}c,0)$. The number
    $\widetilde \rho > 0$ is such that (see also part (b) of Theorem
    \eqref{prop:structureEqmeasure})
    \begin{align}
            \frac{{\rm d} (\sigma-\mu_2)}{|{\rm d}z|}({\rm i}y)=
            \frac{\widetilde\rho}{\pi} |y+c|^{1/2}(1+\OO(y+c)),
            \qquad y\nearrow -c.
    \end{align}

Then $f_2$ is defined on a small enough disk $U_{-{\rm i}c}$
around $-{\rm i}c$ by
\begin{align}\label{eq:deff2}
    f_2(z)=\left(\frac{3}{2}\phi_2(z)\right)^{2/3},
\end{align}
with the $2/3$-root taken so that $f_2(z) > 0$ for $z \in (-{\rm
i}c,0)\cap U_{-{\rm i}c}$. Then $f_2$ is a conformal map from
$U_{-{\rm i}c}$ onto a neighborhood of zero. We adjust the
definition of the lens around $(-{\rm i} \infty, -{\rm i}c]$ so
that the part of $\Sigma_{2} \cap U_{-{\rm i}c}$ in the left
half-plane is mapped into the ray $\arg s=2\pi/3$ and the part in
the right half-plane into the ray $\arg s=-2\pi/3$.

Now we define $P_{-{\rm i}c} $ by
\begin{align}\label{eq:defPc}
    P_{-{\rm i}c}(z)=E_{-{\rm i}c}(z)
        \begin{pmatrix}1& 0  & 0\\
            0 & \Phi(n^{2/3} f_2(z))&  0 \\
            0 & 0 & 1 & \end{pmatrix}
    \begin{pmatrix}
        1 & 0  & 0& \\
        0 & {\rm e}^{-n\phi_2(z)\sigma_3} & 0  \\
        0 & 0 & 1
    \end{pmatrix}
\end{align}
with $\Phi$ as in \eqref{eq:AiryPhi}, and $E_{-{\rm i}c}$  given
by
\begin{multline}\label{eq:defE2}
E_{-{\rm i}c}(z)=\sqrt{\pi} M(z) \begin{pmatrix}
1& 0 & 0 &0 \\
0 &-{\rm i}  & -{\rm i} & 0 \\
0&1& -1  & 0  \\
 0 & 0 & 0 & 1
\end{pmatrix} \begin{pmatrix}
 1 & 0 & 0\\
 0 & \big( n^{2/3} f_2(z)\big)^{-\sigma_3/4} & 0  \\
 0 & 0 & 1
\end{pmatrix}
\end{multline}
The fractional power $\big(f_2(z))^{-1/4}$ is defined with a
branch cut along $(-{\rm i}\infty,-{\rm i}c]$ and so that it is
real and positive for $z \in (-{\rm i}c,0)\cap U_{-{\rm i}c}$.

Then we have the following result, whose proof is again omitted.
\begin{lemma} \label{lem:paramminic}
The function $P_{-{\rm i}c}$ satisfies
\begin{align}\label{eq:paramminic}
    \left\{
    \begin{array}{ll}
    \multicolumn{2}{l}{P_{-{\rm i}c} \textrm{ is analytic  in } U_{-{\rm i}c}\setminus J_S,}\\
    P_{-{\rm i}c,+}(z)=P_{-{\rm i}c,-}(z)J_S(z), & z\in U_{-{\rm i}c}\cap J_S, \\
    P_{-{\rm i}c}(z)=M(z)(1+\OO(1/n)), & \textrm{as } n \to \infty
    \textrm{ uniformly for } z\in \partial U_{-{\rm i}c}.
    \end{array}
\right.
\end{align}
\end{lemma}

The local parametrix $P_{{\rm i}c}$ around ${\rm i}c$ can be
constructed in a similar way. We can also use the symmetry in the
problem and define $P_{ {\rm i}c}$ in terms of $P_{-{\rm i}c}$ as
follows
\begin{align}
P_{{\rm i}c}(z)=\begin{pmatrix}
1 & 0 & 0 & 0 \\
0 & -1 & 0 & 0\\
0 & 0 & 1 & 0\\
0 & 0 & 0 & -1
\end{pmatrix}P_{{-\rm i}c}(-z)
\begin{pmatrix}
 1& 0 & 0 &0 \\
 0 & -1 & 0 & 0\\
 0 & 0 & 1 & 0\\
 0 & 0 & 0 & -1
\end{pmatrix}.
\end{align}

\subsection{The final transformation $S \mapsto R$}

Having constructed the outside parametrix $M$ and the local
parametrices $P_{\pm a}$ and $P_{-{\rm i}c}$, we define
\begin{align}\label{eq:defP}
P(z)=\begin{cases}
M(z),& z\in \C\setminus \big(\overline{U}_{\pm a}\cup \overline{U}_{\pm {\rm i} c} \cup \Sigma_S\big),\\
P_{\pm a}(z), & z\in U_{\pm a} \setminus \Sigma_S,\\
P_{\pm {\rm i} c}(z), & z\in U_{\pm {\rm i}c} \setminus \Sigma_S.
\end{cases}
\end{align}
Recall that the $3\times 3$ matrix-valued function
\eqref{eq:PinIandII}--\eqref{eq:PinIIIandIV} constructed out of
the Pearcey integrals was also denoted by $P$. Since this function
wlll not play a role anymore, we trust that the double use of the
symbol $P$ will not lead to any confusion. From now on $P$ will
always refer to the function defined in \eqref{eq:defP}.

Define the final transformation $S \mapsto R$ by
\begin{align} \label{eq:defR}
    R = S P^{-1}.
\end{align}
Then $R$ is defined and analytic in $\C \setminus (\Sigma_S \cup
\partial U_{\pm a} \cup \partial U_{\pm {\rm i} c})$, and has an
analytic continuation to $\C \setminus \Sigma_R$, with
    \begin{multline}
        \Sigma_R=\left(\partial U_{\pm a}\cup\partial U_{\pm {\rm i}c}
        \cup \Sigma_1 \cup \Sigma_2 \cup \Sigma_3\cup (-{\rm i}c,{\rm i}c) \cup\R \right) \\
        \setminus \left([-a,a] \cup U_{\pm a}\cup  U_{\pm {\rm i}c}\right),
    \end{multline}
see also Figure \ref{fig:RHforR}. The circles $\partial U_{\pm a}$
and $U_{{\pm i}c}$ are oriented in the counterclockwise direction.

\begin{figure}[t]
\centering
\begin{overpic}[width=10cm,height=8cm]{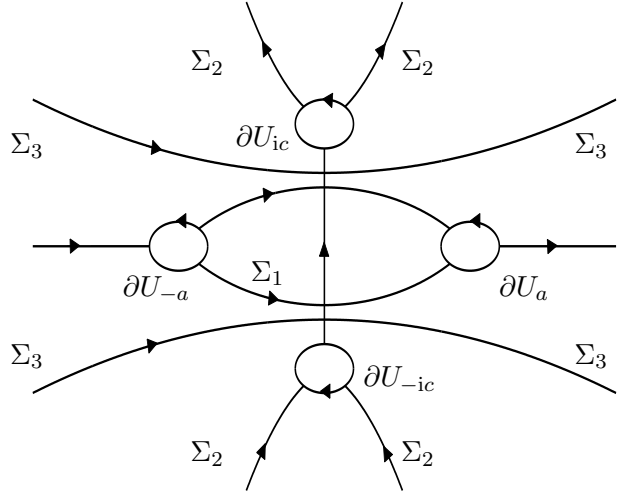}
\put(62,65){\mbox{$\Sigma_2$}}
\put(34,65){\mbox{$\Sigma_2$}}
\put(62,13){\mbox{$\Sigma_2$}}
\put(34,13){\mbox{$\Sigma_{2}$}}
\put(85,26){\mbox{$\Sigma_{3}$}}
\put(85,54){\mbox{$\Sigma_{3}$}}
\put(10,26){\mbox{$\Sigma_{3}$}}
\put(10,54){\mbox{$\Sigma_{3}$}}
\put(42,37){\mbox{$\Sigma_{1}$}}
\put(25,35){\mbox{$\partial U_{-a}$}}
\put(75,35){\mbox{$\partial U_a$}}
\put(40,55){\mbox{$\partial U_{{\rm i}c}$}}
\put(57,23){\mbox{$\partial U_{-{\rm i}c}$}}
\end{overpic}
\caption{The jump contour $\Sigma_R$ for $R$. Only the jumps on
$\partial U_{\pm a}$ and $\partial U_{\pm {\rm i}c}$ contribute to
the asymptotic expansion for $R$. The jump matrices on the other
contours are exponentially close to the identity matrix as $n \to
\infty$.} \label{fig:RHforR}
\end{figure}

We obtain the following RH problem for $R$.

\begin{lemma}
We have that $R$ satisfies the RH problem
\begin{align}
\left\{
\begin{array}{ll}
\multicolumn{2}{l}{R \textrm{ is analytic in } \C\setminus \Sigma_R,} \\
R_+(z)=R_-(z) J_R(z), & z\in \Sigma_R, \\
R(z)=I+\OO(1/z), & \textrm{uniformly as } z\to \infty,
\end{array}
\right.
\end{align}
where the jump matrix $J_R$ is given on the various parts of
$\Sigma_R$ as follows
\begin{align}
J_R(z)& = M(z)P_{\pm a}(z)^{-1}, \quad z\in \partial U_{\pm a}, \label{eq:JRUa}\\
J_R(z)& =M(z) P_{\pm {\rm i}c}(z)^{-1}, \quad z\in \partial U_{\pm
{\rm i}c}, \label{eq:JRUc}
\end{align}
\begin{align}
J_R(z)&=M_-(z)
    \begin{pmatrix}
    1 & {\rm e}^{2n\phi_1(z)} & 0 & 0 \\
    0 & 1 & 0 & 0 \\
    0 & 0 & 0 & 1\\
    0 & 0 & -1 &0
    \end{pmatrix}M_+(z)^{-1}, \quad z \in \R\setminus \big([-a,a]\cup U_{\pm a}\big)
    \end{align}
    \begin{align}
J_R(z)&=M(z)\begin{pmatrix}
 1 & 0 & 0 & 0\\
 {\rm e}^{-2n\phi_1(z)} & 1 & 0 & 0\\
 0 & 0 & 1  & 0\\
 0 & 0 & 0 & 1
\end{pmatrix}
M(z)^{-1}, \quad z\in \Sigma_{1}\setminus U_{\pm a},\\
J_R(z)&=M(z)\begin{pmatrix}
 1 & 0 & 0 & 0\\
 0 & 1 & {\rm e}^{2n\phi_2(z)} & 0\\
 0 & 0 & 1  & 0\\
 0 & 0 & 0 & 1
\end{pmatrix}
M(z)^{-1}, \quad z\in \Sigma_{2}\setminus U_{\pm {\rm i}c}, \\
J_R(z)&=M(z)\begin{pmatrix}
 1 & 0 & 0 & 0\\
 0 & 1 & 0 & 0\\
 0 & 0 & 1  & 0\\
 0 & 0 & {\rm e}^{2n\phi_3(z)} & 1
\end{pmatrix}
M(z)^{-1}, \quad z\in \Sigma_{3},
\end{align}
and on the various pieces of the segment  $(-{\rm i}c,{\rm i}c)$
we have
\begin{align}
J_R(z)&=M(z)\begin{pmatrix}
1 & 0 & 0 & 0\\
0 & 1 & 0 & 0\\
0 & {\rm e}^{-2n\phi_{2}(z)} & 1 & 0  \\
0 & 0 & 0 & 1
\end{pmatrix}M(z)^{-1},
\end{align}
for $ z\in (-{\rm i}c,{\rm i}c)\setminus \left(\Omega_3\cup U_{\pm {\rm i}c}\right)$,
\begin{align}
J_{R}(z)&=M(z)\begin{pmatrix}
1 & 0 & 0 & 0 \\
0 & 1 & 0 & 0 \\
0 & {\rm e}^{-2n\phi_{2}(z)} & 1 & 0\\
0 & {\rm e}^{-2n(\phi_{2}(z)+\phi_3(z))} & 0 & 1\\
\end{pmatrix}M(z)^{-1},
\end{align}
for $ z\in (-{\rm i} c,{\rm i}c)\cap
\left(\Omega_3\setminus\Omega_1\right)$,

\begin{align}
J_R(z)&=M(z)\begin{pmatrix}
1 & 0 & 0 & 0 \\
0 & 1 & 0 & 0 \\
-{\rm e}^{-2n(\phi_1(z)+\phi_{2}(z))}  & {\rm e}^{-2n\phi_{2}(z)} & 1 & 0\\
-{\rm e}^{-2n(\phi_1(z)+\phi_{2}(z)+\phi_3(z))}  & {\rm e}^{-2n(\phi_{2}(z)+\phi_3(z))} & 0 & 1\\
\end{pmatrix}M(z)^{-1},
\end{align}
for $z\in (-{\rm i}c,{\rm i}c)\cap \Omega_1$.
\end{lemma}

\begin{proof}
The formulas for $J_R$ follow by straightforward calculations.

From the asymptotic condition in the RH problems for $S$ and $M$
in \eqref{eq:RHproblemS} and \eqref{eq:RHforM}, respectively, we
obtain
\[ R(z) = S(z) M(z)^{-1} = I + \OO(z^{-1/3}) \qquad \textrm{as } z
\to \infty. \] Since the jump matrices $J_R(z)$ are exponentially
close to the identity matrix as $z \to \infty$ (see
\eqref{eq:JRsmall1}, \eqref{eq:JRsmall2}, \eqref{eq:JRsmall3}
below) the improved error term $R(z) = I + \OO(z^{-1})$ will
follow.
\end{proof}

From Lemmas  \ref{lem:Rephi2}, \ref{lem:Rephi3}, \ref{lem:Rephi1},
and the matching conditions for the local parametrices $P_{\pm a}$
and $P_{\pm {\rm i}c}$ it follows that all jump matrices are close
to the identity matrix if $n$ is large.

In fact, by \eqref{eq:JRUa} and \eqref{eq:JRUc} and the matching
conditions in \eqref{eq:paramplusa} and \eqref{eq:paramminic} it
follows that there exists a constant $C_0 > 0$ such that
\begin{align}
    \|J_R(z)-I\|\leq \frac{C_0}{n}, \qquad z\in \partial U_{\pm {\rm
    i}c}\cup \partial U_{\pm a}.
\end{align}
Here we can use any matrix norm $\| \cdot \|$. The other jump
matrices are exponentially close to the identity matrix as $n \to
\infty$. By Lemma \ref{lem:Rephi1} we have that
\begin{align} \label{eq:JRsmall1}
    \|J_R(z)-I\|\leq C_1 \exp(-n \eps_1 (|z|+1)), \quad z\in \R
    \setminus \left(S(\mu_1)\cup U_{\pm a}\right)
\end{align}
and
\begin{align} \label{eq:JRsmall1b}
\|J_R(z)-I\|\leq C_1 \exp(-n\eps_1),\quad z\in \Sigma_1\setminus
U_{\pm a}
\end{align}
for some constants $C_1 > 0$ and $\eps_1>0$. By Lemma
\ref{lem:Rephi2} we have
\begin{align} \label{eq:JRsmall2}
    \|J_R(z)-I\|\leq C_2 \exp(-n \eps_2 (|z|+1)), \qquad z\in \Sigma_2 \setminus U_{\pm{\rm i}c}
\end{align}
for some constants $C_2 > 0$ and $\eps_2>0$. From the fact that
$\Re \phi_1(z)>0$ for $z\in \Omega_1$, $\Re \phi_3>0$ for  $z\in
\Omega_3$ and by Lemma \ref{lem:Rephi2} we have
\begin{align} \label{eq:JRsmall2b}
 \|J_R(z)-I\|\leq C_2 \exp(-n\eps_2), \qquad z\in (-{\rm i}c,{\rm i}c)\setminus U_{\pm {\rm i}c}
\end{align}
where we can use the same constants $C_2$ and $\eps_2 >0$.
Finally, by Lemma \ref{lem:Rephi3} we have
\begin{align} \label{eq:JRsmall3}
\|J_R(z)-I\|\leq C_3 \exp(-n\eps_3|z|^{4/3}), \quad z\in \Sigma_3
\end{align}
for some constants $C_3 > 0$ and $\eps_3 > 0$.

This leads to the following result.
\begin{proposition}
There exists a constant $C>0$ such that for large enough $n$,
\begin{align}
    \|R(z)-I\|\leq \frac{C}{n(|z|+1)}, \qquad z \in \C \setminus
    \Gamma_R.
\end{align}
\end{proposition}
\begin{proof}
 The proposition follows from the above estimates on $\| J_R(z) - I\|$
 from the arguments as used in  \cite[Theorem 7.10]{DKMVZstrong}.
\end{proof}

This concludes the steepest descent analysis of the RH problem for $Y$.

\section{Proofs of Theorems \ref{th:eigenvaluedistribution}  and \ref{th:universality}}

In this section we prove our main results Theorems
\ref{th:eigenvaluedistribution} and \ref{th:universality}. We
assume that $\mu_1$ is one-cut regular.

\subsection{The kernel  $K^{(n)}_{11}$}

The kernel $K^{(n)}_{11}$ is expressed in terms of $Y$ by formula
\eqref{eq:kernelintermsofY}. By following the sequence of
transformations $Y\mapsto X \mapsto U \mapsto T \mapsto S\mapsto
R$ we obtain expressions for $K^{(n)}_{11}$ in terms of the
solutions of the other RH problems.

\subsubsection*{First transformation $ Y \mapsto X$.} From the
transformation $Y\mapsto X$ \eqref{eq:YtohatX}--\eqref{eq:YtoX},
it follows that the right-hand side of \eqref{eq:kernelintermsofY}
can be written as
    \begin{multline}\label{eq:kernelreversestep1}
    K^{(n)}_{11}(x,y)=\frac{1}{2\pi{\rm i}(x-y)}
          \begin{pmatrix}
    0 & w_{0,n}(y) &  w_{1,n}(y) & w_{2,n}(y)
    \end{pmatrix} \\
    \times
    \begin{pmatrix}
        1 & 0 \\
        0 & D_n Q_+(n^{3/4}\tau y) \Theta_{1,+}(n^{3/4}\tau y)
        \end{pmatrix} X_-(y)^{-1}X_+(x) \\
        \times
        \begin{pmatrix}
        1 & 0 \\
        0 & \Theta_{1,+}(n^{3/4}\tau x)^{-1} Q_+(n^{3/4}\tau x)^{-1} D_n^{-1}
        \end{pmatrix}
        \begin{pmatrix}
        1\\
        0\\
        0\\
        0
        \end{pmatrix},
\end{multline}
for $x>0$ and $y>0$. By arguments as in
\eqref{eq:mainideatrickBHIpre} and \eqref{eq:Theta1Theta2} we can
rewrite \eqref{eq:kernelreversestep1} as
    \begin{multline}\label{eq:kernelreversestep1a}
    K^{(n)}_{11}(x,y)
        = \frac{1}{2\pi{\rm i}(x-y)}\begin{pmatrix}
0 & {\rm e}^{-n \left(V(y)-\frac{3}{4}\tau^{4/3}|y|^{4/3}\right)}
&  0 & 0
\end{pmatrix} X_+^{-1}(y)   X_+(x)\begin{pmatrix}
        1\\
        0\\
        0\\
        0
        \end{pmatrix},
\end{multline}
for $x>0$ and $y>0$. Similar calculations show that
\eqref{eq:kernelreversestep1a} also  holds for general $x,y\in
\R$.

\subsubsection*{Second transformation $X \mapsto U$.} Applying
\eqref{eq:XtoU} and using \eqref{eq:defG} and \eqref{eq:defL}, we
next rewrite \eqref{eq:kernelreversestep1a} as
\begin{multline} \label{eq:kernelreversestep2a}
    K^{(n)}_{11}(x,y)  = \frac{{\rm e}^{ng_{1,+}(x)}}{2\pi{\rm i}(x-y)}
        \begin{pmatrix}
            0 & {\rm e}^{-n\left(V(y)-\frac{3}{4}\tau^{4/3}|y|^{4/3}- g_{1,+}(y)+g_2(y)-k_1\right)} &  0 & 0
        \end{pmatrix} \\
             \times U_+^{-1}(y) U_+(x)
                        \begin{pmatrix}
                        1\\
                        0\\
                        0\\
                        0
            \end{pmatrix}, \qquad x, y \in \R.
\end{multline}
Then by  \eqref{eq:variationalcondmu1complex} and
\eqref{eq:g1+g1-phi1} in case $y \in S(\mu_1)$, and by
\eqref{eq:variationalphi1} in case $y\in \R \setminus S(\mu_1)$,
we have
\begin{align*}
    V(y)- \frac{3}{4} \tau^{4/3} |y|^{4/3}-g_{1,+}(y)+g_2(y)-k_1
        \equiv -2\phi_{1,+}(y)+ g_{1,+}(y)
            \quad (\mod 2 \pi{\rm i}/3),
\end{align*}
Substituting this into \eqref{eq:kernelreversestep2a}, we arrive
at
\begin{multline} \label{eq:kernelreversestep2b}
    K^{(n)}_{11}(x,y) = \frac{{\rm e}^{n\left(g_{1,+}(x)-g_{1,+}(y)\right)}}{2\pi{\rm i}(x-y)}
        \begin{pmatrix}
            0 & {\rm e}^{2n\phi_{1,+}(y)} &  0 & 0
        \end{pmatrix}  \\
                \times U_+^{-1}(y)  U_+(x)
                        \begin{pmatrix}
                        1\\
                        0\\
                        0\\
                        0
            \end{pmatrix}, \quad x, y \in \R.
\end{multline}

\subsubsection*{Third transformation $U \mapsto T$.}
The transformation \eqref{eq:Tomegau3} acts only on the lower
right $2\times 2$ block. It does not affect the kernel
\eqref{eq:kernelreversestep2b}. So we have
\begin{multline} \label{eq:kernelreversestep2c}
    K^{(n)}_{11}(x,y) = \frac{{\rm e}^{n\left(g_{1,+}(x)-g_{1,+}(y)\right)}}{2\pi{\rm i}(x-y)}
        \begin{pmatrix}
            0 & {\rm e}^{2n\phi_{1,+}(y)} &  0 & 0
        \end{pmatrix}  \\
                \times T_+^{-1}(y) T_+(x)
                        \begin{pmatrix}
                        1\\
                        0\\
                        0\\
                        0
            \end{pmatrix}, \quad x, y \in \R.
\end{multline}

\subsubsection*{Fourth transformation $T \mapsto S$.}
Combining \eqref{eq:defS1} and \eqref{eq:defS3} with
\eqref{eq:kernelreversestep2c} we obtain
\begin{multline}\label{eq:kernelreversestep3a}
    K^{(n)}_{11}(x,y)=\frac{{\rm e}^{n\left(g_{1,+}(x)-g_{1,+}(y)\right)}}{2\pi{\rm i}(x-y)}
        \begin{pmatrix}
            -\chi_{[-a,a]}(y)  & {\rm e}^{2n\phi_{1,+}(y)} &  0 & 0
        \end{pmatrix}  \\
                \times S_+^{-1}(y)  S_+(x)
                        \begin{pmatrix}
                        1\\
                        \chi_{[-a,a]}(x){\rm e}^{2n\phi_{1,+}(x)}\\
                        0\\
                        0
            \end{pmatrix}, \quad x, y \in \R,
\end{multline}
where $\chi_{[-a,a]}$ is the characteristic function of the
interval $[-a,a]$.

\subsubsection*{Final transformation $S \mapsto R$.}
Since $S=RP$, we obtain
    \begin{multline}\label{eq:kernelreversestep4}
    K^{(n)}_{11}(x,y)=\frac{{\rm e}^{n\left(g_{1,+}(x)-g_{1,+}(y)\right)}}{2\pi{\rm i}(x-y)}
        \begin{pmatrix}
            -\chi_{[-a,a]}(y)  & {\rm e}^{2n\phi_{1,+}(y)} &  0 & 0
        \end{pmatrix}  \\
                \times P_+^{-1}(y) R_+^{-1}(y) R_+(x) P_+(x)
                        \begin{pmatrix}
                        1\\
                        \chi_{[-a,a]}(x){\rm e}^{2n\phi_{1,+}(x)}\\
                        0\\
                        0
            \end{pmatrix}
    \end{multline}
for $x,y\in \R$. This is the final formula.

\subsection{Proofs of Theorems \ref{th:eigenvaluedistribution} and \ref{th:universality}}

\begin{proof}
Recall that $R(z) = I + \OO(1/n)$ as $n \to \infty$ uniformly for
$z \in \C \setminus \Sigma_R$, which we use in
\eqref{eq:kernelreversestep4}. Theorems
\ref{th:eigenvaluedistribution} and \ref{th:universality} then
follow by calculations that are similar to the ones in e.g.\
\cite{BK1,DKV}.

Note that the factor ${\rm e}^{
n\left(g_{1,+}(x)-g_{1,+}(y)\right)}$ in
\eqref{eq:kernelreversestep4} disappears for $x=y$, and it also
drops out of the determinantal formulas in
 \eqref{eq:universalitysine} and
\eqref{eq:universalityairy}. \end{proof}

\subsection*{Acknowledgements}

We are grateful to Marco Bertola, John Harnad and Alexander Its for providing us with a
copy of  their unpublished manuscript \cite{BHI}.

The first author is a research assistant of the Fund for Scientific Research – Flanders.
The authors were supported by the European Science Foundation Program MISGAM.
The second author is supported by FWO-Flanders project G.0455.04, by K.U. Leuven
research grants OT/04/21 and OT/08/33, by Belgian Interuniversity Attraction Pole NOSY P06/02,
and by a grant from the Ministry of Education and Science of Spain, project code
MTM2005-08648-C02-01.

\bibliographystyle{amsplain}

\bigskip

\obeylines

Maurice Duits
Arno B. J. Kuijlaars
Department of Mathematics
Katholieke Universiteit Leuven
Celestijnenlaan 200B
3001 Leuven, BELGIUM

maurice.duits@wis.kuleuven.be
arno.kuijlaars@wis.kuleuven.be

\end{document}